\documentclass{CSML}
\usepackage{lastpage}
\def\dOi{13(2:15)2017}
\lmcsheading%
{\dOi}
{1--\pageref{LastPage}}
{}
{}
{Aug.~21, 2015}
{Jun.~30, 2017}
{}

\usepackage{hyperref}\hypersetup{hidelinks}
\usepackage{amsmath,amssymb,xcolor,mathtools,ifthen,flushend,bm,bbm,amsthm}
\usepackage{tikz}
\usetikzlibrary{calc,shapes,arrows,automata,positioning,snakes}
\tikzset{->,>=stealth'}
\usepgflibrary{shapes}
\usepackage[all]{xy}
\usepackage{stmaryrd}   
\usepackage{amssymb,amsmath,amsfonts,graphicx,xspace,subfig,multirow}
\usepackage{url}

\newtheorem{theorem}{Theorem}[section]
\newtheorem{lemma}[theorem]{Lemma}
\newtheorem{corollary}[theorem]{Corollary}

\newtheorem{conjecture}[theorem]{Conjecture}
\theoremstyle{definition}
\newtheorem{remark}[theorem]{Remark}
\newtheorem{example}[theorem]{Example}

\newcommand{\theoremlike}[2]{\par\medskip\penalty-250
{{
\textsc{
#2 \ref{#1}.}}}\it}

\newcommand{\thmhelperpre}[2]{\theoremlike{#1}{#2}}
\newcommand{\thmhelperpost}{\par\medskip}

\newcommand{\Nset}{\mathbb N}
\renewcommand{\vec}[1]{\bm{#1}}

\renewcommand{\bar}[1]{#1}

\newcommand{\vecl}[1]{\vec{#1}}

\newcommand{\dist}{\mathit{Dist}}
\newcommand{\poly}{\mathit{poly}}
\newcommand{\ub}{\textbf{U:\ }}
\newcommand{\lb}{\textbf{L:\ }}
\newcommand{\ubl}{\textbf{U}=\textbf{L:} }

\newcommand{\inits}{s_0}
\newcommand{\act}[1]{\mathit{Act}(#1)}
\newcommand{\pat}{\omega}
\newcommand{\Pat}{\mathsf{Runs}}
\newcommand{\fpat}{w}
\newcommand{\mem}{\mathsf{M}}
\newcommand{\Cone}{\mathsf{Cone}}
\newcommand{\calF}{\mathcal{F}}
\newcommand{\mec}{\mathsf{MEC}}

\newcommand{\Qset}{\mathbb{Q}}
\newcommand{\Rset}{\mathbb{R}}

\newcommand{\QED}{\qed}

\newcommand{\pagewidth}{\textwidth}
\newcommand{\textbp}[1]{\textbf{#1}}
\newcommand{\QEE}{\hfill$\triangle$} 

\newcommand{\pr}{\mathbb P}	
\renewcommand{\Pr}[3]{\pr^{#1}\hspace{-0.16em}\left[{#3}\right]}   
\newcommand{\PrS}[3]{\pr^{#1}_{#2}\hspace{-0.16em}\left[{#3}\right]}   
\newcommand{\expected}{\mathbb E}
\newcommand{\Ex}[3]{\expected^{#1}_{#2}\hspace{-0.16em}\left[{#3}\right]}   

\newcommand{\reward}{\vec{r}}
\newcommand{\ex}{\vecl{exp}}
\newcommand{\sat}{\vecl{sat}}
\newcommand{\satscalar}{\mathit{sat}}
\newcommand{\psat}{\vecl{pr}}
\newcommand{\psatscalar}{\mathit{pr}}

\newcommand{\lrLim}[1]{\mathrm{lr}(#1)}  
\newcommand{\lrSf}[1]{\mathrm{lr}_{\mathrm{sup}}(#1)}  
\newcommand{\lrIf}[1]{\mathrm{lr}_{\mathrm{inf}}(#1)}  

\newcommand{\rc}[1]{\textcolor{red}{#1}}
\renewcommand{\rc}[1]{#1}

\newcommand{\rstart}{}
\newcommand{\rstop}{\color{black}}

\newcommand{\kd}[1]{\mathbbm{1}_{#1}} 
\newcommand{\kda}{\kd{a}} 
\newcommand{\reach}{\Diamond}
\newcommand{\stabi}[1]{\kappa_{#1}}
\newcommand{\lfa}{\mathit{lf}_a}
\newcommand{\vfreq}{\vecl{Freq}}

\newcommand{\toss}{\mathit{toss}}
\newcommand{\leave}{\mathit{leave}}

\newcommand{\spacefu}{}
\newcommand{\spacefl}{}
\newcommand{\myspace}{}

\begin{document}

\renewcommand{\thefootnote}{(\fnsymbol{footnote})}

\allowdisplaybreaks

\title[Unifying Two Views on Multiple Mean-Payoff Objectives in MDPs]
{Unifying Two Views on Multiple Mean-Payoff Objectives in Markov Decision Processes\rsuper*}

\author[K.~Chatterjee]{Krishnendu Chatterjee\rsuper a}
\address{{\lsuper a}IST Austria}
\email{Krishnendu.Chatterjee@ist.ac.at}

\author[Z.~K\v{r}et\'insk\'a]{Zuzana K\v{r}et\'insk\'a\rsuper b}
\address{{\lsuper{b,c}}Institut f\"ur Informatik, Technische Universit\"at Mu\"nchen, Germany}
\email{komarkova.zuza@gmail.com, jan.kretinsky@gmail.com}

\author[J.~K\v{r}et\'insk\'y]{Jan K\v{r}et\'insk\'y\rsuper c}
\address{\vspace{-18 pt}}

\titlecomment{{\lsuper*}This is an extended version of the LICS'15 paper with full proofs and additional complexity results.}
\thanks{This research was funded in part by 
Austrian Science Fund Grant No P 23499-N23,
European Research Council Grant No 279307 (Graph Games), 
the DFG Research Training Group PUMA: Programm- und Modell-Analyse (GRK 1480),
the Czech Science Foundation grant No.~\mbox{15-17564S},
and People Programme (Marie Curie Actions) of the European Union's Seventh Framework Programme (FP7/2007-2013) 
REA Grant No 291734.
}

\begin{abstract}
  We consider Markov decision processes (MDPs) with multiple
  limit-average (or mean-payoff) objectives.  There exist two
  different views: (i)~the expectation semantics, where the goal is to
  optimize the expected mean-payoff objective, and (ii)~the
  satisfaction semantics, where the goal is to maximize the
  probability of runs such that the mean-payoff value stays above a
  given vector.  We consider optimization with respect to both
  objectives at once, thus unifying the existing semantics.
  Precisely, the goal is to optimize the expectation while ensuring
  the satisfaction constraint.  Our problem captures the notion of
  optimization with respect to strategies that are risk-averse (i.e.,
  ensure certain probabilistic guarantee).  Our main results are as
  follows: First, we present algorithms for the decision problems,
  which are always polynomial in the size of the MDP.  We also show
  that an approximation of the Pareto curve can be computed in time
  polynomial in the size of the MDP, and the approximation factor, but
  exponential in the number of dimensions.  Second, we present a
  complete characterization of the strategy complexity (in terms of
  memory bounds and randomization) required to solve our problem.
\end{abstract}

\maketitle



\section{Introduction}\label{chap:intro}

\smallskip\noindent{\bf MDPs and mean-payoff objectives.}
The standard models for dynamic stochastic systems with both 
nondeterministic and probabilistic behaviours are  
Markov decision processes (MDPs)~\cite{Howard,Puterman,FV97}.
An MDP consists of a finite state space, and in every state a 
controller can choose among several actions (the nondeterministic choices), 
and given the current state and the chosen action the system evolves 
stochastically according to a probabilistic transition function.
Every action in an MDP is associated with a reward (or cost), and the
basic problem is to obtain a strategy (or policy) that resolves the choice of 
actions in order to optimize the rewards obtained over the run of the system.
An objective is a function that given a~sequence of rewards over the run 
of the system combines them to a single value.
A classical and one of the most well-studied objectives in context of MDPs 
is the \emph{limit-average (or long-run average or mean-payoff)} objective that assigns to every
run the average of the rewards over the run.

\smallskip\noindent{\bf Single vs.\ multiple objectives.}
MDPs with single mean-payoff objectives have been widely studied 
(see, e.g.,~\cite{Puterman,FV97}), with many applications 
ranging from computational biology to analysis of security protocols, 
randomized algorithms, or robot planning, to name a 
few~\cite{BaierBook,PRISM,Bio-Book,KGFP09}.
In verification of probabilistic systems, MDPs are widely used, 
for concurrent probabilistic systems~\cite{CY95,Var85},  
probabilistic systems operating in open environments~\cite{SegalaT,dA97a}, 
and applied in diverse domains~\cite{BaierBook,PRISM}.
However, in several application domains, there is not a single optimization 
goal, but multiple, potentially dependent and conflicting goals.
For example, in designing a computer system, the goal is to maximize 
average performance while minimizing average power consumption, or
in an inventory management system, the goal is to optimize several 
potentially dependent costs for maintaining each kind of product.
These motivate the study of MDPs with multiple mean-payoff objectives, 
which has also been applied in several problems such as dynamic power
management~\cite{FKP12}.

\smallskip\noindent{\bf Two views.} 
There exist two views in the study of MDPs with mean-payoff objectives
\cite{krish}.
The traditional and classical view is the \emph{expectation} semantics, 
where the goal is to maximize (or minimize) the expectation 
of the mean-payoff objective. There are numerous applications of MDPs 
with the expectation semantics, such as in inventory control, planning, and 
performance evaluation~\cite{Puterman,FV97}. 
The alternative semantics is called the \emph{satisfaction} semantics, 
which, given  
a mean-payoff value threshold $\satscalar$ 
and
a probability threshold $\psatscalar$,
asks for a strategy to ensure that the mean-payoff value be at least $\satscalar$ 
with probability at least $\psatscalar$.
In the case with $n$ reward functions, there are two possible interpretations. 
Let $\sat$ and $\psat$ be two vectors of thresholds of dimension $k$, 
and $0\leq \psatscalar \leq 1$ be a single threshold.
The first interpretation (namely, the \emph{conjunctive interpretation}) requires 
the satisfaction semantics in each dimension $1\leq i \leq n$ with thresholds
$\sat_i$ and $\psat_i$, respectively (where $\vec v_i$ is the $i$-th 
component of vector $\vec v$). The sets of satisfying runs for each reward may even be disjoint here.
The second interpretation (namely, the \emph{joint interpretation}) requires the 
satisfaction semantics for all rewards at once. Precisely, it requires that, with probability at least $\psatscalar$, 
the mean-payoff value vector be at least $\sat$.
The distinction of the two views (expectation vs.\ satisfaction) and their 
applicability in analysis of problems related to stochastic reactive systems 
has been discussed in details in~\cite{krish}.
While the joint interpretation of satisfaction has already been introduced 
and studied in~\cite{krish}, 
here we consider also the conjunctive interpretation, which was not considered in~\cite{krish}.
The conjunctive interpretation was considered in~\cite{FKR95}, however, only a partial solution 
was provided, and it was mentioned that a complete solution would be very useful.

\smallskip\noindent{\bf Our problem.}
In this work we consider a new problem that unifies the two different 
semantics. 
Intuitively, the problem we consider asks to \emph{optimize} the expectation
while \emph{ensuring} the satisfaction. 
Formally, consider an MDP with $n$ reward functions, a probability 
threshold vector $\psat$ (or threshold $\psatscalar$ for joint interpretation), 
and a mean-payoff value threshold vector $\sat$.
We consider the set of \emph{satisfaction} strategies that ensure the 
satisfaction semantics. 
Then the optimization of the expectation is considered with respect to the 
satisfaction strategies. 
Note that if $\psat$ is~$\vec{0}$, then the satisfaction strategies is the set of 
all strategies and we obtain the traditional expectation semantics as a 
special case.
We also consider important special cases of our problem, depending on 
whether there is a single reward (mono-reward) or multiple rewards (multi-reward),
and whether the probability threshold is $\psat=\vec{1}$ (qualitative criteria) or 
the general case (quantitative criteria).
Specifically, we consider four cases: 
\begin{enumerate}
\item \emph{Mono-qual:} a single reward function and qualitative satisfaction semantics;
\item \emph{Mono-quant:} a single reward function and  quantitative satisfaction semantics;
\item \emph{Multi-qual:} multiple reward functions and qualitative satisfaction semantics;
\item \emph{Multi-quant:} multiple reward functions and quantitative satisfaction semantics.
\end{enumerate}
Note that for multi-qual and mono cases, the two interpretations (conjunctive and joint) of the 
satisfaction semantics coincide, whereas in the multi-quant problem (which is the most general problem)
we consider both the conjunctive and the joint interpretations, 
separately (\emph{multi-quant-conjunctive}, \emph{multi-quant-joint}) as well as at once (\emph{multi-quant-conjunctive-joint}).

\smallskip\noindent{\bf Motivation.} 
The motivation to study the problem we consider is twofold. 
Firstly, it presents a unifying approach that combines the two existing semantics
for MDPs.
Secondly and more importantly,
it allows us to consider the problem of optimization along with \emph{risk aversion}.
A risk-averse strategy must ensure certain probabilistic guarantee on the payoff 
function.
The notion of risk aversion is captured by the satisfaction semantics, and thus
the problem we consider captures the notion of optimization under risk-averse 
strategies that provide probabilistic guarantee. 
The notion of \emph{strong risk-aversion} where the probability is treated as an 
adversary is considered in~\cite{BFRR14}, whereas we consider probabilistic (both 
qualitative and quantitative) guarantee for risk aversion.
We now illustrate our problem with several examples.
  
\smallskip\noindent{\bf Illustrative examples:}
\begin{itemize}
\item For simple risk aversion, consider a single reward function modelling investment. Positive reward stands for profit, negative for loss. We aim at maximizing the expected long-run average while guaranteeing that it is non-negative with at least 95\%. This is an instance of \emph{mono-quant} with $\psatscalar=0.95,\satscalar=0$.
\item For more dimensions, consider the example \cite[Problems 6.1, 8.17]{Puterman}. A vendor assigns to each customer either a low or a high rank. Further, there is a decision the vendor makes each year either to invest money into sending a catalogue to the customer or not. Depending on the rank and on receiving a catalogue, the customer spends different amounts for vendor's products and the rank can change. The aim is to maximize the expected profit provided the catalogue is almost surely sent with frequency at most $f$. This is an instance of \emph{multi-qual}. Further, one can extend this example to only require that the catalogue frequency does not exceed $f$ with 95\% probability, but 5\% best customers may still receive catalogues very often (instance of \emph{multi-quant}). 
\item The following is again an instance of \emph{multi-quant}. A gratis service for downloading is offered as well as a premium one. For each we model the throughput as rewards $r_1,r_2$. For the gratis service, expected throughput $1\mathit{Mbps}$ is guaranteed as well as $60\%$ connections running on at least $0.8\mathit{Mbps}$. For the premium service, not only have we a higher expectation of $10\mathit{Mbps}$, but also $95\%$ of the connections are guaranteed to run on at least $5\mathit{Mbps}$ and $80\%$ on even $8Mbps$ (satisfaction constraints). In order to keep this guarantee, we may need to temporarily hire resources from a cloud, whose cost is modelled as a reward $r_3$. While satisfying the guarantee, we want to maximize the expectation of $p_2\cdot r_2-p_3\cdot r_3$ where $p_2$ is the price per $\mathit{Mb}$ at which the premium service is sold and $p_3$ is the price at which additional servers can be hired.
Note that since the percentages above are different, the constraints cannot be encoded 
using the joint interpretation, and conjunctive interpretation is necessary. 
\end{itemize}

\smallskip\noindent{\bf The basic computational questions.}
In MDPs with multiple mean-payoff objectives, different strategies may 
produce incomparable solutions. Thus, there is no ``best'' solution
in general. Informally, the set of \emph{achievable solutions} 
is the set of all vectors $\vec{v}$ such that there is a strategy that 
ensures the satisfaction semantics and that the expected mean-payoff value 
vector under the strategy is at least $\vec{v}$. 
The ``trade-offs'' among the goals represented by the individual
mean-payoff objectives are formally captured by the \emph{Pareto curve},
which consists of all maximal tuples (with respect to component-wise ordering) 
that are not strictly dominated by any achievable solution. 
Pareto optimality has been studied in cooperative game theory~\cite{Owen95} 
and in multi-criterion optimization and decision making in both economics 
and engineering~\cite{Koski,YC,Szymanek}.

We study the following fundamental questions related to the properties of
strategies and algorithmic aspects in MDPs:
\begin{itemize}
\item \emph{Algorithmic complexity:} What is the complexity of deciding whether 
a given vector represents an achievable solution, and if the answer is yes, then
compute a witness strategy?
\item \emph{Strategy complexity:} 
What type of strategies is sufficient (and necessary) for achievable solutions?
\item \emph{Pareto-curve computation:} 
Is it possible to compute an approximation of the Pareto curve?
\end{itemize}

\smallskip\noindent{\bf Our contributions.}
We provide comprehensive answers to the above questions.
The main highlights of our contributions are: 
\begin{itemize}
\item \emph{Algorithmic complexity.} 
We present algorithms for deciding whether a given vector is an 
achievable solution and constructing a witness strategy.
All our algorithms are polynomial in the size
of the MDP. 
Moreover, they are polynomial even in the number of dimensions, except for \emph{multi-quant} with conjunctive 
interpretation where it is exponential.

\item \emph{Strategy complexity.} 
It is known that for both expectation and satisfaction semantics with single reward,
deterministic memoryless\footnote{A strategy is memoryless if it is independent of the history, but depends only on the current state. A strategy that is not deterministic is called randomized.} strategies are sufficient~\cite{FV97,BBE10,krish}.
We show this carries over in the \emph{mono-qual} case only. 
In contrast, we show that for \emph{mono-quant} both randomization and memory is necessary.
For randomized strategies, they can be \emph{stochastic-update}, where the 
memory is updated probabilistically, or \emph{deterministic-update}, where the memory
update is deterministic.
We provide precise bounds on the memory size of stochastic-update strategies.
Further, we show that for both mono-quant and multi-qual, deterministic-update strategies require memory size that 
is dependent on the MDP.
Finally, we also show that deterministic-update strategies are sufficient even for \emph{multi-quant}, 
thus extending the results of \cite{krish}.

\item \emph{Pareto-curve computation.} We show that in all cases with 
multiple rewards an $\varepsilon$-approxi\-mation of the Pareto curve can be achieved
in time polynomial in the size of the MDP, exponential in the 
number of dimensions, and polynomial in $\frac{1}{\varepsilon}$, for $\varepsilon>0$.
\end{itemize}
In summary, we unify the two existing semantics, present comprehensive results related
to algorithmic and strategy complexities for the unifying semantics, and  
improve results for the existing semantics.

\smallskip\noindent{\bf Technical contributions.}
In the study of MDPs (with single or multiple rewards), the solution 
approach is often by characterizing the solution as a set of linear 
constraints.
Similar to the previous works~\cite{CMH06,EKVY08,FKN+11,krish} 
we also obtain our results by showing that the set of achievable 
solutions can be represented by a set of linear constraints, and 
from the linear constraints witness strategies for achievable solutions 
can be constructed.
However, previous work on the satisfaction semantics~\cite{krish,RRS15}
reduces the problem to invoking linear-programming solution for each maximal end-component and a separate 
linear program to combine the partial results together.
In contrast, we unify the solution approaches for expectation and satisfaction 
and provide one complete linear program for the whole problem.
This in turn allows us to optimize the expectation \emph{while} guaranteeing satisfaction.
Further, this approach immediately yields a linear program where both conjunctive and joint interpretations are combined, and we can optimize any linear combination of expectations. 
Finally, we can also optimize the probabilistic guarantees while ensuring the required expectation.
The technical device to obtain one linear program is to split the standard variables into several, depending on which subsets of constraints they help to achieve. This causes technical complications that have to be dealt with making use of conditional probability methods.

\smallskip\noindent{\bf Related work.} 
The study of Markov decision processes with multiple expectation objectives 
has been initiated in the area of applied probability theory, where it is 
known as {\em constrained MDPs}~\cite{Puterman,Altman}. 
The attention in the study of constrained MDPs has been mainly focused on 
restricted classes of MDPs, such as unichain MDPs, where all states are 
visited infinitely often under any strategy. 
Such a restriction guarantees the existence of memoryless optimal 
strategies.
The more general problem of MDPs with multiple mean-payoff objectives was first
considered in~\cite{Cha07} and a complete picture was presented in~\cite{krish}.
The expectation and satisfaction semantics was considered in~\cite{krish},
and our work unifies the two different semantics for MDPs. 
For general MDPs,~\cite{CMH06,CFW13} studied multiple discounted 
reward functions.
MDPs with multiple $\omega$-regular specifications 
were studied in~\cite{EKVY08}.
It was shown that the Pareto curve can be approximated
in polynomial
time in the size of MDP and exponential in the number of specifications;
the algorithm reduces the problem to MDPs with multiple
reachability specifications, which can be solved by multi-objective 
linear programming~\cite{PY00}.
In~\cite{FKN+11}, the results of~\cite{EKVY08} were extended to combine
$\omega$-regular and expected total reward objectives.
The problem of conjunctive satisfaction was introduced in \cite{FKR95}.
They present solution for only stationary (memoryless) strategies, and explicitly
mention that such strategies are not sufficient and a solution to the general 
problem  would be very useful.
They also mention that it is unlikely to be a simple extension of the single dimensional case.
Our results not only present the general solution, but we also present results that 
combine both the conjunctive  and joint satisfaction semantics along with the expectation semantics.
The multiple percentile are currently considered 
for various objectives, such as mean-payoff, limsup, liminf, shortest path in~\cite{RRS15}.
However, \cite{RRS15} does not consider optimizing the expectation, whereas
we consider maximizing expectation along with satisfaction semantics.
The notion of risks has been considered in MDPs with discounted
objectives~\cite{WL99}, where the goal is to maximize (resp.,
minimize) the probability (risk) that the expected total discounted
reward (resp., cost) is above (resp., below) a threshold.  The notion
of strong risk aversion, where for risk the probabilistic choices are
treated instead as an adversary was considered in~\cite{BFRR14}.
In~\cite{BFRR14} the problem was considered for single reward for
mean-payoff and shortest path.  In contrast, though inspired
by~\cite{BFRR14}, we consider risk aversion for multiple reward
functions with probabilistic guarantee (instead of adversarial
guarantee), which is natural for MDPs.  Moreover, \cite{BFRR14}
generalizes mean-payoff games, for which no polynomial-time solution
is known, whereas in our case, we present polynomial-time algorithms
for the single reward case and in several cases of multiple rewards
(see the first item of our contributions).  Further, an independent
work \cite{CR15lics} extends \cite{BFRR14} to multiple dimensions, and
they also consider ``beyond almost-sure threshold problem'', which
corresponds to the \emph{multi-qual} problem, which is a special case
of our solution.  Finally, a very different notion of risk has been
considered in~\cite{BCFK13}, where the goal is to optimize the
expectation while ensuring low variance.  The problem has been
considered only for single dimension, and no polynomial-time algorithm
is known.

\section{Preliminaries}\label{chap:prel}

\subsection{Basic definitions}

We mostly follow the basic definitions of \cite{krish} with only minor deviations.
We use  $\Nset,\Qset,\Rset$ to denote the sets of positive integers, rational and real numbers, respectively. 
For $n\in\Nset$, we denote $[n]=\{1,\ldots,n\}$.
For a sequence $\pat=\ell_1\ell_2\cdots$ and $n\in\Nset$, we denote the $n$-th element by $\pat[n]$.

Given two vectors $\vec{v},\vec{w} \in \Rset^k$, where $k \in \Nset$, we write $\vec{v} \geq \vec{w}$ iff 
$\vec{v}_i \geq \vec{w}_i$ for all $1 \leq i \leq k$, where $\vec v_i$ denotes the $i$-th 
component of vector $\vec v$. 
Further, $\vec 1$ denotes $(1,\ldots,1)$, and
$\kd{}$ denotes Kronecker's delta, i.e., $\kd{x}(x)=1$ and $\kd{x}(y)=0$ for $y\neq x$.

Finally, the set of all distributions over a~countable set $X$ is denoted by $\dist(X)$, and
$d\in\dist(X)$ is Dirac if $d(x)=1$ for some $x\in X$, i.e., $d=\kd{x}$. 
\smallskip

\smallskip\noindent{\bf Markov chains.} 
A \emph{Markov chain} is a tuple \mbox{$M = (L,P,\mu)$} where $L$ is a countable set of locations, 
$P:L\to\dist(L)$ is a probabilistic transition function, and
$\mu\in\dist(L)$ is the initial probability distribution.

A \emph{run} in $M$ is an infinite sequence $\pat = \ell_1 \ell_2 \cdots$ of locations,
a \emph{path} in $M$ is a finite prefix of a run. 
Each path $\fpat$ in $M$ determines the set $\Cone(\fpat)$ consisting of all runs that start with $\fpat$. 
To $M$ we associate the probability space $(\Pat,\calF,\pr)$, 
where $\Pat$ is the set of all runs in $M$, $\calF$ is the $\sigma$-field generated by all $\Cone(\fpat)$,
and $\pr$ is the unique probability measure such that
$\pr(\Cone(\ell_1\cdots\ell_k)) = 
\mu(\ell_1) \cdot \prod_{i=1}^{k-1} P(\ell_i)(\ell_{i+1})$.

\smallskip\noindent{\bf Markov decision processes.} 
A \emph{Markov decision process} (MDP) is defined as a tuple $G=(S,A,\mathit{Act},\delta,\inits)$ 
where $S$ is a finite set of states, $A$ is a finite set of actions, 
$\mathit{Act} : S\rightarrow 2^A\setminus \{\emptyset\}$ assigns to each state $s$ the set $\act{s}$ of actions enabled 
in $s$ so that $\{\act{s}\mid s\in S\}$ is a partitioning of $A$,
$\delta : A\rightarrow \dist(S)$ is a probabilistic 
transition function that given 
an action $a$ 
gives a probability distribution over the 
successor states, and $\inits$ is the initial state.
Note that we consider that every action is enabled in exactly one state.

A \emph{run} in $G$ is an infinite alternating sequence of states
and actions $\pat=s_1 a_1 s_2 a_2\cdots$
such that for all $i \geq 1$, we have $a_i\in\act{s_i}$ and $\delta(a_i)(s_{i+1}) > 0$. 
A \emph{path} of length~$k$ in~$G$ is a finite prefix
$\fpat = s_1 a_1\cdots a_{k-1} s_k$ of a run in~$G$.

\smallskip\noindent{\bf Strategies and plays.} 
The semantics of MDPs is defined using the notion of strategies.
Intuitively, a strategy in an MDP $G$ is a ``recipe'' to choose actions.
Usually, a strategy is formally defined as a function 
$\sigma : (SA)^*S \to \dist(A)$ that given a finite path~$\fpat$, representing 
the history of a play, gives a probability distribution over the 
actions enabled in the last state. In this paper, we adopt a slightly
different (though equivalent---see~\cite[Section 6]{krish})
definition, which is more convenient for our setting. 
Let $\mem$ be a countable set of \emph{memory elements}. 
A \emph{strategy} is a triple
$\sigma = (\sigma_u,\sigma_n,\alpha)$, where 
$\sigma_u: A\times S \times \mem \to \dist(\mem)$ and 
$\sigma_n: S \times \mem \to \dist(A)$ are \emph{memory update}
and \emph{next move} functions, respectively, and $\alpha$ is
the initial distribution on memory elements. We require that, for 
all $(s,m) \in S \times \mem$, the distribution $\sigma_n(s,m)$ assigns a
positive value only to actions enabled at~$s$, i.e.\ $\sigma_n(s,m)\in\dist(\act{s})$. 

A \emph{play} of $G$ determined by 
a strategy $\sigma$ is a Markov chain 
$G^\sigma=(S \times \mem \times A,P,\mu)$, where
\begin{align*} 
  \mu(s,m,a) &= \kd{s_0}(s) \cdot \alpha(m) \cdot \sigma_n(s,m)(a)\\
  P(s,m,a)(s',m',a') &= 
   \delta(a)(s')  \cdot \sigma_u(a,s',m)(m') \cdot \sigma_n(s',m')(a')\,. 
\end{align*}
Hence, $G^\sigma$ starts in a location chosen randomly according
to $\alpha$ and $\sigma_n$. In a current location $(s,m,a)$, 
the next action to be performed is $a$, hence the probability of entering
$s'$ is $\delta(a)(s')$. The probability of updating the memory to $m'$
is $\sigma_u(a,s',m)(m')$, and the probability of selecting $a'$ as 
the next action is $\sigma_n(s',m')(a')$. Note that these choices
are independent, and thus we obtain the product above. 
The induced probability measure is denoted by $\pr^{\sigma}$ and when the initial state $s$ is not clear from the context, we use $\pr^\sigma_s$ to denote $\pr^\sigma$ corresponding to the MDP where the initial state is set to $s$. 
``Almost surely'' or ``almost all runs'' refers to happening with probability 1 according to this measure.
The respective expected value of a random variable $f:\Pat\to\Rset$ is $\expected_s^\sigma[f]=\int_\Pat f\ d\,\pr_s^\sigma$ or $\expected^\sigma[f]=\int_\Pat f\ d\,\pr^\sigma$ for short.
For $t\in\Nset$, random variables $S_t,A_t$ return $s,a$, respectively, where $(s,m,a)$ is the $t$-th location on the run.


\smallskip\noindent{\bf Strategy types.}
In general, a strategy may use infinite memory $\mem$, and both 
$\sigma_u$ and $\sigma_n$ may randomize. The strategy is
\begin{itemize}
\item \emph{deterministic-update}, if $\alpha$ is Dirac and 
  the memory update function gives a~Dirac distribution for
  every argument;
\item \emph{stochastic-update}, if it is not necessarily deterministic-update;  
\item \emph{deterministic}, if it is deterministic-update and  
  the next move function gives a Dirac
  distribution for every argument;
\item  \emph{randomized}, if it is not necessarily deterministic.
\end{itemize}
We also classify the strategies according to the size of memory
they use. The important subclasses of strategies are 
\begin{itemize}
\item \emph{memoryless} (or \emph{$1$-memory}) strategies, in which $\mem$ is a singleton,
\item \emph{$n$-memory} strategies, in which $\mem$ has exactly $n$~elements, 
\item \emph{finite-memory} strategies, in which $\mem$ is finite, and
\item \emph{Markov} strategies, in which $\mem=\mathbb N$ and $\sigma_u(\cdot,\cdot,n)(n+1)=1$.
\end{itemize}
Markov strategies have a nice structure: they only need a counter and to know the current state \cite{FV97}.
\medskip

\smallskip\noindent{\bf End components.}
A set $T\cup B$ with $\emptyset\neq T\subseteq S$ and $B\subseteq \bigcup_{t\in T}\act{t}$
is an \emph{end component} of $G$
if (1) for all $a\in B$, whenever $\delta(a)(s')>0$ then $s'\in T$;
and (2) for all $s,t\in T$ there is a path 
$\pat = s_1 a_1\cdots a_{k-1} s_k$ such that $s_1 = s$, $s_k=t$, and all states
and actions that appear in $\pat$ belong to $T$ and $B$, respectively.
An end component $T\cup B$ is a \emph{maximal end component (MEC)}
if it is maximal with respect to the subset ordering. Given an MDP, the set of MECs is denoted by $\mec$.
Finally, if $(S,A)$ is a MEC, we call the MDP \emph{strongly connected}.

\begin{remark}\label{rem:endcomp}
The maximal end component (MEC) decomposition of an MDP, i.e., the computation of $\mec$, can be 
achieved in polynomial time~\cite{CY95}.
For improved algorithms for general MDPs and various special cases see~\cite{CH11,CH12,CH14,CL13}.
\end{remark}

Analogously, for a finite-memory strategy $\sigma$, a 
\emph{bottom strongly connected component} (BSCC) of $G^\sigma$ is 
a subset of locations \mbox{$W \subseteq S\times \mem \times A$} 
such that (i) for all $\ell_1 \in W$ and \mbox{$\ell_2 \in S\times \mem\times A$},
if there is a path from $\ell_1$ to $\ell_2$ then 
$\ell_2 \in W$, and (ii) for all $\ell_1,\ell_2 \in W$ we have a path from $\ell_1$ to $\ell_2$.
Every BSCC $W$ determines a unique end component 
$\{s,a\mid (s,m,a)\in W\}$
of $G$, and we sometimes do not strictly distinguish between $W$ 
and its associated end component.

For $C\in\mec$, let 
\[\Omega_{C}=\{\omega\in\Pat \mid \exists n_0:\forall n>n_0 :\omega[n]\in C\}\]
denote the set of runs with a suffix in $C$. Similarly, we define $\Omega_D$ for a BSCC $D$.
Since almost every run eventually remains in a MEC, e.g.\ \cite[Proposition~3.1]{CY98}, $\{\Omega_C\mid C\in\mec\}$ ``partitions'' almost all runs. More precisely, for every strategy, each run belongs to exactly one $\Omega_C$ almost surely; i.e.\ a run never belongs to two $\Omega_C$'s and for every $\sigma$, we have $\Pr{\sigma}{}{\bigcup_{C\in\mec}\Omega_C}=1$. Therefore, actions that are not in any MEC are almost surely taken only finitely many times.

\subsection{Problem statement}\label{ssec:problem}

In order to define our problem, we first briefly recall how long-run average can be defined.
Let $G=(S,A,\mathit{Act},\delta,\inits)$ be an MDP, $n\in\Nset$ and $\reward : A \to \Qset^n$ an $n$-dimensional \emph{reward function}.  
Since the random variable given by the limit-average function $\lrLim{\reward}=\lim_{T\rightarrow\infty} \frac{1}{T}\sum_{t=1}^{T} {\reward(A_t)}$ 
may be undefined for some runs, we consider maximizing the respective point-wise limit inferior:
\[
  \lrIf{\reward} = \liminf_{T\rightarrow\infty} 
     \frac{1}{T}\sum_{t=1}^{T} {\reward(A_t)}
\]
i.e. for each $i\in[n]$ and $\pat\in\Pat$, we have $\lrIf{\reward}(\pat)_i = \liminf_{T\rightarrow\infty} \frac{1}{T}\sum_{t=1}^{T} {\reward(A_t(\pat))_i}$.
Similarly, we could define $\lrSf{\reward} = \limsup_{T\rightarrow\infty} \frac{1}{T}\sum_{t=1}^{T} {\reward(A_t)}$.
However, maximizing limit superior is less interesting, see~\cite{krish}. 
Further, the respective minimizing problems can be solved by maximization with opposite rewards.

This paper is concerned with the following tasks:\medskip

\noindent
\framebox[\pagewidth]{
\begin{minipage}{0.96\textwidth}
{\smallskip

\textbf{Realizability (multi-quant-conjunctive):}
Given an MDP, 
$n\in\mathbb N,\reward:A\to\Qset^n,$ $\ex\in\Qset^n,\sat\in\Qset^n,\psat\in([0,1]\cap\Qset)^n$, decide whether there is a~strategy $\sigma$ such that $\forall i\in [n]$
\begin{align}
 \bullet\qquad& \Ex{\sigma}{}{\lrIf{\reward}_i}\geq \ex_i\,, \tag{EXP}\label{eq:EXP}\\
 \bullet\qquad& \Pr{\sigma}{}{\lrIf{\reward}_i\geq \sat_i}\geq \psat_i\,. \tag{conjunctive-SAT}\label{eq:SAT} 
\end{align}
\textbf{Witness strategy synthesis:}
If realizable, construct a~strategy satisfying the requirements.\smallskip

\textbf{$\varepsilon$-witness strategy synthesis:}
If realizable, construct a~strategy satisfying the requirements with $\ex-\varepsilon\cdot\vec1$ and 
$\sat-\varepsilon\cdot\vec1$.
\smallskip
}
\end{minipage}
}\medskip

We are mostly interested in \textbp{(multi-quant-conjunctive)} as it is the core of all other discussed problems.
However, we also consider the following important special cases:
\begin{description}
 \item[\textbp{(multi-qual)}\ \ ] \qquad\qquad$\psat=\vec1$\,,
 \item[\textbp{(mono-quant)}] \qquad\qquad$n=1$\,,
 \item[\textbp{(mono-qual)}\ \ ] \qquad\qquad$n=1,\psat=1$\,.
\end{description}

Additionaly, we are also interested in variants of \textbp{(multi-quant-conjunctive)}.
Firstly, in \textbp{(multi-quant-joint)}, the constraint (\ref{eq:SAT}) is \emph{replaced} by 
\begin{align}
\Pr{\sigma}{}{\lrIf{\reward}\geq \sat}\geq \psatscalar \tag{joint-SAT}\label{eq:krish-sat}
\end{align}
for $\psatscalar\in[0,1]$.
Secondly, \textbp{(multi-quant-conjunctive-joint)} arises by \emph{adding} (\ref{eq:krish-sat}) constraint
$\Pr{\sigma}{}{\lrIf{\reward}\geq \widetilde\sat}\geq \widetilde\psatscalar$
for $\widetilde\psatscalar\in[0,1]\cap\Qset$ and $\widetilde\sat\in\Qset^n$. 
The relationship between the problems is depicted in Fig.~\ref{fig:problems}.

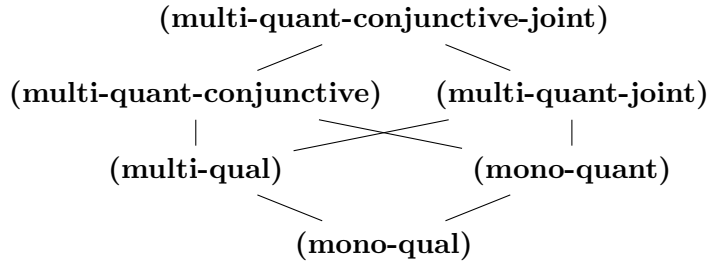
\begin{figure}[ht]
\centering
\begin{tikzpicture}[font=]
\node(p5) at(0,0) {\textbp{(multi-quant-conjunctive-joint)}};
\node(p4) at(-2.5,-1) {\textbp{(multi-quant-conjunctive)}};
\node(p44) at(2.5,-1) {\textbp{(multi-quant-joint)}};
\node(p3) at(-2.5,-2) {\textbp{(multi-qual)}};
\node(p2) at(2.5,-2) {\textbp{(mono-quant)}};
\node(p1) at(0,-3) {\textbp{(mono-qual)}};
\path[-] 
(p5) edge (p4)
(p5) edge (p44)
(p4) edge (p3)
(p4) edge (p2)
(p44) edge (p3)
(p44) edge (p2)
(p3) edge (p1)
(p2) edge (p1)
;
\end{tikzpicture}
\caption{Relationship of the defined problems with lower problems being specializations of the higher ones}\label{fig:problems}
\end{figure} 

\rstart Furthermore, each of the three constraints (\ref{eq:EXP}), (\ref{eq:SAT}), and (\ref{eq:krish-sat}) defines the respective decision problem given solely by that constraint. Each of these three problems is a special case of \textbp{(multi-quant-conjunctive-joint)} where the other constraints are trivial (e.g.\ requiring the average reward be greater or equal to the minimum reward of the MDP). \rstop
Finally, apart from decision problems, one often considers optimization problems, where the task is to maximize the parameters so that the answer to the decision problem is still positive. 
Observe that since optimization in multi-dimensional setting cannot in general produce a single ``best'' solution, one can consider Pareto curves, which are sets of all component-wise optimal and mutually incomparable solutions to the optimization problem.

\begin{example}[Running example]\label{ex:run-strategy}
We illustrate \textbp{(multi-quant-conjunctive)} with an MDP of Fig.~\ref{fig:run-MDP} with $n=2$, rewards as depicted, and $\ex=(1.1,0.5), \sat=(0.5,0.5), \psat=(0.8,0.8)$.
Observe that rewards of actions $\ell$ and $r$ are irrelevant as these actions can almost surely be taken only finitely many times.

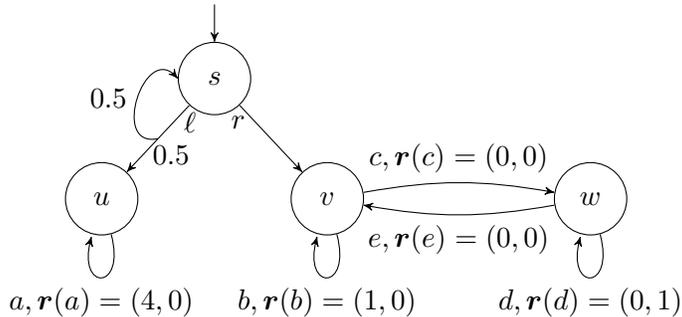
\begin{figure}[ht]
\spacefu
\begin{tikzpicture}
\node[state,initial above,initial text=](s) at(0,1.6) {$s$};
\node[state](u) at(-1.5,0) {$u$};
\node[state](v) at(1.5,0) {$v$};
\node[state](w) at(5,0) {$w$};
\node[coordinate](c) at(-0.75,0.8) {};
\path[->] 
(s) edge[-] node[right]{$\ell$}(c)
(c) edge node[right]{$0.5$}(u)
(u) edge[loop below] node[below]{$a,\reward(a)=(4,0)$}()
(c) edge[bend left=120pt,looseness=2] node[left]{$0.5$} (s.west)
(s) edge node[left,pos=0.25]{$r$} (v)
(v) edge[loop below] node[below]{$b,\reward(b)=(1,0)$}()
(v) edge[bend left=10pt]  node[above]{$c,\reward(c)=(0,0)$} (w)
(w) edge[loop below] node[below]{$d,\reward(d)=(0,1)$}()
(w) edge[bend left=10pt]  node[below]{$e,\reward(e)=(0,0)$} (v)
;
\end{tikzpicture}
\caption{An MDP with two-dimensional rewards}\label{fig:run-MDP}
\end{figure}

This instance is realizable and the witness strategy has the following properties. The strategy plays three ``kinds'' of runs.
Firstly, \rstart due to $\psat=(0.8,0.8)$, with probability at least $0.8+0.8-1=0.6$ runs have to jointly surpass both satisfaction thresholds (at the same time), i.e.\ exceed the vector $(0.5,0.5)$. \rstop This is only possible in the right MEC by playing each $b$ and $d$ half of the time and switching between them with a decreasing frequency, so that the frequency of $c,e$ is in the limit $0$.
Secondly, in order to ensure the expectation of the first reward, we reach the left MEC with probability $0.2$ and play $a$. Thirdly, with probability $0.2$ we reach again the right MEC but only play $d$ with frequency $1$, ensuring the expectation of the second reward.

In order to play these three kinds of runs, in the first step in $s$ we take $\ell$ with probability $0.4$ (arriving to $u$ with probability $0.2$) and $r$ with probability $0.6$, and if we return back to $s$ we play $r$ with probability $1$. If we reach the MEC on the right, we toss a biased coin and with probability $0.25$ we go to $w$ and play the third kind of runs, and with probability $0.75$ play the first kind of runs.

Observe that although both the expectation and satisfaction value thresholds for the second reward are $0.5$, the only solution is not to play all runs with this reward, but some with a lower one and some with a higher one. Also note that each of the three types of runs must be present in any witness strategy. Most importantly, in the MEC at state $w$ we have to play in two different ways, depending on which subset of value thresholds we intend to satisfy on each run.
\rstart Also note that in order to do that, we use memory with stochastic update. \rstop
\QEE\end{example}

\section{Solution}\label{chap:sol}

In this section, we briefly recall a solution to a previously considered problem 
and show our solution to the more general \textbp{(multi-quant-conjunctive)} realizability problem, 
along with an overview of the correctness proof. 
The solution to the other variants is derived and a detailed analysis of the special cases and the respective complexities is given in Section~\ref{sec:algc}.

\subsection{Previous results}

\subsubsection{Linear programming for expectation semantics}\label{ssec:lpforexp}
In \cite{krish}, a solution to the (\ref{eq:EXP}) constraint has been given.
The existence of a~witness strategy was shown equivalent to the existence of a solution to the linear program in Fig.~\ref{fig:lp-old}.

\begin{figure}[ht]
Requiring all variables $y_a,y_s,x_a$ for $a\in A,s\in S$ be non-negative, the program is the following: \bigskip

\begin{enumerate}
 \item transient flow: for $s\in S$
 $$\kd{s_0}(s)+\sum_{a\in A}y_a\cdot \delta(a)(s)=\sum_{a\in \act{s}} y_a+y_s$$
 \item almost-sure switching to recurrent behaviour:
 $$\sum_{s\in C\in \mec}y_s=1$$
 \item probability of switching in a MEC is the frequency of using its actions: for $C\in \mec$
 $$\sum_{s\in C}y_s=\sum_{a\in C}x_a$$
 \item recurrent flow: for $s\in S$
 $$\sum_{a\in A}x_a\cdot \delta(a)(s)=\sum_{a\in \act{s}} x_a$$
 \item expected rewards:
 $$\sum_{a\in A}x_a\cdot\reward\geq\ex$$
\end{enumerate}
\caption{Linear program of \cite{krish} for (\ref{eq:EXP})}
\label{fig:lp-old}
\spacefu
\end{figure}

Intuitively, $x_a$ is the expected frequency of using $a$ on the long run; Equation 4 thus expresses the recurrent flow in MECs and Equation 5 the expected long-run average reward. However, before we can play according to $x$-variables, we have to reach MECs and switch from the transient behaviour to this recurrent behaviour. Equation~1 expresses the transient flow before switching. Variables $y_a$ are the expected number of using $a$ until we switch to the recurrent behaviour in MECs and $y_s$ is the probability of this switch upon reaching $s$. To relate $y$- and $x$-variables, Equation 3 states that the probability to switch within a given MEC is the same whether viewed from the transient or recurrent flow perspective. 
Actually, one could eliminate variables $y_s$ and use directly $x_a$ in Equation~1 and leave out Equation~3 completely,
in the spirit of \cite{Puterman}. 
However, the form with explicit $y_s$ is more convenient for correctness proofs.
Finally, Equation 2 states that switching happens almost surely.
Note that summing Equation 1 over all $s\in S$ yields $\sum_{s\in S}y_s=1$. Since $y_s$ can be shown to equal $0$ for state $s$ not in MEC, Equation 2 is redundant, but again more convenient.

The solution above builds on the work \cite{EKVY08}, which studied MDPs with multiple reachability and $\omega$-regular specifications. It has inspired Equation 1 as well as computation of the Pareto curve.
It was shown that the Pareto curve can be approximated
in polynomial time in the size of MDP and exponential in the number of specifications;
the algorithm reduces the problem to MDPs with multiple
reachability specifications, which can be solved by multi-objective 
linear programming~\cite{PY00}.

\subsubsection{Linear programming for satisfaction semantics}

Apart from considering (\ref{eq:EXP}) separately, \cite{krish} also considers the constraint (\ref{eq:krish-sat}) separately.
While the former was solved using the linear program above, the latter required a reduction to one linear program per each MEC and another one to combine the results. More precisely, for each MEC we first decide whether there is a strategy exceeding the threshold. Second, we maximize the probability to reach  these MECs.
Similarly, in~\cite{RRS15}, for each MEC we decide for every subset of thresholds whether there is a strategy exceeding them. The results are again combined in a linear program for reachability.

In contrast, we shall provide a single linear program for the \textbp{(multi-quant-conjunctive)} problem, 
unifying the solution approaches for expectation and satisfaction problem.
This in turn allows us to optimize the expectation \emph{while} guaranteeing satisfaction.
Further, this approach immediately yields a linear program where both conjunctive and joint interpretations are combined, and we can optimize any linear combination of expectations. 
Finally, we can also optimize the probabilistic guarantees while ensuring the required expectation.
For greater detail, see Section~\ref{ssec:oneLP}.

\subsection{Our unifying solution} \label{sec:gsolution}

There are two main tricks to incorporate the satisfaction semantics. The first one is to ensure that a flow exceeds the value threshold. We first explain it on the qualitative case.

\subsubsection{Solution to \textbp{(multi-qual)}}\label{multiqual}

When the additional constraint (SAT) is added so that almost all runs satisfy $\lrIf{\reward}\geq\sat$, then the linear program of Fig.~\ref{fig:lp-old} shall be extended with the following additional equation:
\begin{enumerate}
 \item[6.] almost-sure satisfaction: for \rc{}$C\in\mec$
 $$ \sum_{a\in C}x_a\cdot\reward(a)\geq \sum_{a\in C}x_a\cdot \sat $$
\end{enumerate}

Note that $x_a$ represents the absolute frequency of playing $a$ (not relative within the MEC).
Intuitively, Equation 6 thus requires in each MEC the average reward be at least $\sat$.
Here we rely on the non-trivial fact, that in a MEC, actions can be played on almost all runs with the given frequencies for any flow, see Corollary~\ref{cor:flow-strat}.

\medskip

The second trick ensures that each conjunct in the satisfaction constraint can be handled separately and, consequently, that the probability threshold can be checked.

\subsubsection{Solution to \textbp{(multi-quant-conjunctive)}}

When each value threshold $\sat_i$ comes with a~non-trivial probability threshold $\psat_i$, some runs may and some may not have the long-run average reward exceeding $\sat_i$.
In order to speak about each group, we split the set of runs, for each reward, into parts which do and which do not exceed the threshold. 

Technically, we keep Equations 1--5 as well as 6, but split $x_a$ into $x_{a,N}$ for $N\subseteq [n]$, where $N$ describes the subset of exceeded thresholds; similarly for $y_s$. 
The linear program $L$ then takes the form displayed in Fig.~\ref{fig:lp-new}.

\begin{figure}
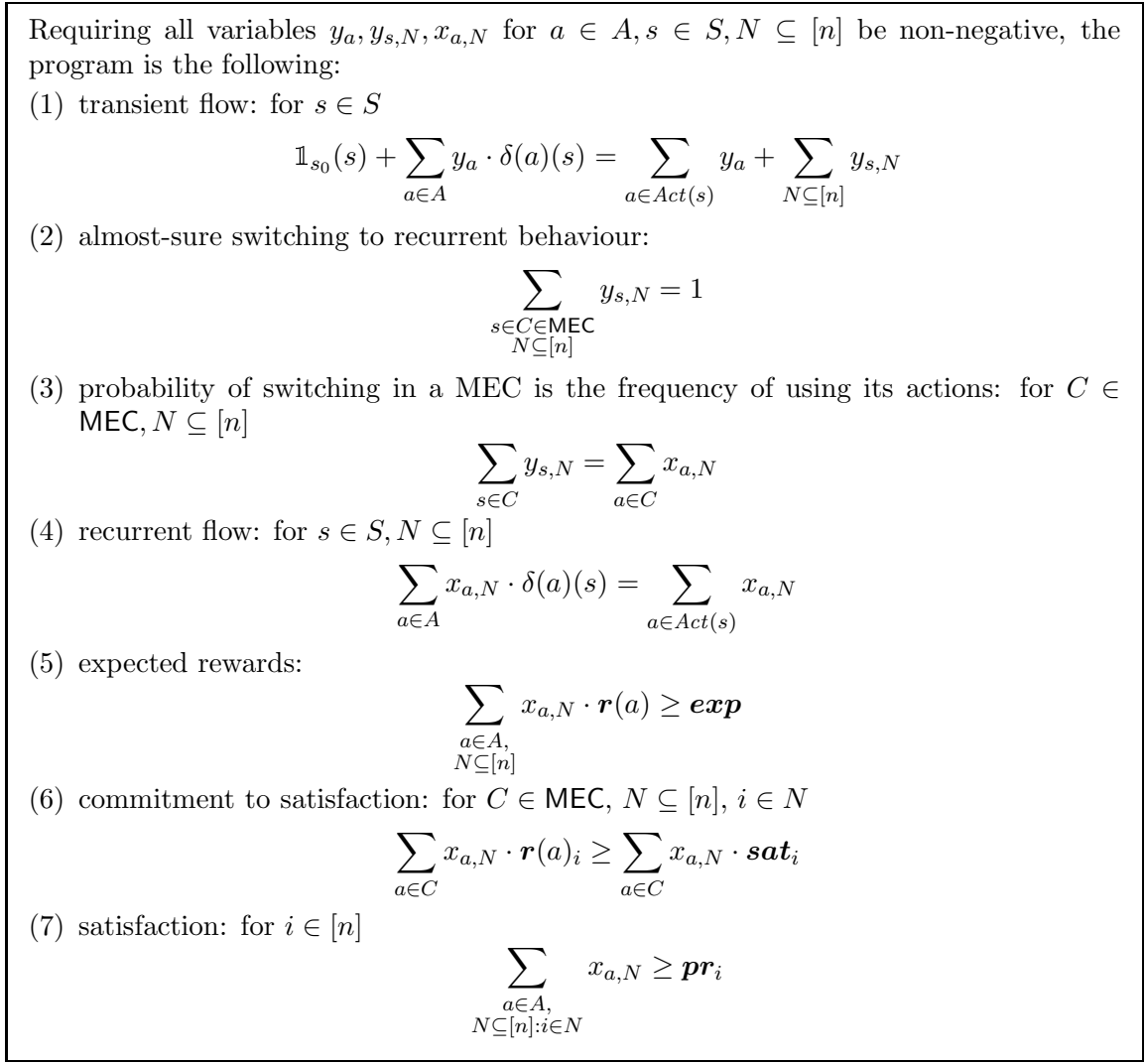

\framebox[\pagewidth]{\parbox{0.96\pagewidth}{\smallskip
Requiring all variables $y_{a},y_{s,N},x_{a,N}$ for $a\in A,s\in S,N\subseteq[n]$ be non-negative, the program is the following:
\begin{enumerate}
 \item transient flow: for $s\in S$
 $$\kd{s_0}(s)+\sum_{a\in A}y_a\cdot \delta(a)(s)=\sum_{a\in \act{s}} y_a+ {\sum_{N\subseteq[n]}y_{s,N}}$$
 \item almost-sure switching to recurrent behaviour:
 $$\sum_{\substack{s\in C\in \mec\\{N\subseteq[n]}}}y_{s, N}=1$$
 \item probability of switching in a MEC is the frequency of using its actions: for $C\in \mec,{N\subseteq[n]}$
 $$\sum_{s\in C}y_{s,{N}}=\sum_{a\in C}x_{a,N}$$
 \item recurrent flow: for $s\in S,{N\subseteq[n]}$
 $$\sum_{a\in A}x_{a, N}\cdot \delta(a)(s)=\sum_{a\in \act{s}} x_{a, N}$$
 \item expected rewards:
 $$\sum_{\substack{a\in A,\\ N\subseteq[n]}}x_{a, N}\cdot\reward(a)\geq\ex$$
 \item commitment to satisfaction: for $C\in\mec$, $N\subseteq[n]$, $i\in N$
 $$\sum_{a\in C}x_{a, N} \cdot \reward(a)_{ i}\geq \sum_{a\in C}x_{a, N} \cdot \sat_{ i}$$
 \item satisfaction: for $i\in[n]$
 $$\sum_{\substack{a\in A,\\N\subseteq[n]:i\in N}}x_{a,N}\geq \psat_i$$
\end{enumerate}
\smallskip}}
\caption{Linear program $L$ for \textbp{(multi-quant-conjunctive)}}
\label{fig:lp-new}
\spacefu
\end{figure} 

Intuitively, only the runs in the appropriate ``$N$-classes'' are required in Equation~6 to have long-run average rewards exceeding the satisfaction value threshold.
However, only the appropriate ``$N$-classes'' are considered for surpassing the probabilistic threshold in Equation~7.

\begin{theorem}\label{thm:main}
Given a \textbp{(multi-quant-conjunctive)} realizability problem, the respective system $L$ (in Fig.~\ref{fig:lp-new}) satisfies the following:
\begin{enumerate}
\item The system $L$ is constructible and solvable in time polynomial in the size of $G$ and exponential in $n$.
\item Every witness strategy induces a solution to $L$.
\item Every solution to $L$ effectively induces a witness strategy.
\end{enumerate}
\end{theorem}

\begin{example}[Running example]
The linear program $L$ for Example~\ref{ex:run-strategy} is shown in Appendix~\ref{app:lp}. Here we spell out some useful points we need later:
Equation 1 for state $s$
\begin{equation*}
1+0.5 y_\ell=y_\ell+y_r+y_{s,\emptyset}+y_{s,\{1\}}+y_{s,\{2\}}+y_{s,\{1,2\}}
\end{equation*}
expresses the Kirchhoff's law for the flow through the initial state. Equation 6 for the MEC $C=\{v,w,b,c,d,e\}$, $N=\{1,2\}$, $i=1$
\begin{equation*}
x_{b,\{1,2\}}\cdot1 
\geq (x_{b,\{1,2\}}+x_{c,\{1,2\}}+x_{d,\{1,2\}}+x_{e,\{1,2\}})\cdot0.5
\end{equation*}
expresses that runs ending up in $C$ and satisfying both satisfaction value thresholds have to use action $b$ at least half of the time. The same holds for $d$ and thus actions $c,e$ must be played with zero frequency on these runs. Equation 7 for $i=1$ sums up the gain of all actions on runs that have committed to exceed the satisfaction value threshold either for the first reward, or for the first \emph{and} the second reward.

Moreover, we show later in Lemma~\ref{lem:x-mec}, that variables $x_{\ell,N},x_{r,N}$ for any $N\subseteq[n]$ can be omitted from the system as they are zero for any solution. Intuitively, transient actions cannot be used in the recurrent flows.
\QEE\end{example}

\subsection{Proof overview}
Here, we briefly describe the main ideas of the proof of Theorem~\ref{thm:main}.

\subsection*{The first point} The complexity follows immediately from the syntax of $L$ and the existence of a~polynomial-time algorithm for linear programming \cite{Schrijver1986}.

\subsection*{The second point} 
Given a witness strategy $\sigma$, we construct values for variables so that a valid solution is obtained. The technical details can be found in Section~\ref{ssec:thm-proof2}. 

The proof of \cite[Proposition 4.5]{krish}, which inspires our proof, sets the values of $x_a$ to be the expected frequency of using $a$ by $\sigma$, i.e.\spacefl\myspace

$$
\lim_{T\to\infty}\frac1T\sum_{t=1}^T\Pr{\sigma}{}{A_t=a}$$
Since this Cesaro limit (expected frequency) may not be defined, a suitable value $f(a)$ between the limit inferior and superior has to be taken. 
In contrast to the approach of \cite{krish}, we need to distinguish among runs exceeding various subsets of the value thresholds $\sat_i,i\in[n]$.
For $N\subseteq[n]$, we call a run \emph{$N$-good} if $\lrIf{\reward}_i\geq\sat_i$ for exactly all $i\in N$.
$N$-good runs thus \emph{jointly} satisfy the $N$-subset of the constraints. 
Now instead of using frequencies $f(a)$ of each action $a$, we use frequencies $f_N(a)$ of the action $a$ on $N$-good runs separately, for each $N$. 
This requires some careful conditional probability considerations, in particular for Equations 1, 4, 6 and 7.

\begin{example}[Running example]
The strategy of Example~\ref{ex:run-strategy} induces the following $x$-values. For instance, action $a$ is played with a frequency $1$ on runs of measure $0.2$, hence $x_{a,\{1\}}=0.2$ and $x_{a,\emptyset}=x_{a,\{2\}}=x_{a,\{1,2\}}=0$. Action $d$ is played with frequency $0.5$ on runs of measure $0.6$ exceeding both value thresholds, and with frequency $1$ on runs of measure $0.2$ exceeding only the second value threshold. Consequently, $x_{d,\{1,2\}}=0.5\cdot0.6=0.3$ and $x_{d,\{2\}}=0.2$ whereas $x_{d,\emptyset}=x_{d,\{1\}}=0$.
\QEE\end{example}

Values for $y$-variables are derived from the expected number of taking actions during the ``transient'' behaviour of the strategy. Since the expectation may be infinite in general, an equivalent strategy is constructed, which is memoryless in the transient part, but switches to the recurrent behaviour in the same way. Then the expectations are finite and the result of \cite{EKVY08} yields values satisfying the transient flow equation. 
Further, similarly as for $x$-values, instead of simply switching to recurrent behaviour in a particular MEC, we consider switching in a MEC \emph{and} the set $N$ for which the following recurrent behaviour is $N$-good.  

\begin{example}[Running example]
The strategy of Example~\ref{ex:run-strategy} plays in $s$ for the first time $\ell$ with probability $0.4$ and $r$ with $0.6$, and next time $r$ with probability $1$. This is equivalent to a memoryless strategy playing $\ell$ with $1/3$ and $r$ with $2/3$. Indeed, both ensure reaching the left MEC with $0.2$ and the right one with $0.8$. Consequently, for instance for $r$, the expected number of taking this action is  \vspace*{-0.85em}

$$y_{r}=\frac23+\frac16\cdot\frac23+\left(\frac16\right)^2\cdot\frac23+\cdots=\frac56\,.$$\vspace*{-0.75em}

\noindent The values $y_{u,\{1\}}=0.2$, $y_{v,\{1,2\}}=0.6$, $y_{v,\{2\}}=0.2$ are given by the probability measures of each ``kind'' of runs (see Example~\ref{ex:run-strategy}). 
\QEE\end{example}

\subsection*{The third point} Given a solution to $L$, we construct a witness strategy $\sigma$, which has a particular structure. The technical details can be found in Section~\ref{ssec:thm-proof3}. The general pattern follows the proof method of \cite[Proposition 4.5]{krish}, but there are several important differences.

First, a strategy is designed to behave in a MEC so that the frequencies of actions match the $x$-values. The structure of the proof differs here and we focus on underpinning the following key principle. Note that the flow described by $x$-variables has in general several disconnected components within the MEC, and thus actions connecting them must not be played with positive frequency. Yet there are strategies that on almost all runs play actions of all components with exactly the given frequencies. The trick is to play the ``connecting'' actions with an increasingly negligible frequency. As a result, the strategy visits all the states of the MEC infinitely often, as opposed to strategies generated from the linear program in Fig.~\ref{fig:lp-old} in \cite{krish}, which is convenient for the analysis.

Second, the construction of the recurrent part of the strategy as well as switching to it has to reflect again the different parts of $L$ for different $N$, resulting in $N$-good behaviours.

\begin{example}[Running example]
A solution with $x_{b,\{1,2\}}=0.3, x_{d,\{1,2\}}=0.3$ induces two disconnected flows. Each is an isolated loop, yet we can play a strategy that plays both actions exactly half of the time. We achieve this by playing actions $c,e$ with probability $1/2^k$ in the $k-$th step. 
In Section~\ref{ssec:thm-proof3} we discuss the construction of the strategy from the solution in greater detail, necessary for later complexity discussion.
\QEE\end{example}

\subsection{Important aspects of our approach and its consequences}\label{ssec:oneLP} 
We now explain some important conceptual aspects of our result.
The previous proof idea from~\cite{krish} is as follows:
(1)~The problem for expectation semantics is solved by a linear program. 
(2)~The problem for satisfaction semantics is solved as follows: each MEC is considered, 
solved separately using a~linear program, and then a reachability problem is solved using a 
different linear program.
In comparison, our proof has two conceptual steps. 
Since our goal is to optimize the expectation (which intuitively requires a linear program), 
the first step is to come up with a single linear program for satisfaction semantics. 
The second step is to come up with a linear program that unifies the linear program for expectation semantics 
and the linear program for satisfaction semantics, allowing us to maximize expectation while ensuring satisfaction.



Since our solution captures all the frequencies separately within one linear program, we can work with all the flows at once. This has several consequences:
\begin{itemize}
\item While all the hard constraints are given as a part of the problem, we can easily find maximal solution with respect to a weighted reward expectation, i.e.\ $\vec w\cdot \lrIf{\reward}$, where $\vec w$ is the vector of weights for each reward dimension. Indeed, it can be expressed as the objective function $\vec w\cdot \sum_{a,N}x_{a,N}\cdot\reward(a)$ of the linear program. Further, it is also relevant for the construction of the Pareto curve.
\item We can also optimize satisfaction guarantees for given expectation thresholds. For more detail, see Section~\ref{sec:complexity-pareto}. 
\item \rstart We can easily add more satisfaction constraints (with different thresholds) on the same resource as well as add joint constraints of the form $\Pr{\sigma}{}{\bigwedge_{k_i} \lrIf{\reward_{k_i}}\geq \psatscalar}$. Both can be solved by adding a copy of Equation 7 for each subset $N$ of all the constraints. \rstop
\item The number of variables used in the linear program immediately yields an upper bound on the computational complexity of various subclasses of the general problem. Several polynomial bounds are proven in Section~\ref{sec:algc}. \QEE
\end{itemize}

\section{Proof of Theorem~\ref{thm:main}: Witness strategy induces solution to \texorpdfstring{$L$}{L}}\label{ssec:thm-proof2}

Now we present the technical proof of Theorem~\ref{thm:main}. We start with the second point and show how to construct a solution to $L$ from a witness strategy.

%
%
%
%
%




Let $\sigma$ be a strategy such that \rc{$\forall i\in [n]$
\begin{itemize}
 \item $\Pr{\sigma}{}{\lrIf{\reward}_i \geq \sat_i}\geq \psat_i$
 \item $\Ex{\sigma}{}{\lrIf{\reward}_i}\geq \ex_i$
\end{itemize}}
We construct a solution to the system $L$. The proof method roughly follows that of \cite[Proposition 4.5]{krish}.
However, separate flows for ``$N$-good'' runs require some careful conditional probability considerations, in particular for Equations 4, 6 and 7.

\subsection{Recurrent behaviour and Equations 4--7}\label{subs:47}

We start with constructing values for variables $x_{a,N},a\in A,N\subseteq[n]$.

In general, the frequencies 
of the actions may 
not be well defined, because the defining limits may not exist. 
\rc{Further, it may be unavoidable to have different frequencies for several sets of runs of positive measure.
There are two tricks to overcome this difficulty. Firstly, we partition the runs into several classes depending on which parts of the objective they achieve. Secondly, within each class we pick suitable values lying between $\lrIf{\reward}$ and $\lrSf{\reward}$ of these runs.}
\rc{
In order to achieve the first point, we define	
for $N\subseteq[n]$, 
$$\Omega_N=\{\omega\in \Pat\mid \forall i\in N: \lrIf{\reward} (\omega)_i\geq \sat_i \wedge 
\forall i\notin N:\lrIf{\reward}(\omega)_i  < \sat_i\}$$
Then $\Omega_N$, $N\subseteq[n]$ form a partitioning of $\Pat$. Further, observe that runs of $\Omega_N$ are the runs where joint satisfaction holds, for all rewards $i\in N$. This is important for the algorithm for \textbp{(multi-quant-joint)} from Section~\ref{sec:algc}.

\rstart In order to achieve the second point,
we define
$f_N(a)$, for every $a$, \rstop 
to be lying between values
$
\liminf_{T\rightarrow \infty} \frac{1}{T}\sum_{t=1}^{T} \Pr{\sigma}{s_0}{A_t=a\rstart\cap\Omega_N\rstop}
$
and
$
\limsup_{T\rightarrow \infty} \frac{1}{T}\sum_{t=1}^{T} \Pr{\sigma}{s_0}{A_t=a\rstart\cap\Omega_N\rstop}
$, which can be safely substituted for $x_{a,N}$ in~$L$. 
\rstart 
Let $A$ be written as $\{a_1,a_2,\ldots,a_{|A|}\}$
and let us first consider the case when $\Pr\sigma{s_0}{\Omega_N}>0$. 
Since every bounded infinite
sequence contains an infinite convergent subsequence, there is an increasing
sequence of indices, $T_0^1, T_1^1, T_2^1 \ldots$, such that
$\lim_{\ell\to\infty}
\frac{1}{T_\ell^1}
\sum_{t=1}^{T_\ell^1} \Pr\sigma{s_0}{A_t=a_1\mid \Omega_N}$ is well defined.
Then we can choose a subsequence $T_0^2, T_1^2, T_2^2 \ldots$ of the sequence $T_0^1, T_1^1, T_2^1 \ldots$ so that $\lim_{\ell\to\infty}
\frac{1}{T_\ell^1}
\sum_{t=1}^{T_\ell^1} \Pr\sigma{s_0}{A_t=a_1\mid \Omega_N}$ is well defined, too.
We continue this process for all actions and finally define the sequence $T_0, T_1, T_2 \ldots$ to be $T_0^{|A|}, T_1^{|A|}, T_2^{|A|} \ldots$.
Consequently, for each action $a\in A$,
the following limit exists
\rstop
\[
f_N(a)
\coloneqq
\lim_{\ell\to\infty}
\frac{1}{T_\ell}
\sum_{t=1}^{T_\ell} \Pr\sigma{s_0}{A_t=a\mid \Omega_N} \cdot \Pr\sigma{s_0}{\Omega_N}
\]
\rstart
and we set for all $a\in A$ 
\[
x_{a,N} \coloneqq f_N(a)
\]
Finally, for $N$ such that $\Pr\sigma{s_0}{\Omega_N}=0$, we set $x_{a,N}:=0$. 
Note that \rstop since actions not in MECs are almost surely taken only finitely many times, we have
\begin{align}
x_{a,N}=0 \qquad\text{ for }a\notin\bigcup\mec, N\subseteq[n] \label{eq:transient-x}
\end{align}

We show that (in)equations 4--7 of $L$ are satisfied.}



\subsection*{Equation 4}

For $N\subseteq[n],t\in \Nset,a\in A,s\in S$, let 
$$\Delta^N_t(a)(s):=\Pr{\sigma}{}{S_{t+1}=s\mid A_t=a,\ \Omega_N}$$
denote the ``transition probability'' at time $t$ restricted to runs in $\Omega_N$. 
In general, $\Delta^N_i(a)(s)$ may be different from $\delta(a)(s)$. 
\rstart
However, we show that if we use the action $a$ with positive frequency then $\Delta^N_i(a)(s)$ approximates $\delta(a)(s)$.

\begin{example}
Consider an action $a$ with $\delta(a)(u)=0.5$. 
Then we have $\Pr{\sigma}{}{S_{2}=u\mid A_1=a}=0.5$.
It may well be that for some set $\Omega\subseteq \Pat$ we have 
$\Pr{\sigma}{}{S_{2}=u\mid A_1=a,\ \Omega}=1$, but then $\Pr{\sigma}{}{\Omega}\leq0.5$.
Similarly, if $\Pr{\sigma}{}{S_{2}=u\mid A_1=a,\ \Omega}=\Pr{\sigma}{}{S_{3}=u\mid A_2=a,\ \Omega}=1$ then $\Pr{\sigma}{}{\Omega}\leq0.25$, and so on.
In general, whenever $\Pr{\sigma}{}{\Omega}>0$, the transition probabilities on $\Omega$ cannot differ from the actual transition probabilities too much all the time. \QEE


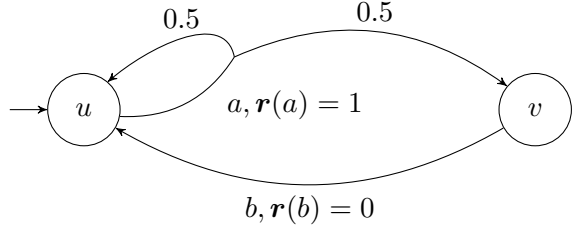
\begin{figure}[ht]
\centering
\begin{tikzpicture}
\node[state,initial,initial text=](u)at(0,0) {$u$};
\node[coordinate](c)at(2,0.7) {};
\node[state](v)at(6,0) {$v$};
\path
(u) edge[-,bend right] node[below right, pos=0.8]{$a,\reward(a)=1$} (c)
(c) edge[->,bend right,out=270] node[above]{$0.5$}(u)
(c) edge[->, bend left] node[above]{$0.5$}(v)
(v) edge[->, bend left]  node[below]{$b,\reward(b)=0$} (u)
;
\end{tikzpicture}
\caption{An MDP illustrating $\Delta$}\label{fig:ex-?}
\end{figure}
\end{example}
\smallskip

%

We first consider a simpler problem:
	\begin{lemma}\label{lem:cond-flow-simpler}
		Let $(\Delta_t)_{t\in\Nset}$ be i.i.d.~Bernoulli variables with expectation $\delta=\expected [\Delta_t]$. Then for any event $\Omega$ with $\pr{[\Omega]}>0$, we have  $\displaystyle\lim_{t\to\infty}\expected_\Omega [\Delta_t]=\delta
		$.
	\end{lemma}
	\begin{proof}
		For a contradiction,
		let w.l.o.g.\ $\limsup_{t\to\infty}\expected_\Omega [\Delta_t]=\delta+3\varepsilon$. 
		(If $\limsup_{t\to\infty}\expected_\Omega [\Delta_t]<\delta$, we can consider the variables $1-\Delta_t$ with this property).
		Moreover, we may safely assume that $\expected_\Omega [\Delta_t]\geq\delta+2\varepsilon$ for all $t\in\Nset$, otherwise we consider the respective subsequence. 
		Let $\mathit{High}_i\subseteq\Omega$ be the set of runs of $\Omega$ such that $\frac1i\sum_{t=1}^i \Delta_t>\delta+\varepsilon$ and similarly $\mathit{Normal}_i\subseteq\Omega$ be the set of runs of $\Omega$ such that $\frac1i\sum_{t=1}^i \Delta_t\leq\delta+\varepsilon$.
		Clearly, $\Omega=\mathit{High}_i\uplus \mathit{Normal}_i$ for every $i$.
		Then
		\begin{align*}
		\delta+2\varepsilon&\leq  \frac 1i \sum_{t=1}^i \expected_\Omega[ \Delta_t]=  
		\frac 1i  \expected_\Omega\left[ \sum_{t=1}^i \Delta_t\right]\\&=
		\frac{\frac 1i \expected_{\mathit{High}_i}[\sum_{t=1}^i \Delta_t]\cdot\pr[\mathit{High}_i]+
		\frac 1i \expected_{\mathit{Normal}_i}[\sum_{t=1}^i \Delta_t]\cdot\pr[\mathit{Normal}_i]}  {\pr[\mathit{High}_i]+\pr[\mathit{Normal}_i]}\\&
		\leq \frac {1\cdot\pr[\mathit{High}_i] + (\delta+\varepsilon)\cdot\pr[\mathit{Normal}_i] }   {\pr[\mathit{High}_i]+\pr[\mathit{Normal}_i]}
		\end{align*}
		Altogether, by comparing the first and the last expression, we get 
		\begin{equation}
		\pr[\mathit{Normal}_i]\leq \frac{1-\delta-2\varepsilon}{\varepsilon} \cdot \pr[\mathit{High}_i] \label{eq:dense-runs-proportion1}
		\end{equation} where the fraction is constant for all $i$.
		Since by the law of large numbers $\lim_{i\to \infty} \pr[\mathit{High}_i]=0$, we obtain $\lim_{i\to \infty} \pr[\mathit{Normal}_i]=0$ and thus $\pr[\Omega]=0$, a contradiction.
	\end{proof}
	
Now we apply the preceding lemma to MDPs:

\begin{lemma}\label{lem:cond-flow}
Let $N\subseteq[n]$ be such that $\Pr{\sigma}{}{\Omega_N}>0$. 
Then for every $a\in A,s\in S$
, we have  
$\displaystyle\lim_{t\to\infty} \Pr{\sigma}{}{A_t=a\mid \Omega_N}\cdot|\Delta^N_t(a)(s)-\delta(a)(s)|=0$.
\end{lemma}

\noindent\textit{Proof plan.}
Note that if $\Pr{\sigma}{}{A_t=a\mid\Omega_N}=1$ for all $t$ then the result follows directly from the previous lemma where we set $\Delta_t(\omega)$ to $1$ if $S_{t+1}=s$ and $0$ otherwise.
Indeed, then $\expected[\Delta_t]=\delta(a)(s)$ and $\expected_{\Omega_N}[\Delta_t]=\Delta_t^N(a)(s)$.
Consequently, $\lim_{t\to\infty} \Pr{\sigma}{}{A_t=a\mid \Omega_N}\cdot|\Delta^N_t(a)(s)-\delta(a)(s)|=1\cdot0$.

In the general case, the probability of taking $a$ on the runs can vary over time.
In order to cope with that, we consider sets $I\subset\Nset$ of positions where $a$ is taken with high enough probability (i.e., in ``many'' runs).
The first step of the proof is thus to derive (\ref{eq:dense-runs-proportion2}), an analogue of (\ref{eq:dense-runs-proportion1}), but now relativized to positions in $I$.
In the previous lemma, the second step consisted in applying the law of large numbers to conclude that probability of overly high preference of some outcome has zero probability, causing a contradiction with (\ref{eq:dense-runs-proportion1}).
In this proof, the second step will require more math to conclude that, due to the relativization.

\begin{proof}
Suppose for a contradiction, that for some $a\in A,s\in S$ there are infinitely many $t$ for which 
$\Pr{\sigma}{}{A_t=a\mid \Omega_N}\cdot|\Delta^N_t(a)(s)-\delta(a)(s)|>\xi$ for some $\xi>0$. 
Denote the set of these $t$'s by $T$.
Since both factors are bounded by $0$ and $1$, there are $\zeta>0$ and $\varepsilon>0$ such that for all $t\in T$ we have
$\Pr{\sigma}{}{A_t=a\mid \Omega_N}>\zeta$ and w.l.o.g.\ $\Delta^N_t(a)(s)>\delta(a)(s)+2\varepsilon$ (if $\Delta^N_t(a)(s)<\delta(a)(s)$ then there is another successor $s'$ of $a$ with this property).
Consequently, for every $t\in T$, we have 
$$\frac {\Pr{\sigma}{}{\Omega_N\cap A_t=a\cap S_{t+1}=s}} {\Pr{\sigma}{}{\Omega_N\cap A_t=a} \hspace*{11mm}}>\delta(a)(s)+2\varepsilon$$ 
\smallskip

\noindent\textit{First step.}
Now we derive (\ref{eq:dense-runs-proportion2}), a version of (\ref{eq:dense-runs-proportion1}) relativized to finite sets $I\subseteq T$.
The positive probability of taking $a$ in these positions guarantees that overly high preference of the outcome $s$ is well defined.

Formally, similarly to the previous inequality for each $t\in T$, the same holds for the average over any finite set of indices $I\subseteq T$:
\[
\delta(a)(s)+2\varepsilon < \frac {\sum_{t\in I} \Pr{\sigma}{}{\Omega_N\cap A_t=a\cap S_{t+1}=s}} {\sum_{t\in I} \Pr{\sigma}{}{\Omega_N\cap A_t=a} \hspace*{18mm}} =(*)
\]
Denoting 
\begin{align*}
i\text{-Tries-In-}I&=\{\omega\in\Omega_N \mid |\{t\in I\mid A_t=a\}|=i\}\\
i\text{-Successes-In-}I&=\{\omega\in\Omega_N \mid |\{t\in I\mid A_t=a\cap S_{t+1}=s\}|=i\}
\end{align*}
we can rewrite the term ($*$) by grouping runs with same ``frequencies'' as
\[
(*)=\frac
{\sum_{i=1}^{|I|}  i\cdot 
	\Pr{\sigma}{}{i\text{-Successes-In-}I}
}
{\sum_{i=1}^{|I|}  i\cdot 
	\Pr{\sigma}{}{i\text{-Tries-In-}I}
	\hspace*{7mm}
}
=(**)
\]
Similarly to the previous lemma, we introduce runs with ``success rate'' higher and lower than $\delta(a)(s)+\varepsilon$, now relative to the indices of $I$. Formally,
\begin{align*}
\mathit{High}_i^I&= i\text{-Tries-In-}I\cap \bigcup_{k>i\cdot \big(\delta(a)(s)+\varepsilon\big)} k\text{-Successes-In-}I\\
\mathit{Normal}_i^I&= i\text{-Tries-In-}I\cap \bigcup_{k\leq i\cdot \big(\delta(a)(s)+\varepsilon\big)} k\text{-Successes-In-}I
\end{align*}
allows us to rewrite
\[
(**)=
\frac
{\sum_{i=1}^{|I|}  (i\cdot \mathit{HighRate}_i)\cdot\Pr{\sigma}{}{\mathit{High}_i^I}
+\sum_{i=1}^{|I|}  (i\cdot \mathit{NormalRate}_i)\cdot\Pr{\sigma}{}{\mathit{Normal}_i^I}
}
{\sum_{i=1}^{|I|}  i\cdot \Pr{\sigma}{}{\mathit{High}_i^I}
	+\sum_{i=1}^{|I|}  i\cdot\Pr{\sigma}{}{\mathit{Normal}_i^I}
}
=({**}*)
\]
where each $\mathit{HighRate}_i\in(\delta(a)(s)+\varepsilon,1]$ and $\mathit{NormalRate}_i\in[0,\delta(a)(s)+\varepsilon]$ are the average portions of ``successes'' among the ``tries'' in the respective $\mathit{High}_i^I$ and $\mathit{Normal}_i^I$. 
Hence we can safely use the upper bounds to show
\[
({**}*)\leq
\frac
{1\cdot\sum_{i=1}^{|I|}  i\cdot\Pr{\sigma}{}{\mathit{High}_i^I}
	+(\delta(a)(s)+\varepsilon)\cdot\sum_{i=1}^{|I|}  i\cdot\Pr{\sigma}{}{\mathit{Normal}_i^I}
}
{\sum_{i=1}^{|I|}  i\cdot \Pr{\sigma}{}{\mathit{High}_i^I}
	+\sum_{i=1}^{|I|}  i\cdot\Pr{\sigma}{}{\mathit{Normal}_i^I}
}
=({**}{**})
\]
Since $({**}{**})\geq(*)\geq \delta(a)(s)+2\varepsilon$, we get by the same computation as for obtaining (\ref{eq:dense-runs-proportion1})
\begin{equation}
\sum_{i=1}^{|I|}  i\cdot \Pr{\sigma}{}{\mathit{Normal}_i^I}\leq \frac{1-\delta-2\varepsilon}{\varepsilon} \cdot \sum_{i=1}^{|I|}  i\cdot \Pr{\sigma}{}{\mathit{High}_i^I} \label{eq:dense-runs-proportion2}
\end{equation}
for every finite $I\subseteq T$.
\smallskip

\noindent\textit{Second step.}
Now we consider particular $I$'s leading to a contradiction.
Let $T$ be written as $\{t_1,t_2,\ldots\}$ so that $t_1< t_2<\cdots$.
For $m<n$, we consider finite subsets $I_m^n=\{t_m,t_{m+1},\ldots,t_n\}$ of $T$ and will prove that 
\begin{equation}
\lim_{m\to\infty} \lim_{n\to\infty} \sum_{i=1}^{|I_m^n|} i\cdot \Pr{\sigma}{}{\mathit{High}_i^{I_m^n}}=0 \label{eq:step2}
\end{equation}
As a consequence of (\ref{eq:dense-runs-proportion2}) we obtain also 
$\lim_{m\to\infty} \lim_{n\to\infty} \sum_{i=1}^{|I_m^n|} i\cdot \Pr{\sigma}{}{\mathit{Normal}_i^{I_m^n}}=0$ and thus 
$\lim_{m\to\infty}\lim_{n\to\infty}\sum_{i=1}^{|I_m^n|}i\cdot\Pr{\sigma}{}{i\text{-Tries-In-}I_m^n}=0$, i.e.\ with growing $m$ the average number of tries after $m$ approaches $0$, a contradiction with $\Pr{\sigma}{}{A_t=a\mid \Omega_N}>\zeta$ for infinitely many $t$ and $\Pr{\sigma}{}{\Omega_N}>0$.

It remains to prove (\ref{eq:step2}).
Intuitively, we consider index sets that start later (at position $m\to\infty$) to avoid initial potentially large elements.
Summands with high $i$'s, i.e.\ runs with many tries, below denoted by $\mathcal C$, will be shown negligible by the central limit theorem (in the previous lemma the law of large numbers was sufficient).
Further, we will have to argue that even summands with low $i$'s are small for high enough $m$. 
This is due to the fact that either $a$ is taken frequently enough on some runs ($\mathcal A$) or for high enough indices not any more on the other runs ($\mathcal B$).

Formally, let $\mathit{Inf}=\Omega_N\cap\{A_t=a\text{ for infinitely many }t\}$ and $\mathit{Fin}_{\geq k}=\Omega_N\cap\{A_t=a\text{ for only finitely many }t\}\cap\{A_t=a\text{ for some }t\geq k\}$. 
We split the sum $\sum_{i=1}^{|I_m^n|} i\cdot \Pr{\sigma}{}{\mathit{High}_i^{I_m^n}}$ into
\[
\underbrace{\sum_{i=1}^{\mathit{middle}(m)} i\cdot \Pr{\sigma}{}{\mathit{High}_i^{I_m^n}\cap \mathit{Inf}}}_{\mathcal A}
	+
\underbrace{\sum_{i=1}^{\mathit{middle}(m)} i\cdot \Pr{\sigma}{}{\mathit{High}_i^{I_m^n}\cap \mathit{Fin}_{\geq m}}}_{\mathcal B}
	+	
\underbrace{\sum_{i=\mathit{middle}(m)+1}^{|I_m^n|} i\cdot \Pr{\sigma}{}{\mathit{High}_i^{I_m^n}}}_{\mathcal C}
\]
by defining an appropriate $\mathit{middle}:\Nset\to\Nset$.
We show that each term approaches zero.
\begin{enumerate}[label=$\mathcal{\alph*}:$]
	\item[$\mathcal A$:] 	
	Observe that for every $i$ and $m$, we have $\lim_{n\to\infty} \Pr{\sigma}{}{ i\text{-Tries-In-}I_m^n\cap\mathit{Inf} }=0$.
	Hence also $\lim_{n\to\infty}\mathcal A=0$ for every $m$ and irrespective of the choice of $\mathit{middle}(m)$, and thus $\lim_{m\to\infty}\lim_{n\to\infty}\mathcal A=0$.
	\item[$\mathcal B$:] 
	We define $\mathit{middle}(m)$ to be the largest number such that $\sum_{i=1}^{\mathit{middle}(m)} i\cdot \Pr{\sigma}{}{ \mathit{Fin}_{\geq m}} < 1/m$.
	This trivially ensures $\lim_{m\to\infty} \mathcal B\leq\lim_{m\to\infty}1/m=0$.
	\item[$\mathcal C$:] 
	Since $\lim_{m\to\infty}\Pr{\sigma}{}{ \mathit{Fin}_{\geq m}}=0$, we obtain by the definition of $\mathit{middle}$ that for $m\to\infty$ also $\mathit{middle}(m)\to\infty$.
	Consequently, it is sufficient to prove that 
	\begin{equation}
	\lim_{n\to\infty}\sum_{i=k}^{|I_m^n|} i\cdot \Pr{\sigma}{}{\mathit{High}_i^{I_m^n}}\to 0 \text{ for } k\to\infty \text{ uniformly for all }m\,.\label{eq:clt}
	\end{equation} 
	Fix an arbitrary $m$.
	Let $X_j$ denote the indicator random variable of the event that $j$th use of action $a$, when looking only at time points $t_m, t_{m+1},t_{m+2}\ldots$, resulted in the successor $s$. Precisely, let $T_j$ be an auxiliary random variable with value $t_{\ell}$ such that $|\{q\mid m\leq q\leq \ell, A_{t_q}=a\}|=j$ and $A_{t_q}=a$; then $X_j$ is $1$ if $S_{T_j+1}=s$ and $0$ otherwise. Due to the Markov property, $X_j$ are Bernoulli i.i.d.\ with mean $\delta(a)(s)$.
	Further, \[\mathit{High}_i^{I_m^n}\subseteq\left\{\frac{\sum_{j=1}^i X_j}{i}>\delta(a)(s)+\varepsilon\right\}\]
	Therefore, by central limit theorem 
	\[
	\Pr{\sigma}{}{\mathit{High}_i^I}\lessapprox 
	\Phi(-\sqrt i \cdot \hat \varepsilon)
	\]
	where $\hat\varepsilon=\varepsilon/{\sqrt{\delta(a)(s)\cdot(1-\delta(a)(s))}}$ and $\Phi$ is the cumulative distribution function of the standard normal distribution and $\lessapprox$ denotes that the inequality $\leq$ holds ``only for large $i$'', i.e.\ in the limit.
	Consequently, for large $k$, we have
	\[\lim_{n\to\infty}\sum_{i=k}^{|I_m^n|} i\cdot \Pr{\sigma}{}{\mathit{High}_i^{I_m^n}}\lessapprox
	\sum_{i=k}^{\infty} i\cdot\Phi(-\sqrt i\cdot \hat{\varepsilon}) 
	\]
	where the right-hand side does not depend on $m$ and is thus a uniform bound for all $m$.
	Further, since $\Phi(-\sqrt i\cdot \hat{\varepsilon})$ decreases exponentially in $\sqrt i$, the right-hand side approaches $0$ as $k\to0$ (independently of $m$) and (\ref{eq:clt}) follows.
\qedhere
\end{enumerate}


%
%

\end{proof}

\rstop

\medskip

\noindent Now we show, that Equation 4 is satisfied. For all $s\in S$ and $N\subseteq[n]$ such that $\Pr\sigma{s_0}{\Omega_N}=0$, we have trivially
\begin{equation*}\label{eq:invf}
\sum_{a\in A} \rstart x_{a,N}\rstop\cdot \delta(a)(s) = \sum_{a \in \act{s}} \rstart x_{a,N}\rstop  
\end{equation*}
 and whenever $\Pr\sigma{s_0}{\Omega_N}>0$ we have
\begin{align*}
&\frac1{\Pr\sigma{s_0}{\Omega_N}}\sum_{a\in A} f_{\rc N}(a)\cdot \delta(a)(s) \\
   & = \frac1{\Pr\sigma{s_0}{\Omega_N}} \sum_{a\in A}\lim_{\ell\to\infty}  \frac{1}{T_\ell} \sum_{t=1}^{T_\ell} \Pr\sigma{s_0}{A_t=a\rc{\ \mid \Omega_N}}
   \cdot {\Pr\sigma{s_0}{\Omega_N}} \cdot \delta(a)(s)\tag{definition of $f_N$}\\
   & = \sum_{a\in A} \lim_{\ell\to\infty} \frac{1}{T_\ell} \sum_{t=1}^{T_\ell}
    \Pr\sigma{s_0}{A_t=a\rc{\ \mid \Omega_N}}
   \cdot \delta(a)(s) \tag{linearity of the limit}\\
   & \rstart = \sum_{a\in A} \lim_{\ell\to\infty} \frac{1}{T_\ell} \sum_{t=1}^{T_\ell}
    \Pr\sigma{s_0}{A_t=a\rc{\ \mid \Omega_N}}
   \cdot \Delta^N_t(a)(s) \tag{Lemma~\ref{lem:cond-flow}}\\
   & \rstart =  \lim_{\ell\to\infty} \frac{1}{T_\ell} \sum_{t=1}^{T_\ell}
   \sum_{a\in A} \Pr\sigma{s_0}{A_t=a\rc{\ \mid \Omega_N}}
   \cdot \Delta^N_t(a)(s) \tag{definition of $T_\ell$}\\ \rstop
   %
   &=\lim_{\ell\to\infty} \frac{1}{T_\ell} \sum_{t=1}^{T_\ell} \Pr\sigma{s_0}{S_{t+1}=s \rc{\ \mid\Omega_{N}}} \tag{definition of $\Delta^N_t$}\\
   & = 
   \lim_{\ell\to\infty} \frac{1}{T_\ell} \sum_{t=1}^{T_\ell}
   \Pr\sigma{s_0}{S_{t}=s\rc{\ \mid\Omega_{N
   }}}\tag{reindexing and Cesaro limit} \\
   & = 
   \lim_{\ell\to\infty} \frac{1}{T_\ell} \sum_{t=1}^{T_\ell}
   \sum_{a \in \act{s}}
   \Pr\sigma{s_0}{A_{t}=a\rc{\ \mid\Omega_{N
   }}} \tag{$s$ must be followed by $a\in \act{s}$}\\
   & =
   \frac1{\Pr\sigma{s_0}{\Omega_N}} 
   \sum_{a \in \act{s}}
   \lim_{\ell\to\infty} \frac{1}{T_\ell} \sum_{t=1}^{T_\ell}
   \Pr\sigma{s_0}{A_{t}=a\rc{\ \mid\Omega_{N}}} \cdot {\Pr\sigma{s_0}{\Omega_N}}
   \tag{linearity of the limit}\\
   & =  
   \frac1{\Pr\sigma{s_0}{\Omega_N}} \sum_{a \in \act{s}}
    f_N(a)\tag{definition of $f_N$}
    \;.
\end{align*}

\subsection*{Equation 5}

For all $i\in[n]$, we have 
\begin{equation*}\label{eq:freq-lrinf}
\rc{\sum_{N\subseteq[n]}} \sum_{a\in A} x_{a,N}\cdot\reward_i(a) \ \geq \ \Ex\sigma{}{\lrIf{\vec{r}_i}} \geq \ex_i
\end{equation*}
where the second inequality is due to $\sigma$ being a witness strategy and the first inequality follows from the following:
\begin{align*}
&\sum_{N\subseteq[n]} \sum_{a\in A} x_{a,N}\cdot\reward_i(a)\\
& = \sum_{\substack{N\subseteq[n]\\\Pr\sigma{}{\Omega_N}>0}} \sum_{a \in A} f_N(a)\cdot\reward_i(a) \tag{definition of $x_{a,N}$}\\
& = \sum_{\substack{N\subseteq[n]\\\Pr\sigma{}{\Omega_N}>0}} \sum_{a \in A} \reward_i(a)\cdot 
\lim_{\ell\to\infty} \frac{1}{T_\ell} \sum_{t=1}^{T_\ell} \Pr\sigma{s_0}{A_t=a\mid \Omega_N} \cdot \Pr\sigma{s_0}{\Omega_N} \tag{definition of $f_N$}\\
& = 
\rc{\sum_{\substack{N\subseteq[n]\\\Pr\sigma{}{\Omega_N}>0}} \Pr\sigma{s_0}{\Omega_N} \cdot}
\lim_{\ell\to\infty} \frac{1}{T_\ell} \sum_{t=1}^{T_\ell} 
 \sum_{a \in A} \reward_i(a)\cdot \Pr\sigma{s_0}{A_t=a\mid \Omega_N} \tag{linearity of the limit}\\
& \geq \sum_{\substack{N\subseteq[n]\\\Pr\sigma{}{\Omega_N}>0}} \Pr\sigma{s_0}{\Omega_N} \cdot
\liminf_{T\to\infty} \frac{1}{T} \sum_{t=1}^T  \sum_{a \in A} \reward_i(a)\cdot 
\Pr\sigma{s_0}{A_t=a\mid \Omega_N} \tag{definition of $\liminf$}\\
& = \sum_{\substack{N\subseteq[n]\\\Pr\sigma{}{\Omega_N}>0}} \Pr\sigma{s_0}{\Omega_N} \cdot
\liminf_{T\to\infty} \frac{1}{T} \sum_{t=1}^T 
\Ex{\sigma}{}{\reward_i(A_t) \mid \Omega_N}\tag{definition of the expectation}\\
& \geq \rc{\sum_{\substack{N\subseteq[n]\\\Pr\sigma{}{\Omega_N}>0}} \Pr\sigma{s_0}{\Omega_N} \cdot}
\Ex\sigma{}{\lrIf{\reward_i} \mid \Omega_N}\tag{Fatou's lemma}\\
& = \Ex\sigma{}{\lrIf{\reward_i}}\tag{$\Omega_N$'s partition $\Pat$}
\end{align*}
Although Fatou's lemma (see, e.g.~\cite[Chapter~4, Section~3]{Royden88})
requires the function $\reward_i(A_t)$ be non-negative, we can replace
it with the non-negative function $\reward_i(A_t)-\min_{a\in A}\reward_i(a)$ and add the subtracted
constant afterwards.

\bigskip



\rstart
In order to show that Equations 6 and 7 hold, we prove the following lemma. This lemma is further necessary when relating the $x$-variables to the transient flow in Equation 3 later.
\rstop

\begin{lemma}\label{cl:mec-reach-xaN-}
	For $N\subseteq[n]$ and $C\in\mec$, we have $$\sum_{a\in C}x_{a,N}=\Pr\sigma{s_0}{\Omega_N\cap\Omega_C}\,.$$
\end{lemma}
\begin{proof}
	\rstart
	The proof is trivial for the case with $\Pr\sigma{}{\Omega_N}=0$. Let us now assume $\Pr\sigma{}{\Omega_N}>0$:
	\rstop
	\begin{align*}
	&\sum_{a\in C} x_{a,N}\\	
	& = \rstart \lim_{\ell\to\infty} \frac{1}{T_\ell} \sum_{t=1}^{T_\ell}  \sum_{a \in C} 
	\Pr\sigma{s_0}{A_t=a\mid \Omega_N}\cdot \Pr\sigma{s_0}{\Omega_N}
	\tag{definition of $x_{a,N}$ and $T_\ell$}\\
	& = \rstart \lim_{\ell\to\infty} \frac{1}{T_\ell} \sum_{t=1}^{T_\ell}  \sum_{a \in C} 
	\Big( \Pr\sigma{s_0}{A_t=a\mid \Omega_N\cap\Omega_C}\cdot \frac{\Pr\sigma{s_0}{\Omega_N\cap\Omega_C}}{\Pr\sigma{s_0}{\Omega_N}}
	+ \\&\rstart \qquad\qquad
	 \Pr\sigma{s_0}{A_t=a\mid \Omega_N\setminus\Omega_C}\cdot \frac{\Pr\sigma{s_0}{\Omega_N\setminus\Omega_C}}{\Pr\sigma{s_0}{\Omega_N}}
	\Big) \cdot \Pr\sigma{s_0}{\Omega_N} 
	\tag{partitioning of $\Pat$}\\
	& \rstart = \lim_{\ell\to\infty} \frac{1}{T_\ell} \sum_{t=1}^{T_\ell}  \sum_{a \in C} 
	 \Pr\sigma{s_0}{A_t=a\mid \Omega_N\cap\Omega_C}\cdot \Pr\sigma{s_0}{\Omega_N\cap\Omega_C}
	\tag{$\displaystyle\lim_{T\to\infty} \frac1T\sum_{t=1}^T \Pr\sigma{s_0}{A_t=a\mid \Omega_N\setminus\Omega_C}=0$ for $a\in C$}\\	
	& = \Pr\sigma{s_0}{\Omega_N\cap\Omega_C} \cdot  \rstart\lim_{\ell\to\infty} \frac{1}{T_\ell} \sum_{t=1}^{T_\ell} \rstop  \sum_{a \in C} 
	\Pr\sigma{s_0}{A_t=a\mid \Omega_N\cap\Omega_C}
	\tag{linearity of the limit}\\
	& = \Pr\sigma{s_0}{\Omega_N\cap\Omega_C} \cdot  \rstart \lim_{\ell\to\infty} \frac{1}{T_\ell} \sum_{t=1}^{T_\ell} \rstop
	\Pr\sigma{s_0}{A_t\in C\mid \Omega_N\cap\Omega_C}
	\tag{taking two different actions at time $t$ are disjoint events}\\
	& = \Pr\sigma{s_0}{\Omega_N\cap\Omega_C} 
	\tag{since $A_t\in C$ for all but finitely many $t$ on $\Omega_C$, \rstart see below}
	\rstop
	\end{align*}
\rstart It remains to prove that the last limit is equal to $1$.
We have
\[
1\geq \lim_{\ell\to\infty} \frac{1}{T_\ell} \sum_{t=1}^{T_\ell} 
\Pr\sigma{s_0}{A_t\in C\mid \Omega_N\cap\Omega_C}
= \lim_{\ell\to\infty} \frac{1}{T_\ell} \sum_{t=1}^{T_\ell} 
\Ex\sigma{}{\sum_{a\in C}\kda(A_t)\mid \Omega_N\cap\Omega_C}
\]	
which is by dominated convergence theorem equal to
\[
\Ex\sigma{}{\lim_{\ell\to\infty} \frac{1}{T_\ell} \sum_{t=1}^{T_\ell} 
	\sum_{a\in C}\kda(A_t)\mid \Omega_N\cap\Omega_C} =
\Ex\sigma{}{1}=1
\]
by definition of $\Omega_C$.
\rstop
\end{proof}

\subsection*{Equation 6}

For all $C\in\mec,N\subseteq[n], i\in N$
$$\sum_{a\in C}x_{a, N} \cdot \reward_i(a) \geq \sum_{a\in C}x_{a, N} \cdot \sat_i$$
follows trivially for $\Pr\sigma{s_0}{\Omega_N}=0$, and whenever $\Pr\sigma{s_0}{\Omega_N}>0$ we have
\begin{align*}
&\sum_{a\in C} x_{a,N}\cdot\reward_i(a)\\
& \geq 
\liminf_{T\to\infty} \frac{1}{T} \sum_{t=1}^T  \sum_{a \in C} \reward_i(a)\cdot 
\Pr\sigma{s_0}{A_t=a\mid \Omega_N} \cdot \Pr\sigma{s_0}{\Omega_N} \tag{as above for Eq.~5, by def.\ of $x_{a,N},f_N$, linearity of $\lim$, def.\ of $\liminf$}\\
& = \liminf_{T\to\infty} \frac{1}{T} \sum_{t=1}^T  \sum_{a \in C} 
\reward_i(a)\cdot \Pr\sigma{s_0}{A_t=a\mid \Omega_N\cap\Omega_C}\cdot \Pr\sigma{s_0}{\Omega_N\cap\Omega_C} 
\hspace*{3cm}
\tag{\rstart as above in Lemma~\ref{cl:mec-reach-xaN-}, by partitioning $\Pat$, now with additional factor $\reward_i(a)$} \rstop\\
& \geq \Pr\sigma{s_0}{\Omega_N\cap\Omega_C} \cdot
\Ex\sigma{}{\lrIf{\reward_i} \mid \Omega_N\cap\Omega_C}\tag{as above for Eq.~5, by def.\ of expectation and Fatou's lemma}\\
& \geq \Pr\sigma{s_0}{\Omega_N\cap\Omega_C} \cdot \sat_i \tag{by definition of $\Omega_N$ and $i\in N$}\\
&\rstart =\sum_{a\in C} x_{a,N} \cdot \sat_i \tag{by Lemma~\ref{cl:mec-reach-xaN-}}
\rstop
\end{align*}


\subsection*{Equation 7}

For every $i\in[n]$, by assumption on the strategy $\sigma$ 
$$\sum_{N\subseteq[n]:i\in N}\Pr\sigma{s_0}{\Omega_N} = 
\Pr\sigma{s_0}{\omega\in\Pat\mid\lrIf{\reward}(\omega)_i \geq \sat_i}\geq {\psat}_i$$
and the first term actually equals
\begin{align*}
\sum_{N\subseteq[n]:i\in N}\sum_{a\in A}x_{a,N} 
& =\sum_{N\subseteq[n]:i\in N}\sum_{C\in\mec}\sum_{a\in C}x_{a,N}
\tag{by (\ref{eq:transient-x})}\\
& =\sum_{N\subseteq[n]:i\in N}\sum_{C\in\mec} \Pr\sigma{s_0}{\Omega_N\cap\Omega_C}
\tag{by Lemma \ref{cl:mec-reach-xaN-}}\\
& =\sum_{N\subseteq[n]:i\in N} \Pr\sigma{s_0}{\Omega_N}
\tag{$\Omega_C$'s partition almost all $\Pat$}
\end{align*}

\color{black}


\subsection{Transient behaviour and Equations 1--3}





Now we set the values for $y_\chi$, $\chi\in A\cup (S\times 2^{[n]})$,
and prove that they satisfy Equations 1--3 of $L$ when the values $f_N(a)$ are 
assigned to $x_{a,N}$. 
One could obtain the values $y_\chi$ using the methods of \cite[Theorem 9.3.8]{Puterman}, which requires the machinery of deviation matrices.
Instead, we can first simplify the behaviour of $\sigma$ in the transient part to memoryless using \cite{krish} and then obtain $y_\chi$ directly, like in \cite{EKVY08}, as expected numbers of taking actions. 
To this end, for a state $s$ we define $\reach s$ to be the set of runs that contain $s$. 

Similarly to \cite[Proposition 4.2 and 4.5]{krish},
we modify the MDP $G$ into another MDP $\overline{ G}$ as follows:
For each $s\in S,N\subseteq[n]$, we add a new absorbing state $f_{s,N}$.
The only available action for $f_{s,N}$ leads back to $f_{s,N}$ with probability $1$. 
We also add a new action $a_{s,N}$ to every $s\in S$ for each $N\subseteq[n]$.
The distribution associated with $a_{s,N}$ assigns probability $1$ to $f_{s,N}$.
Finally, we remove all unreachable states.
The construction of \cite{krish} is the same but with only a single value used for $N$.
We denote the copy of each state $s$ of $G$ in $\overline{ G}$ by $\overline{ s}$.

\begin{lemma}\label{cl:strat-reach}
There is a strategy $\overline{ \sigma}$ in $\overline{ G}$ such that for every $C\in\mec$ and $N\subseteq[n]$,
$$\sum_{s\in C} \PrS{\overline{\sigma}}{\overline{ \inits}}{\reach f_{s,N}} = \PrS{\sigma}{\inits}{\Omega_C\cap\Omega_N}\,.$$
\end{lemma}
\begin{proof}
First, we consider an MDP $G'$ created from $G$ in the same way as $\overline{ G}$, but instead of $f_{s,N}$ for each $s\in S,N\subseteq[n]$, we only have a single $f_s$; similarly for actions $a_s$.
As in \cite[Lemma 4.6]{krish}, we obtain a strategy $\sigma'$ in $G'$ such that 
$\sum_{s\in C} \PrS{\sigma'}{\inits'}{\reach f_{s}} = \PrS{\sigma}{\inits}{\Omega_C}$.
We modify $\sigma'$ into $\overline{\sigma}$ as follows. It behaves as $\sigma'$, but instead of taking action $a_s$ with probability $p$, we take each action $a_{s,N}$ with probability 
$p\cdot\frac{\PrS{\sigma}{\inits}{\Omega_C\cap\Omega_N}}{\PrS{\sigma}{\inits}{\Omega_C}}$. (For $\PrS{\sigma}{\inits}{\Omega_C}=0$, we define $\overline{\sigma}$ arbitrarily.)
Then 
\[
\sum_{s\in C} \PrS{\overline{\sigma}}{\overline{ \inits}}{\reach f_{s,N}}
 = \sum_{s\in C}
 \frac{\PrS{\sigma}{\inits}{\Omega_C\cap\Omega_N}}{\PrS{\sigma}{\inits}{\Omega_C}}
 \cdot \PrS{\sigma'}{\inits'}{\reach f_{s}} =
 \PrS{\sigma}{\inits}{\Omega_C\cap\Omega_N}\vspace{-26 pt} 
\]
%
%
%
\end{proof}

By \cite[Theorem~3.2]{EKVY08}, there is a memoryless strategy $\overline{\sigma}$ satisfying the lemma above such that
\begin{align*}
y_a&:=\sum_{t=1}^\infty \PrS{\overline{\sigma}}{\overline{ s}}{A_t=a} \text{\qquad(for actions $a$ preserved in $\overline{ G}$)}\\ 
y_{s,N}&:=\PrS{\overline{\sigma}}{\overline{ s_0}}{\reach f_{s,N}}
\end{align*}
are finite values satisfying Equations 1 and 2, and, moreover,  
$$y_{s,N}\geq \sum_{s\in C} \Pr{\overline{\sigma}}{\overline{s_0}}{\reach f_{s,N}}.$$
By Lemma~\ref{cl:strat-reach} for each $C\in\mec$ we thus have 
$$\sum_{s\in C}y_{s,N} \geq \Pr{\sigma}{s_0}{\Omega_C\cap\Omega_N}$$
and summing up over all $C$ and $N$ we have
$$\sum_{N\subseteq[n]}\sum_{s\in S}y_{s,N}\geq \sum_{N\subseteq[n]} \Pr{\sigma}{s_0}{\Omega_N}$$
where the first term is $1$ by Equation 2, the second term is $1$ by partitioning of $\Pat$, hence they are actually equal and thus
$$\sum_{s\in C}y_{s,N} = \Pr{\sigma}{s_0}{\Omega_C\cap\Omega_N}=\sum_{a\in C}x_{a,N}$$
where the last equality follows by Lemma~\ref{cl:mec-reach-xaN-}, yielding Equation 3.

\section{Proof of Theorem~\ref{thm:main}: Solution to \texorpdfstring{$L$}{L} induces witness strategy}\label{ssec:thm-proof3}

Now we proceed to the proof of the third point of Theorem~\ref{thm:main}.
Let $\bar x_{a,N},\bar y_a,\bar y_{s,N},s\in S,a\in A,N\subseteq[n]$ be a solution to the system $L$. We show how it effectively induces a witness strategy $\sigma$. 

We start with the recurrent part. We prove that even if the flow of Equation 4 is ``disconnected'' we may still play the actions with the exact frequencies $x_{a,N}$ on almost all runs. To formalize the frequency of an action $a$ on a run, recall $\kda$ is the indicator function of $a$, i.e. $\kda(a)=1$ and $\kda(b)=0$ for $a\neq b\in A$. Then $\vfreq_a=\lrIf{\kda}$ defines a vector random variable, indexed by $a\in A$. For the moment, we focus on strongly connected MDPs, i.e.\ the whole MDP is a MEC, and with $N\subseteq[n]$ fixed.

Firstly, we construct a strategy for each ``strongly connected'' part of the solution $x_{a,N}$ and connect the parts, thus averaging the frequencies. This happens at a cost of a small error used for transiting between the strongly connected parts. Secondly, we eliminate this error as we let the transiting happen with measure vanishing over time.


\subsection{\texorpdfstring{$x$}{x}-values and recurrent behaviour}

To begin with, we show that $x$-values describe the recurrent behaviour only:

\begin{lemma}\label{lem:x-mec}
Let $\bar x_{a,N}, a\in A,N\subseteq[n]$ be a non-negative solution to Equation 4 of system $L$. Then for any fixed $N$, $X_N:=\{s,a\mid \bar x_{a,N}>0,a\in\act{s}\}$ is a union of end components.

In particular, $X_N\subseteq\bigcup\mec$, and for every $a\in A \setminus \bigcup\mec$ and $N\subseteq[n]$, we have $\bar x_{a,N}=0$.
\end{lemma}
\begin{proof}
Denoting $ \bar x_{s,N}:=\sum_{a\in\act{s}}\bar x_{a,N}=\sum_{a\in A}\bar x_{a,N}\cdot\delta(a)(s)$ for each $s\in S$, we can write 
$$X_N=\{a\mid x_{a,N}>0\}\cup\{s\mid \bar x_{s,N}>0\}\,.$$
Firstly, we need to show that for all $a\in X_N$, whenever $\delta(a)(s')>0$ then $s'\in X_N$. Since $\bar x_{s',N}\geq \bar x_{a,N}\cdot\delta(a)(s')>0$, we have $s'\in X_N$.

Secondly, let there be a path from $\hat s$ to $\hat t$ in $X_N$. We need to show that there is a path from $\hat t$ to $\hat s$ in $X_N$. Assume the contrary and denote $T\subseteq X_N$ the set of states with no path to $\hat s$ in $X_N$; we assume $\hat t\in T$. We write the path from $\hat s$ to $\hat t$ as $\hat s\cdots s'bt'\cdots \hat t$ where $s'\in X_N\setminus T$ and $t'\in T$. Then $b\in \act{s'}$ and $\delta(b)(t')>0$. Consequently, 

\rstart
\begin{align*}
\sum_{s\in X_N\setminus T}\sum_{a\in A} x_a\cdot \delta(a)(s)
& = \sum_{s\in X_N\setminus T}\sum_{a\in\act{s}} x_a
\tag{by summing Equation 4 over $s\in X_N\setminus T$}\\
&= \sum_{s\in X_N\setminus T}\sum_{a\in\act{s}} \sum_{\overline s\in X_N\setminus T}\! x_a \cdot \delta(a)(\overline s) +\!\! \sum_{s\in X_N\setminus T}\sum_{a\in\act{s}} \sum_{\overline s\in  T} x_a \cdot \delta(a)(\overline s) 
\tag{case split over target states}\\
& > \sum_{s\in X_N\setminus T}\sum_{a\in\act{s}} \sum_{ \overline s\in X_N\setminus T} \!x_a \cdot \delta(a)(\overline s)
\tag{by $\delta(b)(t')>0$}\\
& = \sum_{ \overline s\in X_N\setminus T} \sum_{\substack{a\in\act{s}:\\s\in X_N\setminus T}} \! x_a \cdot \delta(a)(\overline s)
\tag{rearranging}\\
& = \sum_{ \overline s\in X_N\setminus T} \sum_{a\in A}  x_a \cdot \delta(a)(\overline s)
\tag{see below}
\end{align*}
which is a contradiction. The last equality follows by definition of $T$: actions enabled in $T$ cannot lead to $X_N\setminus T$ since from $X_N\setminus T$ there is always a path to $\hat s$ and from $T$ there is no path to $\hat s$.
\rstop
%
%
%
\end{proof}

We thus start with the construction of the recurrent behaviour from $x$-values. For the moment, we restrict to strongly connected MDP and focus on Equation 4 for a particular fixed $N\subseteq[n]$.  Note that for a fixed $N\subseteq[n]$ we have a system of equations equivalent to the form 
\begin{equation}
\sum_{a\in A}x_{a}\cdot \delta(a)(s)=\sum_{a\in \act{s}} x_{a}\qquad \text{for each }s\in S. \label{eq:4a}
\end{equation}
We set out to prove Corollary~\ref{cor:flow-strat}. This crucial observation states that even if the flow of Equation 4 is ``disconnected'', we may still 
play the actions with the exact frequencies $x_{a,N}$ on almost all runs.

\medskip

Firstly, we construct a strategy for each ``strongly connected'' part of the solution $\bar x_a$ (each end-component of $X_N$ of Lemma~\ref{lem:x-mec}). 

\begin{lemma}\label{lem:BSCC}
In a strongly connected MDP $G$, let $\bar x_{a,N}, a\in A$ be a non-negative solution to Equation 4 of system $L$ for a fixed $N\subseteq[n]$ and $\sum_{a\in A}x_{a,N}>0$.
It induces a memoryless strategy $\zeta$ such that for every
BSCCs $D$ of $G^{\zeta}$, every $a\in D\cap A$, and almost all runs in $D$ holds
$$\vfreq_a
 = \frac{\bar{x}_{a,N}}{\sum_{a\in D\cap A} \bar{x}_{a,N}}$$
\rstart i.e. $\Pr{\zeta}{}{\vfreq_a
	= \frac{\bar{x}_{a,N}}{\sum_{a\in D\cap A} \bar{x}_{a,N}}\mid \Omega_D}=1$. Moreover, if all $x_{a,N}$'s are positive then $G^{\zeta}$ is a BSCC and $\vfreq_a$ is almost surely constant. \rstop
\end{lemma}

\begin{proof}
By \cite[Lemma 4.3]{krish} applied to Equation (\ref{eq:4a}), we get a memoryless strategy $\zeta$ such that $\Ex{\zeta}{}{\vfreq_a\mid \Omega_D} = \bar{x}_{a,N}/\sum_{a\in D\cap A} \bar{x}_{a,N}$.
Furthermore, by the ergodic theorem, $\vfreq_a$ returns the same value for almost all runs in $\Omega_D$, hence is equal to  $\Ex{\zeta}{}{\vfreq_a\rstart\mid\Omega_D\rstop}$.
\rstart Finally, if all $x_{a,N}$'s are positive then all actions of $G$ are used. Consequently, since $G$ is strongly connected, $G^\zeta$ is also strongly connected. \rstop
\end{proof}

\smallskip

Secondly, we connect the parts (more end components of Lemma~\ref{lem:x-mec} within one MEC) and thus average the frequencies. This happens at a cost of small error used for transiting between the strongly connected parts.

\begin{lemma}\label{lem:connecting-BSCCs-}
In a strongly connected MDP, let $\bar x_{a,N}, a\in A$ be a non-negative solution to Equation 4 of system $L$ for a fixed $N\subseteq[n]$ and $\sum_{a\in A}x_{a,N}>0$. For every $\varepsilon>0$, there is a memoryless strategy $\zeta^{\varepsilon}$ such that for all $a\in A$ almost surely 
$$\vfreq
_a>\frac {x_{a,N}}  {\sum_{a\in A} x_{a,N}}  -\varepsilon$$
\end{lemma}

\begin{proof}
 We obtain $\zeta^{\varepsilon}$ by a suitable perturbation of the strategy
$\zeta$ from previous lemma in such a way that all actions get positive probabilities and
the frequencies of actions change only slightly, similarly as in \cite[Proposition 5.1, Part 2]{krish}.

There exists an arbitrarily small
(strictly) positive solution $x'_a$ of Equation (\ref{eq:4a}).
Indeed, it suffices to consider a strategy $\tau$ which always
takes the uniform distribution over the actions
in every state and then assign $\Ex{\tau}{}{\vfreq_a}/M$ to
$x_a'$ for sufficiently large $M$.
As the system of Equations (\ref{eq:4a}) is linear and homogeneous,
assigning $\bar{x}_{a,N}+x'_a$ to $x_{a,N}$ also solves this system (and thus Equation 4 as well) \rstart and all values are positive. 
Consequently, \rstop Lemma~\ref{lem:BSCC} gives us a memoryless strategy $\zeta^{\varepsilon}$
satisfying almost surely (with $\pr^{\zeta^\varepsilon}$-probability 1)
$$\vfreq_a=\frac{(\bar{x}_{a,N}+x'_a)}{\sum_{a'\in A}\big(\bar{x}_{a',N}+x'_{a'}\big)}\,.$$
We may safely assume that
\rc{ $\sum_{a\in A} x'_{a}\leq \frac{\varepsilon}{1-\varepsilon} \cdot  \sum_{a\in A} \bar x_{a,N}.$}
\rstart Then almost surely \rstop
\rc{
\begin{align*}
\vfreq_a
&=    \frac{\bar{x}_{a,N}+x'_a}{\sum_{a\in A}(\bar{x}_{a,N}+x'_{a})} \tag{by Lemma~\ref{lem:BSCC}}\\
&> \frac{\bar{x}_{a,N}}{\sum_{a\in A}\bar{x}_{a,N}+\sum_{a\in A}x'_{a}} \tag{by $x'_a>0$}\\
&\geq  \frac{\bar{x}_{a,N}}{\sum_{a\in A}\bar{x}_{a,N}+\frac{\varepsilon}{1-\varepsilon} \cdot  \sum_{a\in A} \bar x_{a,N}} \tag{by  $\sum_{a\in A} x'_{a}\leq \frac{\varepsilon}{1-\varepsilon} \cdot  \sum_{a\in A} \bar x_{a,N}$}\\
 &=   \frac{\bar{x}_{a,N}}{\frac{1}{1-\varepsilon} \cdot \sum_{a\in A} \bar x_{a,N}} \tag{rearranging }\\
 &=  \frac{\bar{x}_{a,N}}{\sum_{a\in A}\bar{x}_{a,N}}- \varepsilon \cdot \frac{\bar{x}_{a,N}}{\sum_{a\in A}\bar{x}_{a,N}} \tag{rearranging }\\
  &\geq  \frac{\bar x_{a,N}}{\sum_{a\in A}\bar{x}_{a,N}}-\varepsilon \tag{by $\frac{\bar{x}_{a,N}}{\sum_{a\in A}\bar{x}_{a,N}} \leq 1$}
\end{align*} }
\end{proof}

\medskip

Thirdly, we eliminate this error as we let the transiting (by $x_a'$) happen with probability vanishing over time.

\begin{lemma}\label{lem:limit-BSCC-}
In a strongly connected MDP, let $\xi_i$ be a sequence of strategies, each with $\vfreq=\vec f^i$ almost surely, and 
such that $\lim_{i\to\infty} \vec{f}^i$ is well defined. Then there is Markov strategy $\xi$ such that almost surely 
$$\vfreq=\lim_{i\to\infty} \vec f^i\,.$$
\end{lemma}

\begin{proof}
This proof very closely follows the computation in \cite[Proposition 5.1, Part ``Moreover'']{krish}, but for general $\xi_i$.

Given $a\in A$, let $\lfa:=\lim_{i\to\infty} \vec f^i_a$.
By definition of limit and the assumption that $\vfreq_a=\lrIf{\kda}$ is almost surely equal to $\vec f^i_a$ for each $\xi_i$, there is a subsequence $\xi_j$ of the sequence $\xi_i$ such that 
$\Pr{\xi_j}{s}{\lrIf{\kda} \ge \lfa - 2^{-j-1}}=1$.
Note that for every $j\in \Nset$ there is $\stabi{j}\in\Nset$ such that for all
$a\in A$ and $s\in S$ we get
\[
 \Pr{\xi_{j}}{s}{\inf_{T \ge \stabi{j}} \frac{1}{T}\sum_{t=0}^T \kda(A_t) \ge  \lfa  - 2^{-j}} \ge 1-2^{-j}
 .
\]

Let us consider a sequence $n_0,n_1,\ldots$ of numbers where
$n_j\ge \kappa_j$ and 
\begin{align}
\frac{\sum_{k<j} \rstart n_k \rstop}{n_j}&\leq 2^{-j} \label{eq:k1}\\
\frac{\stabi{j+1}}{n_j}&\leq 2^{-j} \label{eq:k2}
\end{align}  
We define $\xi$ to behave as $\xi_1$ for the first $n_1$ steps, then as $\xi_2$ for the
next $n_2$ steps, etc. In
general, denoting by $N_j$ the sum $\sum_{k<j} n_k$, the strategy
$\xi$ behaves as $\xi_j$ between the $N_j$-th step (inclusive) and
$N_{j+1}$-th step (non-inclusive). Note that such strategy is a Markov strategy.

Let us give some intuition behind $\xi$. The numbers in the sequence $n_0,n_1,\ldots$ grow rapidly
so that after $\xi_j$ is simulated for $n_j$ steps, the part of the history when $\xi_k$ for $k<j$ were simulated
becomes relatively small and has only minor impact on the current average reward (this is
ensured by the condition $\frac{\sum_{k<j} n_k}{n_j}\leq 2^{-j}$). This gives us that almost every
run has infinitely many prefixes on which the average reward w.r.t. $\kda$ is arbitrarily close to
$\rc{ \lfa }$ infinitely often. To get that $\rc{ \lfa }$ is also the long-run average reward,
one only needs to be careful when the strategy $\xi$ ends behaving
as $\xi_j$ and starts behaving as $\xi_{j+1}$, because then up to the $\kappa_{j+1}$ steps we have
no guarantee that the average reward is close to $\rc{ \lfa }$. This part is taken
care of by picking $n_j$ so large that the contribution (to the average reward) of the $n_j$ steps according to
$\xi_j$ prevails over fluctuations introduced by the first
$\kappa_{j+1}$ steps according to $\xi_{j+1}$ (this is ensured by the condition 
$\frac{\stabi{j+1}}{n_j}\leq 2^{-j}$).

Let us now prove the correctness of the definition of $\xi$ formally.
We prove that almost all runs $\omega$ of $G^{\xi}$ satisfy 
\[
\liminf_{T\rightarrow\infty} \frac{1}{T}\sum_{t=0}^T \kda(A_t(\omega))\geq  \rc{\lfa}\,.
\]
Denote by $E_k$ the set of all runs $\omega=s_0a_0s_1a_1\cdots$ of $G^{\xi}$ such that for some $\stabi{k}\le d \le n_k$ we have
\[
 \frac{1}{d}\sum_{j=N_j}^{N_j+d\rstart -1 \rstop}\kda(a_{k})\quad < \quad \rc{\lfa} - 2^{-k}
 .
\]
We have $\Pr{\xi}{s_0}{E_j} \le 2^{-j}$ and thus
$\sum_{j=1}^{\infty} \Pr{\xi}{s_0}{E_j}=\frac{1}{2}<\infty$ holds. By the
Borel-Cantelli lemma \cite{Royden88}, almost surely only
finitely many of $E_{j}$ take place. Thus, almost
every run $\omega=s_0a_0s_1a_1\cdots$ of $G^{\xi}$ satisfies the following: there is $\ell$ such that for all $j\ge\ell$ and all $\stabi{j}\le d \le n_j$
we have that
\begin{equation}
 \frac{1}{d}\sum_{k=N_j}^{N_j+d\rstart -1 \rstop}\kda(a_{k})\quad \geq \quad \rc{ \lfa } - 2^{-j}\,.\label{eq:k3}
\end{equation}
Consider $T\in \Nset$ such that $N_{j}\leq T<N_{j+1}$ where $j>\ell$. 
Below, we prove the following inequality
\begin{equation}\label{eq:fin}
\frac{1}{T} \sum_{t=0}^T  \kda(a_t)\quad \geq\quad (\rc{ \lfa }-2^{\rstart 1 \rstop-j})(1-2^{1-j})\,.
\end{equation}
Taking the limit of (\ref{eq:fin}) where $T$ (and thus also $j$) 
goes to $\infty$, we obtain 
\[
 \vfreq_a(\omega)=\liminf_{T\rightarrow\infty} \frac{1}{T}\sum_{t=0}^{T}\kda(a_t) \ge \liminf_{j\rightarrow\infty} (\rc{ \lfa }-2^{\rstart 1 \rstop-j})(1-2^{1-j}) = \rc{ \lfa }=\lim_{i\to\infty} \vec f^i_a
\]
yielding the lemma. It remains to prove (\ref{eq:fin}).
First, note that
\[
\frac{1}{T} \sum_{t=0}^T  \kda(a_t) \quad \geq\quad
\frac{1}{T}\sum_{t=N_{j-1}}^{N_{j}-1} \kda(a_t)+\frac{1}{T}\sum_{t=N_j}^{T} \kda(a_t)
\]
and that by (\ref{eq:k3})
\begin{align*}
\frac{1}{T}  \sum_{t=N_{j-1}}^{N_{j}-1}  \kda(a_t) & = \frac{1}{n_j} \sum_{t=N_{j-1}}^{N_{j}-1} \kda(a_t)\cdot \frac{n_j}{T} 
   \geq (\rc{ \lfa }-2^{\rstart 1 \rstop-j}) \frac{n_j}{T}
\end{align*}
which gives
\begin{equation}\label{eq:main}
\frac{1}{T} \sum_{t=0}^T  \kda(a_t) \ \geq \  
(\rc{ \lfa }-2^{\rstart 1 \rstop-j}) \frac{n_j}{T}\, +\, \frac{1}{T}\sum_{t=N_j}^{T} \kda(a_t)
.
\end{equation}
Now, we distinguish two cases.
First, if $T-N_j\leq \stabi{j+1}$, then 
\begin{align*}
 \frac{n_j}{T}\geq\frac{n_j}{N_j+\stabi{j+1}}=\frac{n_j}{N_{j-1}+n_j +\stabi{j+1}} = 1-\frac{N_{j-1}+\stabi{j+1}}{N_{j-1}+n_j+\stabi{j+1}} \geq (1-2^{1-j})
\end{align*}
by (\ref{eq:k1}) and (\ref{eq:k2}). Therefore, by (\ref{eq:main}),
\[
\frac{1}{T} \sum_{t=0}^T  \kda(a_t)  \quad \geq \quad (\rc{ \lfa }-2^{\rstart 1 \rstop-j})(1-2^{1-j})
.
\]
Second, if $T-N_j\geq \stabi{j+1}$, then 
\begin{align*}
\frac{1}{T}\sum_{t=N_j}^{T} \kda(a_t) & =\frac{1}{T-N_j\rstart +1 \rstop}\sum_{t=N_j}^{T} \kda(a_t)\cdot \frac{T-N_j\rstart +1 \rstop}{T} \\
 &  \geq (\rc{ \lfa }-2^{-j})\left(1-\frac{N_{j-1}+n_j}{T}\right) \tag{by (\ref{eq:k3})}\\
 &  \geq (\rc{ \lfa }-2^{-j})\left(1-2^{-j}-\frac{n_j}{T}\right) \tag{by (\ref{eq:k1})}
\end{align*}
and thus, by (\ref{eq:main}),
\begin{align*}
\frac{1}{T} \sum_{t=0}^T  \kda(a_t) & \geq 
(\rc{ \lfa }-2^{\rstart 1 \rstop-j}) \frac{n_j}{T} + (\rc{ \lfa }-2^{-j\rstart + \rstop 1})\left(1-2^{-j}-\frac{n_j}{T}\right) \\
& \geq (\rc{ \lfa }-2^{\rstart 1 \rstop-j}) \left(\frac{n_j}{T} + \left(1-2^{-j}-\frac{n_j}{T}\right) \right) \\
& \geq (\rc{ \lfa }-2^{\rstart 1 \rstop-j}) (1-2^{\rstart 1 \rstop-j}) 
\end{align*}
which finishes the proof of (\ref{eq:fin}).
\end{proof}

Now we know that strategies within an end component can be merged into a strategy with frequencies corresponding to the solution of Equation 4 for each fixed $N$.

\begin{corollary}\label{cor:flow-strat}
For a strongly connected MDP, let $\bar x_{a,N}, a\in A$ be a non-negative solution to Equation 4 of system $L$ for a fixed $N\subseteq[n]$ and $\sum_{a\in A}x_{a,N}>0$. Then there is Markov strategy $\xi_N$ such that for each $a\in A$ almost surely 
$$\vfreq
_a
=\frac {x_{a,N}}   {\sum_{a\in A} x_{a,N}} \,. $$ 
\end{corollary}
\begin{proof}
The strategy $\xi_N$ is constructed by Lemma~\ref{lem:limit-BSCC-} taking $\xi_i$ to be $\zeta^{1/i}$ from Lemma~\ref{lem:connecting-BSCCs-}. 
\end{proof}

\begin{remark}
Note that \rstart using such strategy, \rstop all actions and states in \rstart the single \rstop MEC are visited infinitely often. (This will be later useful for the strategy complexity analysis.)
\end{remark}

Since the fraction is independent of the initial state of the MDP, the frequency is almost surely the same also for all initial states.
The reward of $\xi_N$ is almost surely 
$$\lrIf{\reward}(\omega)=\frac{\sum_{a}\bar x_{a,N}\cdot\reward(a)}{\sum_{a}\bar x_{a,N} }\,.$$
{\sloppy
When the MDP is not strongly connected, we obtain such $\xi_N$ in each
MEC $C$ with $\sum_{a\in C}\bar x_{a,N}>0$ and the respective reward
of almost all runs in $C$ is thus
}
\begin{align}
&\Ex{\xi_N}{}{\lrIf{\reward}\mid \Omega_C}=\frac{\sum_{a\in C\cap A}\bar x_{a,N}\cdot\reward(a)}{\sum_{a\in C\cap A}\bar x_{a,N} }\label{eq:xiN-reward-ex}\,.
\end{align}
Moreover, the long-run average reward is the same for almost all runs, which is a stronger property than in \cite[Lemma 4.3]{krish}, which does not hold for the induced strategy there. We need this property here in order to combine the satisfaction requirements.
\begin{align}
&\Pr{\xi_N}{}{\lrIf{\reward}=\frac{\sum_{a\in C\cap A}\bar x_{a,N}\cdot\reward(a)}{\sum_{a\in C\cap A}\bar x_{a,N}}\mid \Omega_C}=1 \label{eq:xiN-reward-pr}\,.
\end{align}

\subsection{\texorpdfstring{$y$}{y}-values and transient behaviour}

We now consider the transient part of the solution that plays $\xi_N$'s with various probabilities. Let ``$\text{switch to }\xi_N\text{ in }C$'' denote the event that a strategy updates its memory, while in $C$, into such an element that it starts playing exactly as $\xi_N$. We can stitch all $\xi_N$'s together as follows:

\begin{lemma}\label{lem:transient}
Let $\xi_N,N\subseteq[n]$ be strategies. Then every non-negative solution $\bar y_a,\bar y_{s,N}$, $a\in A,s\in S,N\subseteq[n]$ to Equation 1 effectively induces a strategy $\sigma$ such that 
$$\pr^\sigma[\text{switch to }\xi_N\text{ in }s]=\bar y_{s,N}$$
and $\sigma$ is memoryless before the switch.
\end{lemma}
\begin{proof}
The idea is similar to \cite[Proposition 4.2, Step 1]{krish}. However, instead of switching in $s$ to $\xi$ with some probability $p$, here we have to branch this decision and switch to $\xi_N$ with probability $p\cdot\frac{\bar y_{s,N}}{\sum_{N\subseteq[n]}\bar y_{s,N}}$. 

Formally,
\rc{for every MEC $C$ of $G$, we denote the
number $\sum_{s\in C} \sum_{N\subseteq[n]} \bar{y}_{s,N}$ by $y_C$. 
According to the Lemma 4.4 of~\cite{krish} we have a
stochastic-update strategy $\vartheta$ which stays eventually in each MEC
$C$ with probability $y_C$. }

Then the strategy $\overline{\sigma}$ works as follows. It plays according to $\vartheta$ until a
 BSCC of $G^{\vartheta}$ is reached. This means that every possible 
 continuation of the path stays in the current MEC $C$ of $G$.
Assume that $C$ has states $s_1,\ldots,s_k$.  
At this point, the strategy
$\overline{\sigma}$ changes its behaviour as follows: First, it strives to reach $s_1$ with probability one. Upon
reaching $s_1$, it chooses randomly with probability
\rc{$\frac{\bar{y}_{s_1,N}}{y_C}
$ to behave as $\xi_N$ forever, or otherwise to
follow on to $s_2$.}
If the strategy $\overline{\sigma}$ chooses to go on to $s_2$, it strives 
to reach $s_2$ with probability one. Upon reaching $s_2$,
it chooses with probability 
\rc{$\frac{\bar{y}_{s_2,N}}{y_C-\sum_{N\subseteq[n]}\bar{y}_{s_1,N}}$} to behave as $\xi_N$ forever, or to follow on to $s_3$, and so on, till 
$s_k$. That is, the probability of switching to $\xi_N$ in $s_i$ is 
\rc{$$\frac{\bar{y}_{s_i,N}}{y_C-\sum_{j=1}^{i-1}\sum_{N\subseteq[n]} \bar{y}_{s_j,N}}\,.$$}

Since $\vartheta$ stays in a MEC $C$ with probability $y_C$, the
probability that the strategy $\overline{\sigma}$ switches to $\xi_N$ 
in $s_i$ is equal to $\bar{y}_{s_i,N}$. 
Further, as in \cite{krish} we can transform the part of $\overline{\sigma}$ before switching to $\xi_N$ to a memoryless strategy and thus get strategy $\sigma$.
\end{proof}

\begin{corollary}\label{cl:m-reach-strat}
Let $\xi_N,N\subseteq[n]$ be strategies. 
Then every non-negative solution \\
$\bar y_a,\bar y_{s,N},\bar x_{a,N},a\in A,s\in S,N\subseteq[n]$ to Equations 1 and 3 effectively induces a strategy $\sigma$ such that for every MEC $C$
$$\pr^\sigma[\text{switch to }\xi_N\text{ in }C]=\sum_{a\in C\cap A} x_{a,N}$$
and $\sigma$ is memoryless before the switch.
\end{corollary}
\begin{proof}
By Lemma~\ref{lem:transient} and Equation 3.
\end{proof}

\subsection{Proof of witnessing}

We now prove that the strategy $\sigma$ of Corollary~\ref{cl:m-reach-strat} with $\xi_N,N\subseteq[n]$ of Corollary~\ref{cor:flow-strat} is indeed a witness strategy. 
Note that existence of $\xi_N$'s depends on the sums of $\bar x$-values being positive. This follows by Equation 2 and 3. 
We evaluate the strategy $\sigma$ as follows: 
\begin{align*}
&\Ex{\sigma}{}{\lrIf{\reward}}\\
&= \sum_{C\in \mec} \sum_{N\subseteq[n]} \Pr{\sigma}{}{\text{switch to }\xi_N\text{ in }C}\cdot  \Ex{\xi_N}{}{\lrIf{\reward}\mid \Omega_C}
\tag{by Equation 2, $\displaystyle\sum_{N\subseteq [n]}\Pr{\sigma}{}{\text{switch to }\xi_N}=1$}\\
&= \sum_{C\in \mec}\sum_{N\subseteq[n]} \Big(\sum_{a\in C\cap A}\bar x_{a,N}\Big)\cdot \Ex{\xi_N}{}{\lrIf{\reward}\mid \Omega_C} \tag{by Corollary~\ref{cl:m-reach-strat}}\\
&= \sum_{C\in \mec} \sum_{\substack{N\subseteq[n]:\\\sum_{a\in C\cap A}\bar x_{a,N}>0}} \Big(\sum_{a\in C\cap A}\bar x_{a,N}\Big)\cdot 
\Big(\sum_{a\in C\cap A}\bar x_{a,N}\cdot\reward(a)/\sum_{a\in C\cap A}\bar x_{a,N}\Big) \tag{by (\ref{eq:xiN-reward-ex})}\\
&= \sum_{N\subseteq[n]}\sum_{C\in \mec} \sum_{a\in C\cap A}\bar x_{a,N} \cdot\reward(a)\\
&= \sum_{N\subseteq[n]} \sum_{a\in A\cap\bigcup\mec}\bar x_{a,N}\cdot\reward(a)\\
&= \sum_{N\subseteq[n]} \sum_{a\in A}\bar x_{a,N}\cdot\reward(a) \tag{by Lemma~\ref{lem:x-mec}}\\
&\geq \ex \tag{by Equation 5}
\end{align*}

and for each $i\in[n]$ we have
\begin{align*}
&\Pr{\sigma}{}{\lrIf{\reward}_i\geq\sat_i}= \\
&\sum_{C\in \mec}\sum_{N\subseteq[n]} \Pr{\sigma}{}{\text{switch to }\xi_N\text{ in }C}\cdot \Pr{\xi_N}{}{\lrIf{\reward}_i\geq\sat_i\mid \Omega_C}
\tag{by Equation 2, $\displaystyle\sum_{N\subseteq [n]}\Pr{\sigma}{}{\text{switch to }\xi_N}=1$} \\
&= \sum_{C\in \mec}\sum_{N\subseteq[n]} \Big(\sum_{a\in C\cap A}\bar x_{a,N}\Big)\cdot \Pr{\xi_N}{}{\lrIf{\reward}_i\geq\sat_i\mid \Omega_C}
\tag{by Corollary~\ref{cl:m-reach-strat}} \\
&=\sum_{C\in \mec} \sum_{\substack{N\subseteq[n]:\\\sum_{a\in C\cap A}\bar x_{a,N}>0}} \Big(\sum_{a\in C\cap A}\bar x_{a,N}\Big)\cdot \Pr{\xi_N}{}{\sum_{a\in C\cap A}\bar x_{a,N}\cdot\reward(a)_i\big/\sum_{a\in C\cap A}\bar x_{a,N} \geq\sat_i}\tag{by (\ref{eq:xiN-reward-pr})}\\
&\geq \sum_{C\in \mec} \sum_{\substack{i\in N\subseteq[n]:\\\sum_{a\in C\cap A}\bar x_{a,N}>0}} \Big(\sum_{a\in C\cap A}\bar x_{a,N}\Big)\cdot \Pr{\xi_N}{}{\sum_{a\in C\cap A}\bar x_{a,N} \cdot\sat_i/\sum_{a\in C\cap A}\bar x_{a,N} \geq\sat_i}\tag{by Equation 6}\\
&= \sum_{i\in N\subseteq[n]}\sum_{C\in \mec} \sum_{a\in C\cap A}\bar x_{a,N} \\
&= \sum_{i\in N\subseteq[n]}\sum_{a\in A\cap\bigcup\mec}\bar x_{a,N}\\
&= \sum_{i\in N\subseteq[n]}\sum_{a\in A} \bar x_{a,N} \tag{by Lemma~\ref{lem:x-mec}}\\
&\geq \psat_i \tag{by Equation 7}
\end{align*}

\begin{remark}\label{construction}
The proof of the corresponding claim for $\varepsilon$-witness strategies proceeds as above. We get that the strategy $\sigma$ of Corollary~\ref{cl:m-reach-strat} with $\zeta^\varepsilon_N,N\subseteq[n]$ of Lemma~\ref{lem:connecting-BSCCs-} is an $\varepsilon$-witness strategy.
\QEE
\end{remark}



\section{Algorithmic complexity}\label{sec:algc}

In this section, we discuss the solutions to and complexity of all the introduced problems.

\subsection{Solution to \textbp{(multi-quant-conjunctive)}}

As we have seen, there are $\mathcal O(|G|\cdot n)\cdot2^n$ variables in the linear program $L$.
By Theorem~\ref{thm:main}, the upper bound on the algorithmic time complexity is polynomial in the number of variables in system $L$.
Hence, 
the
realizability problem for \textbp{(multi-quant-conjunctive)} can be decided in time polynomial in $|G|$ and exponential in $n$.

\subsection{Solution to \textbp{(multi-quant-joint)} and the special cases}\label{joint}

In order to decide \textbp{(multi-quant-joint)}, the only subset of runs to exceed the probability threshold is the set of runs with all long-run rewards exceeding their thresholds, i.e.\ $\Omega_{[n]}$ (introduced in Section~\ref{subs:47}). The remaining runs need not be partitioned and can be all considered to belong to $\Omega_\emptyset$ without violating any constraint. Intuitively, each $x_{a,\emptyset}$ now stands for the original sum $\sum_{N\subseteq[n]: N\neq[n]}x_{a,N}$; similarly for $y$-variables.
Consequently, the only non-zero variables of $L$ indexed by $N$ satisfy $N=[n]$ or $N=\emptyset$. 
The remaining variables can be left out of the system.

Requiring all variables $y_{a},y_{s,N},x_{a,N}$ for $a\in A,s\in S,\rstart N\in\{\emptyset,[n]\} \rstop$ be non-negative, the program is the following:
\begin{enumerate}
 \item transient flow: for $s\in S$
 $$\kd{s_0}(s)+\sum_{a\in A}y_a\cdot \delta(a)(s)=\sum_{a\in \act{s}} y_a+ y_{s,\emptyset}+y_{s,[n]}$$
 \item almost-sure switching to recurrent behaviour:
 $$\sum_{s\in C}y_{s, \emptyset}+y_{s,[n]}=1$$
 \item probability of switching in a MEC is the frequency of using its actions: for $C\in \mec$
 $$\sum_{s\in C}y_{s,{\emptyset}}=\sum_{a\in C}x_{a,\emptyset}$$
 $$\sum_{s\in C}y_{s,{[n]}}=\sum_{a\in C}x_{a,[n]}$$
 \item recurrent flow: for $s\in S$
 $$\sum_{a\in A}x_{a, \emptyset}\cdot \delta(a)(s)=\sum_{a\in \act{s}} x_{a, \emptyset}$$
 $$\sum_{a\in A}x_{a, [n]}\cdot \delta(a)(s)=\sum_{a\in \act{s}} x_{a, [n]}$$
 \item expected rewards:
 $$\sum_{a\in A}\Big(x_{a, \emptyset}+x_{a,[n]}\Big)\cdot\reward(a)\geq\ex$$
 \item commitment to satisfaction: for $C\in\mec$ and $i\in[n]$
 $$\sum_{a\in C}x_{a, [n]} \cdot \reward(a)_{ i}\geq \sum_{a\in C}x_{a, [n]} \cdot \sat_{ i}$$
 \item satisfaction: 
 $$\sum_{a\in A}x_{a,[n]}\geq \psatscalar$$
\end{enumerate}

Since there are now $\mathcal O(|G|\cdot n)$ variables, the problem as well as its special cases can be decided in polynomial time.

Similarly, for \textbp{(mono-quant)} it is sufficient to consider $N=[n]=\{1\}$ and $N=\emptyset$ only. Consequently, for \textbp{(multi-qual)} $N=[n]$, and for \textbp{(mono-qual)} $N=[n]=\{1\}$ are sufficient, thus the index $N$ can be removed completely.




\begin{theorem}\label{thm:timepoly}
The \textbp{(multi-quant-joint)} realizability problem (and thus also all its special cases) can be decided in time polynomial in $|G|$ and $n$.
\QED
\end{theorem}

\subsection{Solution to \textbp{(multi-quant-conjunctive-joint)}}\label{sec:conjunctivejoint}
The linear program for this ``combined'' problem can be easily derived from the program $L$ in Fig.~\ref{fig:lp-new} as follows.

The first step consists in splitting the recurrent flow into two parts, $\mathit{yes}$ and $\mathit{no}$
Requiring all variables be non-negative, the program is the following:

\begin{enumerate}
 \item transient flow: for $s\in S$
 $$\kd{s_0}(s)+\sum_{a\in A}y_a\cdot \delta(a)(s)=\sum_{a\in \act{s}} y_a+ {\sum_{N\subseteq[n]}(y_{s,N,\mathit{yes}}}+y_{s,N,\mathit{no}})$$
 \item almost-sure switching to recurrent behaviour:
 $$\sum_{\substack{s\in C\in \mec\\{N\subseteq[n]}}}(y_{s, N,\mathit{yes}}+y_{s,N,\mathit{no}})=1$$
 \item probability of switching in a MEC is the frequency of using its actions: for $C\in \mec,{N\subseteq[n]}$
 $$\sum_{s\in C}y_{s,{N},\mathit{yes}}=\sum_{a\in C}x_{a,N,\mathit{yes}}$$
 $$\sum_{s\in C}y_{s,{N},\mathit{no}}=\sum_{a\in C}x_{a,N,\mathit{no}}$$
 \item recurrent flow: for $s\in S,{N\subseteq[n]}$
 $$\sum_{a\in A}x_{a, N,\mathit{yes}}\cdot \delta(a)(s)=\sum_{a\in \act{s}} x_{a, N,\mathit{yes}}$$
 $$\sum_{a\in A}x_{a, N,\mathit{no}}\cdot \delta(a)(s)=\sum_{a\in \act{s}} x_{a, N,\mathit{no}}$$
 \item expected rewards:
 $$\sum_{\substack{a\in A,\\ N\subseteq[n]}}(x_{a, N,\mathit{yes}}+x_{a,N,\mathit{no}})\cdot\reward(a)\geq\ex$$
 \item commitment to satisfaction: for $C\in\mec$, $N\subseteq[n]$, $i\in N$
 $$\sum_{a\in C}x_{a, N,\mathit{yes}} \cdot \reward(a)_{ i}\geq \sum_{a\in C}x_{a, N,\mathit{yes}} \cdot \sat_{ i}$$
 $$\sum_{a\in C}x_{a, N,\mathit{no}} \cdot \reward(a)_{ i}\geq \sum_{a\in C}x_{a, N,\mathit{no}} \cdot \sat_{ i}$$
 \item satisfaction: for $i\in[n]$
 $$\sum_{\substack{a\in A,\\ N\subseteq[n]:i\in N}}x_{a, N,\mathit{yes}}+x_{a,N,\mathit{no}}\geq\psat_i$$
\end{enumerate}
Note that this program has the same set of solutions as the original program, considering substitution $\alpha_{\beta,N}=\alpha_{\beta,N,\mathit{yes}}+\alpha_{\beta,N,\mathit{no}}$.

\medskip

The second step consists in using the ``$\mathit{yes}$'' part of the flow for ensuring satisfaction of the (\ref{eq:krish-sat}) constraint. Formally, we add the following additional equations (of type 6 and 7, respectively):

\begin{enumerate}
\item[$(\widetilde 6)$]
$$\sum_{a\in C}x_{a,N,\mathit{yes}} \cdot \reward(a)_{i}\geq \sum_{a\in C}x_{a,N,\mathit{yes}} \cdot \widetilde\sat_{ i} \text{\qquad for }i\in [n] \text{ and }N\subseteq[n]$$
\item[$(\widetilde 7)$] 
$$\sum_{\substack{a\in A\\N\subseteq[n]}}x_{a,N,\mathit{yes}}\geq \widetilde\psatscalar$$
\end{enumerate}

Note that the number of variables is double that for \textbp{(multi-quant-conjunctive)}.
Therefore, the complexity remains essentially the same: 
\begin{corollary}\label{cor:conjjoint}
The algorithmic complexity for the \textbp{(multi-quant-conjuctive-joint)} is polynomial in the size of the MDP and exponential in $n$.
\QED
\end{corollary}

\begin{remark}\label{rem:cj-struct}
The strategies for the case of \textbp{(multi-quant-conjunctive-joint)} are very similar to that of \textbp{(multi-quant-conjunctive)}.
Indeed, the structure of the constructed ($\varepsilon$-)witness strategies is the same: the memoryless strategy for reaching the desired MECs is followed by a stochastic-update switch to strategies for the recurrent behaviour. The only difference is the following.
($\varepsilon$-)witness strategies for \textbp{(multi-quant-conjunctive)} switch to strategies $\xi_N$ (or $\zeta^\varepsilon_N$), each given by values of $x$-variables indexed by a fixed $N\subseteq[n]$.
In contrast, strategies for \textbp{(multi-quant-conjunctive-joint)} switch to strategies $\xi_{N,b}$ (or $\zeta^\varepsilon_{N,b}$), each given by values of $x$-variables indexed by a fixed $N\subseteq[n]$ \emph{and} $b\in\{\mathit{yes},\mathit{no}\}$.
\QEE
\end{remark}

Furthermore, we can also allow multiple constraints, i.e.\ more (\ref{eq:krish-sat}) constraints or more (\ref{eq:SAT}), thus specifying probability thresholds for more value thresholds for each reward. Then instead of subsets of $[n]$ as so far, we consider subsets of the set of all constraints. The number of variables is then exponential in the number of constraints rather than just in the dimension of the rewards.

\subsection{Hardness}

The \textbp{(multi-quant-conjunctive-joint)} problem is also of significant theoretical interest since we can also prove the following hardness result:

\begin{theorem}\label{thm:hard}
The \textbp{(multi-quant-conjunctive-joint)} problem is NP-hard (even without the (\ref{eq:EXP}) constraint).
\end{theorem}

\begin{proof}

We proceed by reduction from SAT. Let $\varphi$ be a formula with the set of clauses $C=\{c_1,\ldots,c_k\}$ over atomic propositions $Ap=\{a_1,\ldots,a_p\}$. We denote $\overline{Ap}=\{\overline{a_1},\ldots,\overline{a_p}\}$ the literals that are negations of the atomic propositions.

We define an MDP $G_\varphi=(S,A,\mathit{Act},\delta,\inits)$ as follows:
\begin{itemize}
\item $S=\{s_i\mid i\in[p]\}$,
\item $A=Ap\cup\overline{Ap}$,
\item $\act{s_i}=\{a_i,\overline{a_i}\}$ for $i\in[p]$,
\item $\delta(a_i)(s_{i+1})=1$ and $\delta( \overline{a_i})(s_{i+1})=1$ (actions are assigned Dirac distributions),
\item $\inits=s_1=s_{p+1}$.
\end{itemize}
\rstart
The constructed MDP is illustrated in Fig.~\ref{fig:mdp-phi}.
Intuitively, a run in $G_\varphi$ repetitively chooses a valuation. 

\begin{figure}[ht]
	\begin{tikzpicture}[scale=1.5]
	\node[state,initial,initial text=] (s)at(0,0){$s_1$};
	\node[state] (t)at(1,1){$s_2$};
	\node (u)at(2.5,0.5){$\cdots$};
	\node[state] (v)at(1,-1){$s_p$};
	\node (w)at(2.5,-0.5){$\cdots$};
	\node (x)at(2.5,0){$\vdots$};
	
	\path[->]
	(s) edge[bend left=15pt] node[above]{$a_1$} (t)
	(s) edge[bend right=15pt] node[below]{$\overline a_1$} (t)
	(t) edge[bend left=15pt] node[above]{$a_2$} (u)
	(t) edge[bend right=15pt] node[below]{$\overline a_2$} (u)
	(v) edge[bend left=15pt] node[below]{$a_p$} (s)
	(v) edge[bend right=15pt] node[above]{$\overline a_p$} (s)
	(w) edge[bend left=15pt] node[below]{$a_{p-1}$} (v)
	(w) edge[bend right=15pt] node[above]{$\overline a_{p-1}$} (v)
	;	
	\end{tikzpicture}
	\caption{MDP $G_\varphi$}
	\label{fig:mdp-phi}
\end{figure}
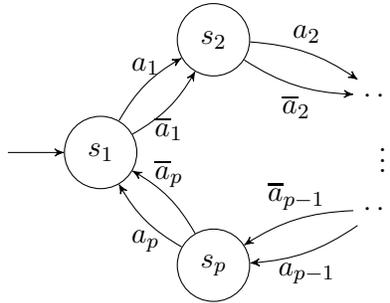
\rstop

We define the dimension of the reward function to be $n=k+2p$. We index the components of vectors with this dimension by  $C\cup Ap\cup\overline{Ap}$.
The reward function is defined for each $\ell\in A$ as follows: \medskip
\begin{itemize}
\item $\reward(\ell)(c_i)=
\begin{cases}
1 & \text{if } \ell\models c_i\\
0 & \text{if } \ell\not\models c_i
\end{cases}$\medskip
\item $\reward(\ell)(a_i)=\kd{a_i}$
\medskip
\item $\reward(\ell)(\overline{a_i})=\kd{\overline{a_i}}$
\end{itemize}
Intuitively, we get a positive reward for a clause when it is guaranteed to be satisfied by the choice of a literal.
The latter two items simply count the number of uses of a~literal; 
\rstart thus $\lrIf{\reward}_a=\vfreq_a$. \rstop

The realizability problem instance $R_\varphi$ is then defined by a conjunction of the following (\ref{eq:SAT}) and (\ref{eq:krish-sat}) constraints:
\begin{align*}
 &\Pr{\sigma}{}{\lrIf\reward_{\ell}\geq \rstart \frac1p\rstop}\geq \frac12 \text{\qquad for each }\ell\in Ap\cup\overline{Ap} \tag{conjunctive-S} \\
 &\Pr{\sigma}{}{\bigwedge_{c\in C}\lrIf\reward_{c}\geq \rstart \frac1p\rstop}\geq \frac12 \tag{joint-S}
\end{align*}
Intuitively, (conjunctive-S) ensures that almost all runs choose, for each atomic proposition, either the positive literal with frequency 1, or the negative literal with frequency 1; in other words, it ensures that the choice of valuation is consistent  within the run almost surely. Indeed, since the choice between $a_i$ and $\overline{ a_i}$ happens every $\rstart p\rstop$ steps, runs that mix both with positive frequency cannot exceed the value threshold $1/\rstart p\rstop$. Therefore, half of the runs must use only $a_i$, half must use only $\overline{ a_i}$. Consequently, almost all runs choose one of them consistently.

Further, (joint-S) on the top ensures that there is a (consistent) valuation that satisfies all the clauses. Moreover, we require that this valuation is generated with probability at least $1/2$. Actually, we only need probability strictly greater than $0$.

\medskip

We now prove that $\varphi$ is satisfiable if and only if the problem instance defined above on MDP $G_\varphi$ is realizable.

``Only if part'': Let $\nu\subseteq Ap\cup\overline{Ap}$ be a satisfying valuation for $\varphi$. We define $\sigma$ to have initial distribution on memory elements $m_1,m_2$ with probability $1/2$ each. With memory $m_1$ we always choose action from $\nu$ and with memory $m_2$ from the ``opposite valuation''  $\overline \nu$ (where $\overline{\overline{ a}}$ is identified with $a$).

Therefore, each literal has frequency $1/\rstart p\rstop$ either in the first or the second kind of runs. Further, the runs of the first kind (with memory $m_1$) satisfy all clauses.

``If part'': Given a witness strategy $\sigma$ for $R(\varphi)$, we construct a satisfying valuation. First, we focus on the property induced by the (conjunctive-S) constraint. We show that almost all runs uniquely induce a valuation 
$$\nu_\sigma:=\{\ell\in Ap\cup\overline{Ap}\mid \vfreq_{\ell}>0\}$$
which follows from the following lemma:

\begin{lemma}
For every witness strategy $\sigma$ satisfying the (conjunctive-S) constraint, and for each $a\in Ap$, we have 
$$\Pr{\sigma}{}{\vfreq_a=\rstart \frac1p\rstop \text{ and } \vfreq_{\overline{ a}}=0}+\Pr{\sigma}{}{\vfreq_a=0 \text{ and } \vfreq_{\overline{ a}}=\rstart \frac1p\rstop}=1\,.$$
\end{lemma}
\begin{proof}
Let $a\in Ap$ be an arbitrary atomic proposition.
To begin with, observe that due to the circular shape of MDP $G_\varphi$, we have 
\begin{equation}
\vfreq_a+\vfreq_{\overline{ a}}\rstart\leq\rstop 1/\rstart p\rstop \label{eq:hard-freq}
\end{equation} for every run.
\rstart
Indeed, $\vfreq_a+\vfreq_{\overline{ a}}= \liminf_{T\to\infty}\frac1T\sum_{t=1}^T\kda +
\liminf_{T\to\infty}\frac1T\sum_{t=1}^T\kd{\overline a}\leq
\liminf_{T\to\infty}\frac1T\sum_{t=1}^T(\kda+\kd{\overline a})=1/p$.
\rstop

Therefore, the two events $\vfreq_a\geq 1/\rstart p\rstop$ and $\vfreq_{\overline{ a}}\geq 1/p$ are disjoint.
Due to the (conjunctive-S) constraint, almost surely exactly one of the events occurs.
\rstart
Indeed, \[1\geq\Pr{\sigma}{}{\vfreq_a\geq \rstart \frac1p
	\cup\vfreq_{\overline{ a}}\geq\rstart \frac1p\rstop}= 
\Pr{\sigma}{}{\vfreq_a\geq \rstart \frac1p\rstop}+
\Pr{\sigma}{}{\vfreq_{\overline{ a}}\geq\rstart \frac1p\rstop}\geq \frac12+\frac12=1\]
with the equality by disjointness of the events and the last inequality by (conjunctive-S).
\rstop

Therefore, by (\ref{eq:hard-freq}), almost surely either $\vfreq_a=\rstart 1/p\rstop \text{ and } \vfreq_{\overline{ a}}=0$, or $\vfreq_a=~0$ and $\vfreq_{\overline{ a}}=\rstart 1/p\rstop$.\end{proof}

By the (joint-S) constraint, we have a set $\Omega_{\mathit{sat}}$, with non-zero measure, of runs satisfying $\lrIf{\reward}_c\geq 1$ for each $c\in C$. By the previous lemma, almost all runs  of $\Omega_{\mathit{sat}}$ induce unique valuations.
Since there are finitely many valuation, at least one of them is induced by a set of non-zero measure.
Let $\pat$ be one of the runs and $\nu$ the corresponding valuation.
We claim that $\nu$ is a satisfying valuation for $\varphi$.

Let $c\in C$ be any clause, we show $\nu\models c$.
Since $\lrIf{\reward}(\omega)_c\geq 1$, there is an action $\ell$ such that
\begin{itemize}
\item $\vfreq_\ell(\omega)>0$, and
\item $\reward(a)_\ell\geq 1$.
\end{itemize}
The former inequality implies that $\ell\in \nu$ and the latter that $\ell\models c$. 
Altogether, $\nu\models c$ for every $c\in C$, hence $\nu$ witnesses satisfiability of $\varphi$.
\end{proof}

\medskip


Theorem~\ref{thm:hard} contrasts Theorem~\ref{thm:timepoly}: while extension of (\ref{eq:krish-sat}) with (\ref{eq:EXP}) can be solved in polynomial time, extending (\ref{eq:krish-sat}) with (\ref{eq:SAT}) makes the problem NP-hard. 
Intuitively, adding (\ref{eq:SAT}) 
enforces us to consider the subsets of dimensions,
and explains the exponential dependency on the number of dimensions in Theorem~\ref{thm:main} (though our lower bound
does not work for (\ref{eq:SAT}) with (\ref{eq:EXP})).

The results are summarized in Table~\ref{tab:complexity} and contrasted to the previously known polynomial bounds in Table~\ref{tab:complexity4}.

\section{Strategy complexity}

First, we recall the structure of witness strategies generated from $L$ in Section~\ref{ssec:thm-proof3}.
In the first phase, a memoryless strategy is applied to reach MECs and switch to the recurrent strategies $\xi_N$.
This switch is performed as a stochastic update, remembering the following two pieces of information: (1) the binary decision to stay in the current MEC $C$ forever, and (2) the set $N\subseteq[n]$, such that almost all the produced runs belong to $\Omega_N$.
Each recurrent strategy $\xi_N$ is then an infinite-memory strategy, where the memory is simply a counter. The counter determines which memoryless strategy $\zeta^\varepsilon_N$ is played.

\subsection{Randomization and memory}

Similarly to the traditional setting with the expectation or the satisfaction semantics considered separately, the case with a single objective is simpler.

\begin{lemma}\label{lem:1mem}
Deterministic memoryless strategies are sufficient for witness strategies for \textbp{(mono-qual)}.
\end{lemma}
\begin{proof}
For each MEC, there is a value, which is the maximal long-run average reward. This is achievable for all runs in the MEC and using a memoryless strategy $\xi$. We prune the MDP to remove MECs with values below the threshold $\sat$. A~witness strategy can be chosen to maximize the single long-run expected average objective, and thus also to be deterministic and memoryless \cite{Puterman}. Intuitively, in this case each MEC is either stayed at almost surely, or left almost surely if the value of the outgoing action is higher.
\end{proof}

Further, both for the expectation and the satisfaction semantics, deterministic memoryless strategies are sufficient for \emph{quantitative} queries~\cite{FV97,BBE10} with single objective. In contrast, we show that both randomization and memory is necessary in our combined setting even for $\varepsilon$-witness strategies.

\begin{example}\label{ex:rm}
Randomization and memory is necessary for \textbp{(mono-quant)} with 
$\sat=1,\ex=3,\psat=0.55$ and the MDP and $\reward$ depicted in Fig.~\ref{fig:ex-randmem}.
We have to remain in MEC $\{s,a\}$ with probability $p\in[0.1,2/3]$, hence we need a randomized decision. 
Further, memoryless strategies would either never leave $\{s,a\}$ or would leave it eventually almost surely. 
Finally, the argument applies to $\varepsilon$-witness strategies, 
since 
the inter\-val for $p$ contains neither $0$ nor $1$ for sufficiently small $\varepsilon$.

\begin{figure}[ht]
\centering
\spacefu
\begin{tikzpicture}
\node[state,initial,initial text=](s)at(0,0) {s};
\node[coordinate](c)at(1,0) {};
\node[state](t)at(3,0.5) {t};
\node[state](u)at(3,-0.5) {u};
\path
(s) edge[loop above] node[above]{$a,\reward(a)=2$} ()
(c) edge[->] node[above]{$0.5$}(t)
(c) edge[->] node[below]{$0.5$}(u)
(s) edge[-]  node[above]{$b$} (c)
(t) edge[loop right] node[right]{$c,\reward(c)=0$} ()
(u) edge[loop right] node[right]{$d,\reward(d)=10$} ()
;
\end{tikzpicture}
\caption{An MDP with a single objective, where both randomization and
  memory is necessary}\label{fig:ex-randmem}\vspace{-24 pt}
\spacefl
\end{figure}
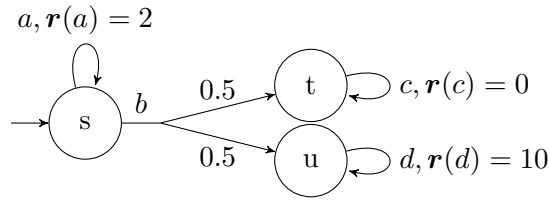
\QEE
\end{example}

In the rest of the section, we discuss bounds on the size of the memory and the degree of randomization.
Due to \cite[Section 5]{krish}, infinite memory is indeed necessary for witnessing (\ref{eq:krish-sat}) with $\psatscalar=1$, hence also for \textbp{(multi-qual)}.

\subsection{Memory bounds for deterministic update}

We prove that finite memory is sufficient in several cases, namely for all $\varepsilon$-witness strategies and for \textbp{(mono-quant)} witness strategies.
Moreover, these results also hold for deterministic-update strategies.
Indeed, as one of our technical contributions, we prove that stochastic update at the moment of switching is not necessary and deterministic update is sufficient, requiring only a finite blow up in the memory size.

\begin{lemma}\label{lem:detup}
Deterministic update is sufficient for witness strategies for \textbp{(multi-quant-conjuctive)} and \textbp{(multi-quant-joint)}. Moreover, finite memory is sufficient before switching to $\xi_N$'s.
\end{lemma}
\begin{proof}[Proof idea]
The stochastic decision during the switching in MEC $C$ can be done as a deterministic update after a ``toss'', a random choice between two actions in $C$ in one of the states of $C$.
Such a toss does not affect the long-run average reward as it is only performed finitely many times.

More interestingly, in MECs where no toss is possible, we can remember which states were visited how many times and choose the respective probability of leaving or staying in $C$.
\end{proof}
\begin{proof}
Let $\sigma$ be a strategy induced by $L$. 
We modify it into a strategy $\varrho$ with the same distribution of the long-run average rewards.
The only stochastic update that $\sigma$ performs is in a MEC, switching to $\xi_N$ with some probability.
We modify $\sigma$ into $\varrho$ in each MEC $C$ separately. 

\medskip
\noindent \textbp{Tossing-MEC case} \quad First, we assume that there are $\toss,a,b\in C$ with $a,b\in\act{\toss}$.
Whenever $\sigma$ should perform a step in $s\in C$ and possibly make a~sto\-chastic-update, say to $m_1$ with probability $p_1$ and $m_2$ with probability $p_2$, $\varrho$ performs a ``toss'' instead. A $(p_1,p_2)$-toss consists of reaching $\toss$ with probability $1$ (using a memoryless strategy), taking $a,b$ with probabilities $p_1,p_2$, respectively, and making a deterministic update based on the result, in order to remember the result of the toss. After the toss, $\varrho$ returns back to $s$ with probability $1$ (again using a memoryless strategy). Now as it already remembers the result of the $(p_1,p_2)$-toss, it changes the memory to $m_1$ or $m_2$ accordingly, by a deterministic update.

In general, since the stochastic-update probabilities depend on the action chosen and the state to be entered, we have to perform the toss for each combination before returning to $s$. Further, whenever there are more possible results for the memory update (e.g.\ various $N$), we can use binary encoding of the choices, say with $k$ bits, and repeat the toss with the appropriate probabilities $k$-times before returning to $s$.

This can be implemented using finite memory. Indeed, since there are finitely many states in a MEC and $\sigma$ is memory\-less, there are only finitely many combinations of tosses to make and remember till the next simulated update of $\sigma$.

\medskip
\noindent \textbp{Tossfree-MEC case}\quad It remains to handle the case where, for each state $s\in C$, there is only one action $a\in\act{s}\cap C$. Then all strategies staying in $C$ behave the same here, call this memoryless deterministic strategy $\xi$. 
Therefore, the only stochastic update that matters is to stay in $C$ or not. 
The MEC $C$ is left via each action $a$ with the probability
$$\leave_a:=\sum_{t=1}^\infty\Pr{\sigma}{}{S_t\in C \text{ and } A_t=a \text{ and } S_{t+1}\notin C}$$
and let $\{a\mid leave_a>0\}=\{a_1,\ldots, a_\ell\}$ be the leaving actions.
The strategy $\varrho$ upon entering $C$ performs the following. First, it leaves $C$ via $a_1$ with probability $\mathit{leave}_{a_1}$ (see below how), then via $a_2$ with probability $\frac{\leave_{a_2}}{1-\leave_{a_1}}$, and so on via $a_i$ with probability 
$$\frac{\leave_{a_i}}{1-\sum_{j=1}^{i-1}\leave_{a_j}}$$
subsequently for each $i\in[\ell]$.  
After the last attempt with $a_\ell$, if we are still in $C$, we update memory to stay in $C$ forever (playing $\xi$).

Leaving $C$ via $a$ with probability $\leave$ can be done as follows. 
Let $\mathit{rate}=\sum_{s\notin C}\delta(a)(s)$ be the probability to actually leave $C$ when taking $a$ once. Then to achieve the overall probability $\leave$ of leaving we can reach $s$ with $a\in\act{s}$ and play $a$ with probability $1$ and repeat this 
\rstart $m$ times for some $m\in\Nset$ (if $\leave=1$ then $m=\infty$)
and finally reach $s$ once more and play $a$ with probability 
$p\in[0,1]$ 
and an action staying in $C$ with the remaining probability.
We now define $m$ and $p$.
If $\mathit{rate}=1$ then $m=0$ and $p=\leave$.
Assume $\mathit{rate}<1$. Then we must ensure that the probability not to leave via $a$ be 
\begin{equation}
1-\leave=(1-\mathit{rate})^m\cdot \big(p(1-\mathit{rate})+(1-p)\big)
\label{eq:leave}
\end{equation}
Indeed, $(1-\mathit{rate})^m$ stands for failing to leave $m$-times, and the last time we either choose $a$ and fail again or not choose $a$ at all. This requirement is equivalent to
\[
m=\frac{\ln(1-\leave)-\ln(1-p\cdot\mathit{rate})}{\ln(1-\mathit{rate})}
\]
For $p\in[0,1]$ we have also $\frac{\ln(1-p\cdot\mathit{rate})}{\ln(1-\mathit{rate})}\in[0,1]$.
Therfore, in order to choose $m\in\Nset$, we can simply set $m:=\lfloor\frac{\ln(1-\leave)}{\ln(1-\mathit{rate})}\rfloor$, which also ensures that $p\in[0,1]$ for the respective  $p:=\frac 1{\mathit{rate}}(1-\frac{1-\leave}{(1-\mathit{rate})^m})$, obtained from (\ref{eq:leave}).
\rstop

In order to implement the strategy in MECs of this second type, for each action it is sufficient to have a counter up to
\rstart the respective $m$\rstop.
\end{proof}

\begin{remark}
Moreover, our proof also shows, that finite memory is sufficient before switching to $\xi_N$'s (as defined in Section~\ref{ssec:thm-proof3}) for deterministic-update witnessing (and $\varepsilon$-witnessing) strategies. 
Therefore, finite memory deterministic update is sufficient for $\varepsilon-$witness strategies, in particular also for (\ref{eq:krish-sat}), which improves the strategy complexity known from \cite{krish}.
Note that in general, conversion of a stochastic-update strategy to a deterministic-update strategy requires an infinite blow up in the memory~\cite{TCS:AlfaroHK07}.
\QEE
\end{remark}

As a consequence, we obtain several bounds on memory size valid even for deterministic-update strategies. Firstly, infinite memory is required only for witness strategies:

\begin{lemma}\label{lem:finmem}
Deterministic-update with finite memory is sufficient for $\varepsilon$-witness strategies for \textbp{(multi-quant-conjuctive)} and \textbp{(multi-quant-joint)}.
\end{lemma}
\begin{proof}
After switching, memoryless strategies $\zeta^\varepsilon_N$ can be played instead of the sequence of $\zeta^{1/2^i}_N$. 
\end{proof}

\begin{remark}\label{rem:strategycj}
The previous proof of sufficiency of deterministic-update finite memory for $\varepsilon$-witness strategies applies also to \textbp{(multi-quant-conjunctive-joint)}. Indeed, firstly,  Lemma~\ref{lem:detup} applies verbatim to \textbp{(multi-quant-conjunctive-joint)}.
Secondly, we switch to only finitely many recurrent strategies due to Remark~\ref{rem:cj-struct}.
\QEE
\end{remark}
Secondly, infinite memory is required only for multiple objectives:

\begin{lemma}\label{lem:finmem-mono}
Deterministic-update strategies with finite memory are sufficient witness strategies for \textbp{(mono-quant)}.
\end{lemma}
\begin{proof}
After switching in a MEC $C$, we can play the following memoryless strategy. In $C$, there can be several components of the flow. We pick any with the largest long-run average reward.
\end{proof}

Further, the construction in the toss-free case gives us a hint for the respective lower bound on memory, even for the single-objective case.

\begin{example}\label{ex:finmem-mono}
For deterministic-update $\varepsilon$-witness strategies for \textbp{(mono-quant)} problem, memory with size dependent on the transition probabilities is necessary.
Indeed, consider the same realizability problem as in Example~\ref{ex:rm}, but with a slightly modified MDP parametrized by $\lambda$, depicted in Fig.~\ref{fig:ex-largemem}. Again, we have to remain in MEC $\{s,a\}$ with probability $p\in[0.1,2/3]$. For $\varepsilon$-witness strategies the interval is slightly wider;
let $\ell>0$ denote the minimal probability with which any ($\varepsilon$-)witness strategy has to leave the MEC and all ($\varepsilon$-)witness strategies have to stay in the MEC with positive probability. We show that at least $\lceil\frac\ell\lambda\rceil$-memory is necessary. 
Observe that this setting also applies to the (\ref{eq:EXP}) setting of \cite{krish}, e.g.\ $\ex=(0.5,0.5)$ and the MDP of Fig.~\ref{fig:forkrish}. Therefore, we provide a lower bound also for this simpler case 
(no MDP-dependent lower bound is provided in \cite{krish}).

\begin{figure}[ht]
\centering
\spacefl\spacefl
\begin{tikzpicture}
\node[state,initial,initial text=](s)at(0,0) {s};
\node[coordinate](c)at(1,0) {};
\node[state](t)at(3,0.5) {t};
\node[state](u)at(3,-0.5) {u};
\path
(s) edge[loop above] node[above]{$a,\reward(a)=2$} ()
(c) edge[->] node[above]{$\frac\lambda2$}(t)
(c) edge[->] node[below]{$\frac\lambda2$}(u)
(c) edge[->,bend left=90pt,looseness=2] node[below]{$1-\lambda$}(s.south)
(s) edge[-]  node[above]{$b$} (c)
(t) edge[loop right] node[right]{$c,\reward(c)=0$} ()
(u) edge[loop right] node[right]{$d,\reward(d)=10$} ()
;
\end{tikzpicture}
\caption{An MDP family with a single objective, where memory with size dependent on transition probabilities is necessary for deterministic-update strategies}\label{fig:ex-largemem}
\end{figure}
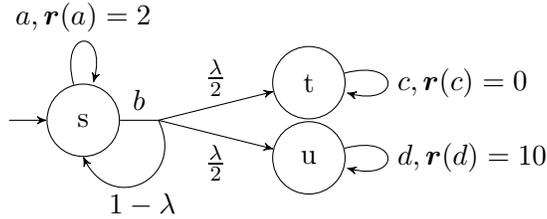

\begin{figure}[ht]
\centering
\spacefu
\begin{tikzpicture}
\node[state,initial,initial text=](s)at(0,0) {s};
\node[coordinate](c)at(1,0) {};
\node[state](t)at(3,0) {t};
\path
(s) edge[loop above] node[above]{$a,\reward(a)=(1,0)$} ()
(c) edge[->] node[above]{$\lambda$}(t)
(c) edge[->,bend left=90pt,looseness=2] node[below]{$1-\lambda$}(s.south)
(s) edge[-]  node[above]{$b$} (c)
(t) edge[loop right] node[right]{$c,\reward(c)=(0,1)$} ()
;
\end{tikzpicture}
\caption{An MDP family, where 
memory with size dependent on transition probabilities is necessary for deterministic-update strategies
even for (\ref{eq:EXP}) studied in \cite{krish}}\label{fig:forkrish}
\spacefl
\spacefl
\end{figure}
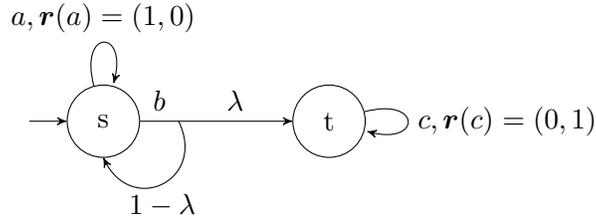

For a contradiction, assume there are less than $\lceil\frac\ell\lambda\rceil$ memory elements. Then, by the pigeonhole principle, in the first $\lceil\frac\ell\lambda -1\rceil$ visits of $s$, some memory element $m$ appears twice. Note that due to the deterministic updating, each run generates the same play, thus the same sequence of memory elements. Let $p$ be the probability to eventually leave $s$ provided we are in $s$ with memory $m$.

If $p=0$ then the probability to leave $s$ at the start is less than $\lceil\frac\ell\lambda -2\rceil\cdot\lambda<\ell$, a contradiction. Indeed, we have at most $\lceil\frac\ell\lambda -2\rceil$ tries to leave $s$ before obtaining memory $m$ and with every try we leave $s$ with probability at most $\lambda$; we conclude by the union bound.

Let $p>0$. Due to the deterministic updates, all runs staying in $s$ use memory $m$ infinitely often. Since $p>0$, there is a finite number of steps such that (1) during these steps the overall probability to leave $s$ is at least $p/2$ and (2) we are using $m$ again. Consequently, the probability of the runs staying in $s$ is $0$, a contradiction.
\QEE
\end{example}

\subsection{Memory bounds for stochastic update}

Although we have shown that stochastic update is not necessary, it may be helpful when memory is small. 

\begin{lemma}\label{lem:2mem}
Stochastic-update 2-memory strategies are sufficient for witness strategies for \textbp{(mono-quant)}.
\end{lemma}
\begin{proof}
The strategy $\sigma$ of Section~\ref{ssec:thm-proof3}, which reaches the MECs and stays in them with given probability, is memoryless up to the point of switch by Corollary~\ref{cl:m-reach-strat}.
Further, we can achieve the optimal value in each MEC using a memoryless strategy as in Lemma~\ref{lem:finmem-mono}.
\end{proof}

\begin{theorem}\label{epsilonstrategy}
Upper bounds on memory size for stochastic-update $\varepsilon$-witness strategies are as follows:
\begin{itemize}
\item \textbp{(multi-qual)} $2$ memory elements,
\item \textbp{(multi-quant-joint)} $3$ memory elements,
\item \textbp{(multi-quant-conjunctive)} $2^n+1$ memory elements,
\item \textbp{(multi-quant-conjunctive-joint)} $2^{n+1}+1$ memory elements.
\end{itemize}
\end{theorem}

\begin{proof}


The structure of $\varepsilon$-witness strategies is described in Remark~\ref{construction}. Let us recall from Corollary~\ref{cl:m-reach-strat} that strategy $\sigma$ is memoryless before the switch. 
For \textbp{(multi-qual)}, \textbp{(multi-quant-joint)} and \textbp{(multi-quant-conjunctive)}, we perform the stochastic-update switch to different memory elements corresponding to the different strategies $\zeta^\varepsilon_N$. From Lemma \ref{lem:connecting-BSCCs-} we have that every such strategy $\zeta^\varepsilon_N$ is also memoryless. From Lemma \ref{lem:transient} we have that we switch only to such $\zeta^\varepsilon_N$ for $N\subseteq [n]$, which correspond to possible nonzero variables $y_{s,N}$.
Therefore, the number of memory elements needed is the number of possible nonzero variables $y_{s,N}$ for $N\subseteq [n]$ and additionally one element for the strategy $\sigma$ before the switch.

Altogether, we get the following upper bounds on memory size of $\varepsilon$-witness strategies. 
For \textbp{(multi-quant-conjunctive)}, $2^n+1$ memory elements are sufficient, since all of the $y_{s,N}$ for $N\subseteq [n]$ can be positive.
For \textbp{(multi-quant-joint)}, $3$ memory elements are sufficient, because we use only $y_{s,[n]}$ and $y_{s,\emptyset}$ as discussed in \ref{joint}.
 Finally for \textbp{(multi-qual)}, $2$ memory elements are sufficient, because we use only $y_s$ as in \ref{multiqual}.

Due to Remark~\ref{rem:cj-struct}, the bound on the number of recurrent strategies for \textbp{(multi-quant-conjunctive-joint)} is twice as large as for \textbp{(multi-quant-conjunctive)}, i.e., $2^{n+1}$. The upper bound on the size of memory for $\varepsilon$-witness strategies for \textbp{(multi-quant-conjunctive-joint)} is thus $1+2^{n+1}$, compared to $1+2^n$ for \textbp{(multi-quant-conjunctive)}.  
\end{proof}

\begin{example}\label{ex:3mem} 
For \textbf{(multi-quant-joint)}$, \varepsilon$-witness strategies may require 
memory with at least $3$ elements. 
Consider an MDP with two states $s$ and $t$ with transitions and rewards as depicted in Fig.~\ref{fig:ex-3mem}.
Further, let $\sat=\vec(1,0,0)$, $\psat=\frac{1}{2}$ and ${\ex=\vec(0,1,1)}$.

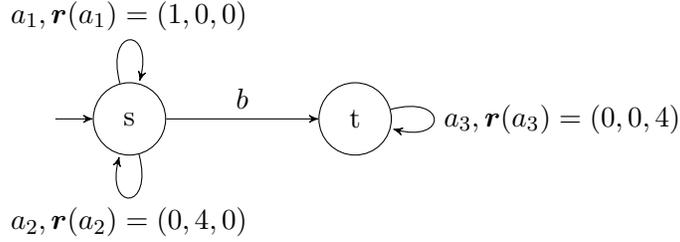
\begin{figure}[ht]
\centering
\spacefu
\begin{tikzpicture}
\node[state,initial,initial text=](s)at(0,0) {s};
\node[state](t)at(3,0) {t};
\path
(s) edge[loop above] node[above]{$a_1,\reward(a_1)=(1,0,0)$} ()
(s) edge[->] node[above]{$b$}(t)
(s) edge[loop below] node[below]{$a_2,\reward(a_2)=(0,4,0)$} ()
(t) edge[loop right] node[right]{$a_3,\reward(a_3)=(0,0,4)$} ()
;
\end{tikzpicture}
\caption{An MDP where $3$-memory is necessary for \textbf{(multi-quant-joint)}}\label{fig:ex-3mem}
\spacefl
\end{figure}

Suppose $2$ memory elements are sufficient. In state $s$ for each memory element we can either stay in $s$ or go with some positive probability to state $t$. Therefore we have three cases on the behaviour in $s$ regarding the transition to $t$: 
\begin{enumerate}
\item for each memory element we have positive probability $p_1$ and $p_2$ respectively, to go to state $t$,
\item for both memory elements we have zero probability to go to $t$ and
\item for one memory element, say memory element 1, we have zero probability and for the other one, say memory element 2, we have positive probability $p$ to go to $t$.

\end{enumerate}

In the first case, we go to $t$ eventually almost surely. Indeed, in each step we enter $t$ with probability at least $\min(p_1,p_2)$ and cannot return back. Therefore, we stay in $t$ forever and thus we cannot satisfy the satisfaction constraint.

In the second case, we never enter state $t$. Hence, we cannot satisfy the expectation constraint, because $\reward(a_1)_3=\reward(a_2)_3=0.$ 

In the third case, we firstly assume that we switch from memory $1$ to $2$ with some positive probability $p_1$. Then in each step we have at least probability $p_1\cdot p$ to enter $t$. Therefore, we end up in state $t$ almost surely, not satisfying constraints, as shown above.
Secondly, suppose we cannot switch from memory $1$ to $2$. Then we almost surely end up in state $s$ with memory $1$ or in state $t$. In state $s$ with memory $1$ we can either play action $a_1$ with probability $1$ or with smaller potentially zero probability $q$. In the former case, $\lrLim{\reward_2}=0$, thus violating the expectation constraint. In the latter case, for almost every run $\lrLim{\reward_1}\leq 1-q$, contradicting the satisfaction constraint.

Note that a witnessing strategy exists, which uses only $3$ memory elements. On half of the runs, we play only action $a_1$ to satisfy the satisfaction constraint. So we define $\sigma_n(s,1)(a_1)=1$. To satisfy the expectation constraint for $\reward_2$ we define $\sigma_n(s,2)(a_2)=1$. With the  last memory element we want to satisfy the expectation constraint for $\reward_3$ and thus we define $\sigma_n(s,3)(b)=1$ and $\sigma_n(t,3)(a_3)=1$.
We define the initial distribution by $\alpha(1)=\frac{1}{2}$, $\alpha(2)=\frac{1}{4}$ and $\alpha(3)=\frac{1}{4}$ and therefore the memory update function not to change memory. Consequently, the achieved expectation is $(\frac{1}{2}\cdot 1, \frac{1}{4}\cdot 4,\frac{1}{4}\cdot 4)\geq \ex$.
\QEE
\end{example}

However, even with stochastic update, the size of the finite memory cannot be bounded by a constant for \textbp{(multi-quant-conjunctive)}.

\begin{example}\label{ex:mem-n}
Even $\varepsilon$-witness strategy for \textbp{(multi-quant-conjunctive)} may require 
memory with at least $n$ memory elements. 
Consider an MDP with a single state $s$ and self-loop $a_i$ with reward $\reward_i(a_j)$ equal to $1$ for $i=j$ and $0$ otherwise, for each $i\in[n]$. 
Fig.~\ref{fig:ex-nmem} illustrates the case with $n=3$.
Further, let $\sat=\vec1$ and $\psat=1/n \cdot\vec 1$.

\begin{figure}[ht]
\centering
\spacefu
\begin{tikzpicture}
\node[state,initial,initial text=](s)at(0,0) {s};
\path
(s) edge[loop above] node[above]{$a_1,\reward(a_1)=(1,0,0)$} ()
(s) edge[loop right] node[right]{$a_2,\reward(a_2)=(0,1,0)$} ()
(s) edge[loop below] node[below]{$a_3,\reward(a_3)=(0,0,1)$} ()
;
\end{tikzpicture}
\caption{An MDP where $n$-memory is necessary, depicted for $n=3$}\label{fig:ex-nmem}
\spacefl
\end{figure}
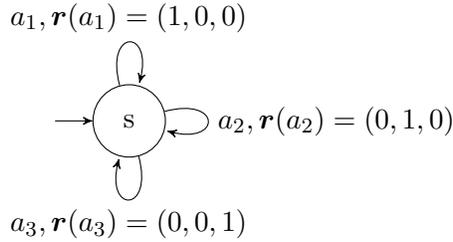
The only way to $\varepsilon$-satisfy the constraints is that for each $i$, $1/n$ runs take only $a_i$, but for a negligible portion of time. 
Since these constraints are mutually incompatible for a single run, $n$ different decisions have to be repetitively taken at $s$, showing the memory requirement.
\QEE
\end{example}


We summarize the upper and lower bounds for witness and $\varepsilon$-witness strategies in Table~\ref{tab:complexity2} and Table~\ref{tab:complexity3}, respectively.

\spacefl

\section{Pareto curve approximation and complexity summary}\label{sec:complexity-pareto}

For a single objective, no Pareto curve is required and we can compute the optimal value of expectation in polynomial time by the linear program $L$ with the objective function $\max\sum_{a\in A}(x_{a,\emptyset}+x_{a,\{1\}})\cdot\reward(a)$. For multiple objectives we obtain the following:

\begin{theorem}\label{thm:pareto}
For $\varepsilon>0$, an $\varepsilon$-approximation of the Pareto curve for \textbp{(multi-quant-conjunctive-joint)} can be constructed in time polynomial in $|G|$ and $\frac1\varepsilon$ and exponential in $n$.
\end{theorem}\vspace*{-0.7em}
\begin{proof}
We replace $\ex$ in Equation 5 of $L$ by a vector $\vec v$ of variables. Maximizing with respect to $\vec v$ is a multi-objective linear program. By~\cite{PY00}, we can $\varepsilon$-approximate the Pareto curve in time polynomial in the size of the program and $\frac1\varepsilon$, and exponential in the number of objectives (dimension of $\vec v$).
%
\end{proof}
The proof of Theorem~\ref{thm:pareto} shows that we can obtain a Pareto-curve approximation also for possible values of the $\sat$ or $\psat$ vectors for a given $\ex$ vector. We simply replace these vectors by vectors of variables, obtaining a multi-objective linear program. If we want the complete Pareto-curve approximation for all the parameters $\sat$, $\psat$, and $\ex$, the number of objectives rises from $n$ to $3\cdot n$. The complexity is thus still polynomial in the size of the MDP and $1/\varepsilon$, and exponential in $n$.

In particular, for the single-objective case, we can compute also the optimal $\psat$ given $\ex$ and $\sat$, or the optimal $\sat$ given $\psat$ and $\ex$.

\bigskip

The complexity results are summarized in the following theorem:

\begin{theorem}
The algorithmic complexities are shown in Table~\ref{tab:complexity}.
The bounds on the complexity of the witness and $\varepsilon$-witness strategies are as shown in Table~\ref{tab:complexity2} and Table~\ref{tab:complexity3}, respectively. 
\end{theorem}

\subsubsection*{Comments on the tables}

 \ub denotes upper bounds (which suffice for all MDPs) and \lb lower bounds (which are required in general for some MDPs). Results without reference are induced by the specialization or generalization relation depicted in Fig.~\ref{fig:problems} and for Table \ref{tab:complexity2} and \ref{tab:complexity3} by $\varepsilon-$witness strategies being a weaker notion than witness strategies.
The abbreviations stoch.-up., det.-up., rand., det., inf., fin., and $X$-mem. stand for stochastic update, deterministic update, randomizing, deterministic, infinite-, finite- and $X$-memory strategies, respectively. Here $n$ is the dimension of reward function and $p=1/p_{\mathit{\min}}$ where $p_{\mathit{\min}}$ is the smallest positive probability in the MDP.
Note that inf.\ actually means that the strategy is in form of a Markov strategy, see Section~\ref{ssec:thm-proof3}.

 \begin{remark}
 For a comparison, the results on previously studied subcases of our problems are depicted in Table~\ref{tab:complexity4}.
 \QEE
 \end{remark}

\begin{table*}[h]
\caption{Previous results on algorithmic and strategy complexities. 
The abbreviations alg., strat., and c. stand for algorithmic, strategy, and complexity, respectively. 
Cases multiple and single refer to the number of objectives.
Results for single-objective MDPs are based on classical literature, e.g.~\cite[Thm.9.1.8]{Puterman}. 
Results for MDPs with multiple objectives are due to \cite{krish}.
}
\label{tab:complexity4}
 \begin{tabular}{|l|l|l|l|} \hline
 Case & Alg. c. & Witness strat. c. & $\varepsilon$-witness strat. c. \\\hline

{multiple} & $\poly(|G|,n)$ & \ub det.-up. inf. & \ub stoch.-up. 2-mem.\\
 (\ref{eq:krish-sat})&&\lb rand. inf. & \lb rand. $2$-mem.\\\hline
{multiple} & $\poly(|G|,n)$ & \ub det.-up. inf. & \ub stoch.-up. 2-mem., det.-up. fin.\\
(\ref{eq:EXP}) && \lb rand. inf. & \lb rand. $2$-mem.\\\hline
{single} & $\poly(|G|)$ & \ubl det. 1-mem. &  \ubl det. 1-mem.\\
(\ref{eq:krish-sat})&&&\\\hline
{single} & $\poly(|G|)$ & \ubl det. 1-mem. & \ubl det. 1-mem.\\
(\ref{eq:EXP})&&&\\\hline
 \end{tabular} 
\end{table*}

\medskip

 %

\begin{table*}[h]
\caption{Algorithmic complexity results for each of the discussed cases. }
\label{tab:complexity}
\begin{tabular}{|l|l|} \hline
Case & Algorithmic complexity \\\hline
\textbp{(multi-quant-conj.-joint)} \hspace{8mm} &  $\poly(|G|,2^n)$ [Cor.\ref{cor:conjjoint}], NP-hard [Thm. \ref{thm:hard}]\hspace{8mm} \\\hline
\textbp{(multi-quant-conj.)}&  $\poly(|G|,2^n)$ [Thm.\ref{thm:main}]\\\hline
\textbp{(multi-quant-joint)} &$\poly(|G|,n)$ [Thm.\ref{thm:timepoly}] \\\hline
\textbp{(multi-qual)} &  $\poly(|G|,n)$ \\\hline
\textbp{(mono-quant)} & $\poly(|G|)$  \\\hline
\textbp{(mono-qual)} & $\poly(|G|)$ \\\hline
\end{tabular} 
\medskip
\end{table*}

\begin{table*}[h]
\caption{Witness strategy complexity bounds for each of the discussed cases.}
\label{tab:complexity2}
\begin{tabular}{|l|l|} \hline
Case & Witness strategy complexity \\\hline

\textbp{(multi-quant-conj.-joint)} & \ub det.-up. [Rem.\ref{rem:strategycj}] inf.\\
& \lb rand. inf. \\\hline

\textbp{(multi-quant-conj.)} & \ub det.-up. [Lem.\ref{lem:detup}] inf.\\
& \lb rand. inf. \\\hline

\textbp{(multi-quant-joint)}  & \ub det.-up. inf. \\
& \lb rand. inf. \\\hline

\textbp{(multi-qual)}  & \ub det.-up. inf. \\
& \lb rand. inf. \cite[Sec.5]{krish} \\\hline

\textbp{(mono-quant)}  & \ub stoch.-up. 2-mem. [Lem.\ref{lem:2mem}], det.-up. fin. [Lem.\ref{lem:finmem-mono}]  \\
 & \lb rand. 2-mem., 
for det.-up. $p$-mem.
\\\hline

\textbp{(mono-qual)} & \ub (trivially also \lb) det. 1-mem. [Lem.\ref{lem:1mem}] \\\hline
\end{tabular} 
\medskip
\end{table*}

\begin{table*}[h]
\caption{$\varepsilon$-witness strategy complexity bounds for each of the discussed cases.}
\label{tab:complexity3}
\begin{tabular}{|l|l|} \hline
Case &  $\varepsilon$-witness strategy complexity \\\hline

\textbp{(multi-quant-} & \ub stoch.-up. $(2^{n+1}+1)$-mem. [Thm.\ref{epsilonstrategy}], det.-up. fin. [Rem.\ref{rem:strategycj}]\\
\qquad\textbp{conj.-joint)}& \lb rand. $n$-mem. [Ex.\ref{ex:mem-n}], for det.-up. $p$-mem.\\\hline

\textbp{(multi-quant-} & \ub stoch.-up. $(2^n+1)$-mem. [Thm.\ref{epsilonstrategy}], det.-up. fin. [Lem.\ref{lem:finmem}]\\
\qquad\textbp{conj.)}& \lb rand. $n$-mem. [Ex.\ref{ex:mem-n}], for det.-up. $p$-mem.\\\hline

\textbp{(multi-quant-}  & \ub stoch.-up. 3-mem.  [Thm.\ref{epsilonstrategy}], det.-up. fin.\\
\qquad\textbp{joint)}& \lb rand. $3$-mem. [Ex.\ref{ex:3mem}]\\\hline

\textbp{(multi-qual)}  & \ub stoch.-up. 2-mem. [Thm.\ref{epsilonstrategy}], det.-up. fin.\\
 & \lb rand. mem. \cite[Sec.3]{krish}\\\hline

\textbp{(mono-quant)}  &  \ub stoch.-up. 2-mem., det.-up. fin. \\
 
& \lb rand. [Ex.\ref{ex:rm}] 2-mem. [Ex.\ref{ex:rm}], for det.-up. $p$-mem. [Ex.\ref{ex:finmem-mono}] \\\hline

\textbp{(mono-qual)} &  \ub (trivially also \lb) det. 1-mem. \\\hline
\end{tabular} 
\end{table*}

\section{Conclusion}\label{chap:concl}

We have presented a unifying solution framework to the expectation and satisfaction optimization of Markov decision processes with multiple objectives. This allows us to synthesize optimal and $\varepsilon$-optimal risk-averse strategies.
We have considered several possible combinations of the two semantics and provided algorithms for their solution as well as the complete picture of the complexities for all these cases.

Regarding the algorithmic complexity, we have shown that \textbp{(multi-quant-joint)} and all its special cases can be solved in polynomial time. For both \textbp{(multi-quant-conjunctive)} and \textbp{(multi-quant-conjunctive-joint)}, we have presented an algorithm that works in time polynomial in the size of MDP, but exponential in the dimension of reward function. However, the exponential in the dimension of reward function is not a limitation for \rstart most of practical purposes since the dimension is typically low. \rstop For the latter case we have also proved that the problem is NP-hard. The complexity of \textbp{(multi-quant-conjunctive)} remains an interesting open question.
Moreover, our algorithms for Pareto-curve approximation work in time
polynomial in the size of MDPs and exponential in the dimension of reward
function. However, note that even for the special case of expectation semantics the
current best known algorithms depend exponentially on the dimension of
reward function~\cite{krish}.

We have also provided comprehensive results on strategy complexities. It is known that for both expectation and satisfaction semantics with single objective, deterministic memoryless strategies are sufficient~\cite{FV97,BBE10,krish}.
We have shown this carries over in the \textbp{(mono-qual)} case only. 
In contrast, for \textbp{(mono-quant)} both randomization and memory is necessary. However, we have also shown that only a restricted form of randomization (deterministic update) is necessary even for \textbp{(multi-quant)}, thus improving the upper bound for $\varepsilon-$witness strategies for the satisfaction problem of \cite{krish} to finite-memory deterministic update. Furthemore, we have established that with deterministic update the memory size is dependent on the MDP; the result also applies to the expectation problem of \cite{krish}, where no MDP-dependent lower bound was given. 
We have presented upper bounds on stochastic update $\varepsilon-$witness strategies, which are constant for \textbp{(multi-qual)} and \textbp{(multi-quant-joint)}, and exponentially dependent on the dimension of reward function for \textbp{(multi-quant-conjunctive)} and \textbp{(multi-quant-conjunctive-joint)}. The question whether there are polynomially dependent upper bounds for the latter two cases stays open.

\textbf{Acknowledgements} We are very thankful to the anonymous reviewers for their helpful suggestions and pointing at gaps in the proofs of Lemma~\ref{lem:cond-flow} and Lemma~\ref{lem:x-mec}, and to Rasmus Ibsen-Jensen for discussing the proof of Lemma~\ref{lem:cond-flow-simpler}.

%

\bibliographystyle{alpha}
\bibliography{bib}

\appendix

\section{Limear program for the running example}\label{app:lp}

\begin{enumerate}
 \item 
$1+0.5 y_\ell=y_\ell+y_r+y_{s,\emptyset}+y_{s,\{1\}}+y_{s,\{2\}}+y_{s,\{1,2\}}$\\
$0.5 y_\ell + y_a=y_a+y_{u,\emptyset}+y_{u,\{1\}}+y_{u,\{2\}}+y_{u,\{1,2\}}$\\
$y_r + y_b + y_e=y_b + y_c+y_{v,\emptyset}+y_{v,\{1\}}+y_{v,\{2\}}+y_{v,\{1,2\}}$\\
$y_c + y_d=y_d + y_e+y_{w,\emptyset}+y_{w,\{1\}}+y_{w,\{2\}}+y_{w,\{1,2\}}$\\
 \item 
 $y_{u,\emptyset}+y_{u,\{1\}}+y_{u,\{2\}}+y_{u,\{1,2\}}+y_{v,\emptyset}+y_{v,\{1\}}+y_{v,\{2\}}+y_{v,\{1,2\}} + y_{w,\emptyset}+y_{w,\{1\}}+y_{w,\{2\}}+y_{w,\{1,2\}}=1$\\
 \item 
$y_{u,\emptyset}= x_{a,\emptyset}$\\
$y_{u,\{1\}}= x_{a,\{1\}}$\\
$y_{u,\{2\}}= x_{a,\{2\}}$\\
$y_{u,\{1,2\}}= x_{a,\{1,2\}}$\medskip

\noindent
$y_{v,\emptyset}+y_{w,\emptyset}= x_{b,\emptyset}+x_{c,\emptyset}+x_{d,\emptyset}+x_{e,\emptyset}$\\
$y_{v,\{1\}}+y_{w,\{1\}}= x_{b,\{1\}}+x_{c,\{1\}}+x_{d,\{1\}}+x_{e,\{1\}}$\\
$y_{v,\{2\}}+y_{w,\{2\}}= x_{b,\{2\}}+x_{c,\{2\}}+x_{d,\{2\}}+x_{e,\{2\}}$\\
$y_{v,\{1,2\}}+y_{w,\{1,2\}}= x_{b,\{1,2\}}+x_{c,\{1,2\}}+x_{d,\{1,2\}}+x_{e,\{1,2\}}$\\
 \item 
 $0.5 x_{\ell,\emptyset} = x_{\ell,\emptyset} + x_{r,\emptyset}$\\
 $0.5 x_{\ell,\{1\}} = x_{\ell,\{1\}} + x_{r,\{1\}}$\\
 $0.5 x_{\ell,\{2\}} = x_{\ell,\{2\}} + x_{r,\{2\}}$\\
 $0.5 x_{\ell,\{1,2\}} = x_{\ell,\{1,2\}} + x_{r,\{1,2\}}$\medskip
 
 \noindent
$0.5 x_{\ell,\emptyset} + x_{a,\emptyset}= x_{a,\emptyset}$\\
$0.5 x_{\ell,\{1\}} + x_{a,\{1\}}= x_{a,\{1\}}$\\
$0.5 x_{\ell,\{2\}} + x_{a,\{2\}}= x_{a,\{2\}}$\\
$0.5 x_{\ell,\{1,2\}} + x_{a,\{1,2\}}= x_{a,\{1,2\}}$\medskip

\noindent
$x_{r,\emptyset} + x_{b,\emptyset} + x_{e,\emptyset}= x_{b,\emptyset} + x_{c,\emptyset}$\\
$x_{r,\{1\}} + x_{b,\{1\}} + x_{e,\{1\}}= x_{b,\{1\}} + x_{c,\{1\}}$\\
$x_{r,\{2\}} + x_{b,\{2\}} + x_{e,\{2\}}= x_{b,\{2\}} + x_{c,\{2\}}$\\
$x_{r,\{1,2\}} + x_{b,\{1,2\}} + x_{e,\{1,2\}}= x_{b,\{1,2\}} + x_{c,\{1,2\}}$\medskip

\noindent
$x_{c,\emptyset} + x_{d,\emptyset}= x_{d,\emptyset} + x_{e,\emptyset}$\\
$x_{c,\{1\}} + x_{d,\{1\}}= x_{d,\{1\}} + x_{e,\{1\}}$\\
$x_{c,\{2\}} + x_{d,\{2\}}= x_{d,\{2\}} + x_{e,\{2\}}$\\
$x_{c,\{1,2\}} + x_{d,\{1,2\}}= x_{d,\{1,2\}} + x_{e,\{1,2\}}$\\
 \item 
 $\reward(\ell) x_{\ell,\emptyset}+\reward(\ell) x_{\ell,\{1\}}+\reward(\ell) x_{\ell,\{2\}}+\reward(\ell) x_{\ell,\{1,2\}} + \reward(r) x_{r,\emptyset}+\reward(r) x_{r,\{1\}}+\reward(r) x_{r,\{2\}}+\reward(r) x_{r,\{1,2\}} + (4,0) x_{a,\emptyset}+(4,0) x_{a,\{1\}}+(4,0) x_{a,\{2\}}+(4,0) x_{a,\{1,2\}} + (1,0) x_{b,\emptyset}+(1,0) x_{b,\{1\}}+(1,0) x_{b,\{2\}}+(1,0) x_{b,\{1,2\}} + (0,0) x_{c,\emptyset}+(0,0) x_{c,\{1\}}+(0,0) x_{c,\{2\}}+(0,0) x_{c,\{1,2\}} + (0,1) x_{d,\emptyset}+(0,1) x_{d,\{1\}}+(0,1) x_{d,\{2\}}+(0,1) x_{d,\{1,2\}} + (0,0) x_{e,\emptyset}+(0,0) x_{e,\{1\}}+(0,0) x_{e,\{2\}}+(0,0) x_{e,\{1,2\}}\geq (1.1,0.5)$\\
 \item 
 $4 x_{a,\{1\}}\geq 0.5 x_{a,\{1\}}$\\
 $0 \geq 0.5 x_{a,\{2\}}$\\
 $4 x_{a,\{1,2\}}\geq 0.5 x_{a,\{1,2\}}$\\
 $0 \geq 0.5 x_{a,\{1,2\}}$\medskip
 
 \noindent
 $x_{b,\{1\}} \geq 0.5 x_{b,\{1\}} + 0.5 x_{c,\{1\}}+ 0.5 x_{d,\{1\}}+ 0.5 x_{e,\{1\}}$\\
 $x_{d,\{2\}} \geq 0.5 x_{b,\{2\}} + 0.5 x_{c,\{2\}}+ 0.5 x_{d,\{2\}}+ 0.5 x_{e,\{2\}}$\\
 $x_{b,\{1,2\}} \geq 0.5 x_{b,\{1,2\}} + 0.5 x_{c,\{1,2\}}+ 0.5 x_{d,\{1,2\}}+ 0.5 x_{e,\{1,2\}}$\\
 $x_{d,\{1,2\}} \geq 0.5 x_{b,\{1,2\}} + 0.5 x_{c,\{1,2\}}+ 0.5 x_{d,\{1,2\}}+ 0.5 x_{e,\{1,2\}}$\\
 \item 
$x_{\ell,\{1\}} + x_{\ell,\{1,2\}} + x_{r,\{1\}} + x_{r,\{1,2\}} + x_{a,\{1\}} + x_{a,\{1,2\}} + x_{b,\{1\}} + x_{b,\{1,2\}} + x_{c,\{1\}} + x_{c,\{1,2\}} + x_{d,\{1\}} + x_{d,\{1,2\}} + x_{e,\{1\}} + x_{e,\{1,2\}} \geq 0.8$\medskip

\noindent
$x_{\ell,\{2\}} + x_{\ell,\{1,2\}} + x_{r,\{2\}} + x_{r,\{1,2\}} + x_{a,\{2\}} + x_{a,\{1,2\}} + x_{b,\{2\}} + x_{b,\{1,2\}} + x_{c,\{2\}} + x_{c,\{1,2\}} + x_{d,\{2\}} + x_{d,\{1,2\}} + x_{e,\{2\}} + x_{e,\{1,2\}} \geq 0.8$\\
\end{enumerate}

\end{document}